\pdfoutput=1
\documentclass[a4paper,11pt]{article}
\usepackage{jheppub}
\makeatletter
\gdef\@fpheader{}
\makeatother
\usepackage{amsthm}
\usepackage{mathtools}
\usepackage{bm}
\usepackage{mathrsfs}

\title{\boldmath Finite-energy hard celestial current algebra from the Banerjee--Mandal--Sahoo dipole Ward identity in QED}

\author[a]{Ruiliang Li}
\affiliation[a]{Tsinghua University, Beijing 100084, China}
\emailAdd{lirl23@mails.tsinghua.edu.cn}

\abstract{We use the Banerjee--Mandal--Sahoo dipole-current Ward identity for the one-loop logarithmic soft-photon theorem as input and determine its finite-energy action on Mellin-difference hard currents.  The commutator with such hard currents has a scheme-independent hard-hard residue that survives every one-particle redefinition.  With the meromorphic continuation stated explicitly below, a two-particle Plancherel transform identifies this residue with an analytic two-particle primary module, and the coefficient map is a hard-current one-cocycle.  The cocycle defines a minimal filtered abelian extension.  It has a canonical two-particle primitive and integrates to an affine action.  For scalar hard legs, the fixed-leg operator agrees coefficient by coefficient with the symmetry-governed long-range logarithmic tower of Choi, Kadhe, and Puhm.  Applied to a tree-level scalar-QED photon-exchange block, the finite-energy analysis determines the logarithmic two-particle coefficient functional from the ordinary hard amplitude and the Banerjee--Mandal--Sahoo ordered-pair soft kernel.  This gives a finite-energy relation between the Banerjee--Mandal--Sahoo dipole-current Ward identity and the exponentiated long-range celestial OPE.}

\newtheorem{theorem}{Theorem}[section]
\newtheorem{lemma}[theorem]{Lemma}
\newtheorem{proposition}[theorem]{Proposition}

\newtheorem{corollary}[theorem]{Corollary}

\newtheorem{assumption}[theorem]{Assumption}

\begin{document}
\maketitle
\flushbottom

\section{Introduction and statement of results}
\label{sec:introduction}

Celestial holography expresses four-dimensional scattering amplitudes as correlation functions on the celestial sphere.  For massless external states, the Mellin transform of the energy produces a conformal dimension and the Lorentz group becomes the global conformal group of the sphere \cite{Cheung:2016iub,Pasterski:2016qvg,Pasterski:2017kqt,Pasterski:2017kqt2,Pasterski:2021rjz,Pasterski:2021rjp}.  Soft theorems give celestial currents whose Ward identities realize infrared symmetries and higher-spin current algebras \cite{Weinberg:1965nx,Low:1958sn,Strominger:2017zoo,Pate:2019lpp,Donnay:2018neh,Nande:2017dba,Guevara:2021abz,Himwich:2021dau,Strominger:2021mtt}.  Loop-level gluon OPEs and logarithmic celestial multiplets have been studied \cite{Bhardwaj:2022gluon,Fiorucci:2023log}.  Low's subleading soft photon theorem already exhibits the finite Mellin shift that appears in the hard-current module \cite{Himwich:2019dug}.

The one-loop logarithmic soft theorem has a different support structure.  After infrared subtraction, its coefficient depends on an ordered pair of charged hard particles.  Banerjee, Mandal and Sahoo reformulated this coefficient as a Ward identity for an antiholomorphic dipole-current doublet of $SL(2,\mathbb R)_R$ \cite{Banerjee:2026soft}.  The high-energy hierarchy, the dipole-current doublet, the Ward identity, the soft-hard OPE, the transformation law, the spinor notation, the local normal-ordered dipole descendant and the higher-spin-current construction used below are taken from their work and rewritten in the conventions of section~\ref{sec:dipole-currents}.  Exponentiated soft exchange also admits a celestial shadow description \cite{Kapec:2021shadow}.  Choi, Kadhe and Puhm derived the pairwise nonlocality from long-range exchange and showed that the symmetry-governed logarithmic sector exponentiates \cite{Choi:2026longrange}.  A note added to that work leaves open the relation between the conformally soft loop-operator OPE and the emergent dipole operator.  The finite-energy version of that question asks whether the Banerjee--Mandal--Sahoo dipole Ward kernel defines a consistent action on Mellin-dependent hard currents.

Section~\ref{sec:dipole-currents} fixes the normalization of the Banerjee--Mandal--Sahoo construction used throughout the paper.  The finite-energy analysis begins with the action of that Ward kernel on Mellin-difference hard currents and with the second-support class produced by this action.

At tree level a celestial soft current acts as a sum of single-leg endomorphisms.  The logarithmic loop kernel instead has a hard-hard Cauchy residue.  This residue is the celestial form of the long-range exchange between two charged hard trajectories and acts on two-particle data in the flat-hologram spectrum.  A one-particle endomorphism is supported at a single insertion, so a single-leg current redefinition cannot remove the loop term.  Resolving the residue into celestial partial waves fixes the associated two-particle coefficient functional.  For a specified hard amplitude, the symmetry-governed logarithmic correction is obtained by a fixed backward-difference operator and the analytic two-particle spectral resolution.  Section~\ref{sec:checks-examples-outlook} applies this rule to a photon-exchange hard block and to a contact block.

Ball formulated hard-current insertions and related their consistency to double residues, celestial locality and Jacobi identities \cite{Ball:2022bgg,Ball:2024currents,Mago:2021wje,Ren:2022sws}.  Liu and Ma constructed a Mellin-difference hard-current algebra in which ordinary soft currents arise from finite-energy generators \cite{Liu:2026hard}.  The problem addressed here starts after the universal long-range kernel is fixed: its commutator with a finite-energy hard current must be represented in the celestial spectrum.  We prove that the commutator has a nonzero class in $F_2/F_1$, identify its analytic two-particle target, derive its hard-current cocycle and affine action, and establish the minimality of that target.  Meromorphic continuation of the two-particle Plancherel transform connects the support class with the multiparticle spectrum of the flat hologram \cite{Kulp:2025flat,Guevara:2024dcy,Ball:2023qtp,Pacifico:2025cpw}.  Proposition~\ref{prop:ckp-coefficient-matching} compares the fixed-leg Mellin shift with the symmetry-governed scalar-QED tower of \cite{Choi:2026longrange}.  An infrared-subtracted hard correlator determines the symmetry-governed depth-two coefficients through a backward Mellin difference and a universal two-particle transform; the cocycle then fixes their recursion under further hard-current commutators.

We obtain the filtered abelian extension
\begin{equation}
\label{eq:intro-schematic-closure}
  \mathfrak g_{\rm ren}^{\log}
  =\mathfrak g_1\oplus_{\nu}t\mathcal M_2,
  \qquad t^2=0,
  \qquad \bmod F_3\, ,
\end{equation}
with exact sequence
\begin{equation}
\label{eq:intro-short-exact-sequence}
  0\longrightarrow \mathcal M_2
  \longrightarrow \mathfrak g_{\rm ren}^{\log}
  \longrightarrow \mathfrak g_1
  \longrightarrow 0
  \qquad \bmod F_3\, .
\end{equation}
Here $\mathfrak g_1=\mathfrak h_{\rm hard}\rtimes\mathfrak s_{\log}$, with $\mathfrak h_{\rm hard}$ the Mellin-difference hard-current algebra and $\mathfrak s_{\log}$ the one-particle soft closure generated by normalized logarithmic insertions, ordinary dipole charges and the local currents produced by same-leg products.  The ideal $\mathcal M_2$ is the soft-accessible analytic two-particle primary module.  ``Normalized'' means that the overall perturbative coefficient multiplying the Ward identity has been stripped off; the displayed kernel is the universal coefficient fixed by the soft theorem.  The extension is determined by
\begin{equation}
\label{eq:intro-one-cocycle-map}
  \mathfrak m_X:\mathfrak h_{\rm hard}\longrightarrow \mathcal M_2,
  \qquad
  \mathfrak m_X(\Phi)=\sigma_2[\,X,H[\Phi] \,] ,
\end{equation}
where $X\in\mathfrak s_{\log}$ and $\sigma_2:F_2\to F_2/F_1$ is the second-support symbol.  If
\begin{equation}
\label{eq:intro-canonical-primitive}
  \kappa_X:=\sigma_2(K_X)\in\mathcal M_2
\end{equation}
is the class of the ordered-pair Ward kernel, then
\begin{equation}
\label{eq:intro-relative-exactness}
  \mathfrak m_X(\Phi)=-\Phi\cdot\kappa_X.
\end{equation}
Thus the cocycle is exact after the two-particle target is admitted, while the pairwise residue criterion excludes any primitive in $F_1$.  The map $\mathfrak m_X$ is the finite-energy response of the long-range dipole kernel under a Mellin displacement of a charged hard leg.  Its nonzero pairwise residue is the obstruction to treating the loop correction as a single-particle current, and the integrated affine action is the corresponding celestial form of the long-range logarithmic orbit.

The cocycle integrates to an affine action of the formal hard-current group.  On a fixed ordered pair with kernel $K_X^{ab}=L_X^{ab}T_a^{-1}$ and scalar Mellin label $f_a(\Delta_a)$,
\begin{equation}
\label{eq:intro-exponentiated-difference}
  \operatorname{ad}_{K_X^{ab}}^{n}f_a
  =(L_X^{ab})^n(\delta_a^-)^n f_a\,T_a^{-n} \, .
\end{equation}
The exponentiated pairwise action is obtained by iterating the backward difference and then projecting to the same analytic two-particle channel.  Proposition~\ref{prop:ckp-coefficient-matching} identifies this series with the long-range exponential of \cite{Choi:2026longrange} after the standard crossing and normalization conventions are aligned.  The dipole kernel is the canonical two-particle primitive of the corresponding affine hard-current action; a one-particle realization would require a primitive with the same hard-pair pole, which does not exist in $F_1$.

\subsection{Logarithmic soft input}
\label{subsec:intro-logarithmic-soft-input}

We work with infrared-subtracted hard functions,
\begin{equation}
\label{eq:intro-ir-factorization}
  \mathcal A_n^{\rm bare}=Z_n^{\rm IR}\,\mathcal H_n^{\rm ren}\, .
\end{equation}
All Ward identities below act on the Mellin transform of \(\mathcal H_n^{\rm ren}\).  For an outgoing positive-helicity photon with energy \(\omega_s\), the logarithmic coefficient has the form
\begin{equation}
\label{eq:intro-log-soft-theorem}
  \mathcal H_{n+1}^{\rm ren}(s^+;1,\ldots,n)\Big|_{\log\omega_s}
  =\log \omega_s\;S^{(1)}_{\log}(s;1,\ldots,n)\,
  \mathcal H_n^{\rm ren}+O(\omega_s^0)\, .
\end{equation}
After projection to the dipole doublet, its normalized local Ward identity contains a one-particle term and an ordered-pair term,
\begin{equation}
\label{eq:intro-dipole-ward-schematic}
  \bigl\langle D_\alpha(\bar u)\prod_{a=1}^n\mathcal O_a\bigr\rangle
  =\sum_a e_a d^a_\alpha(\bar u)\bigl\langle\prod_a\mathcal O_a\bigr\rangle
  +\sum_{a\ne b}K^{ab}_\alpha(\bar u)
  \bigl\langle \mathbb P_{ab}^{(2)}\prod_c\mathcal O_c\bigr\rangle\, .
\end{equation}
The ordered-pair kernel has a Cauchy pole on the diagonal \(\bar z_a=\bar z_b\) of two hard insertions.  One-particle counterterms have no such singularity, and the nonzero residue yields the non-absorption theorem.

\subsection{Structural results}
\label{subsec:intro-three-results}

The dipole-hard OPE takes the following form.  For a hard current \(H[\Phi]\), theorem~\ref{thm:dipole-hard-ope} gives
\begin{equation}
\label{eq:intro-dipole-hard-ope}
  [X,H[\Phi]]
  =H[\mathcal L_X\Phi]+\mathbb M_X[\Phi]
  \qquad \bmod F_3,
\end{equation}
with
\begin{equation}
\label{eq:intro-M-symbol}
  \mathfrak m_X(\Phi)
  :=\sigma_2\bigl(\mathbb M_X[\Phi]\bigr)
  \in F_2/F_1\simeq\mathcal M_2\, .
\end{equation}
The pairwise residue criterion, theorem~\ref{thm:nonabsorption-section4}, proves that this class is nonzero whenever some ordered pair has nonzero charge factor and the corresponding backward Mellin difference of \(\Phi\) does not vanish.

The coefficient map obeys the hard-current cocycle identity
\begin{equation}
\label{eq:intro-hard-cocycle}
  \mathfrak m_X([\Phi,\Psi]_\star)
  =\Phi\cdot\mathfrak m_X(\Psi)
   -\Psi\cdot\mathfrak m_X(\Phi)\, .
\end{equation}
This identity is the algebraic form of the fact that the long-range soft kernel is transported coherently by finite-energy hard-current motions.  Equation~\eqref{eq:intro-relative-exactness} identifies its canonical primitive and shows that the obstruction is relative to the one-particle filtration, not an independent dynamical coupling.  The cocycle integrates to the formal hard-current group, and eq.~\eqref{eq:intro-exponentiated-difference} gives its pairwise finite action in closed form.  Polynomial Mellin labels generate the analytic Mellin component of \(\mathcal M_2\), while Cauchy--Pompeiu completeness gives the angular component.

Two normalized logarithmic soft insertions are compatible with the hard-current action.  If \(K_X^{ab}\) and \(K_Y^{ab}\) are the ordered-pair kernels on a fixed pair, the second-support part of the mixed Jacobi identity is
\begin{equation}
\label{eq:intro-mixed-jacobi-symbol}
  [K_X^{ab},[K_Y^{ab},\Phi_{ab}]]
  -[K_Y^{ab},[K_X^{ab},\Phi_{ab}]]
  -[[K_X^{ab},K_Y^{ab}],\Phi_{ab}]=0\, .
\end{equation}
The equality is the Jacobi identity in the finite-part Cauchy kernel algebra and shows that the nonzero pairwise symbols admit a common internal target.  Theorem~\ref{thm:minimal-closure-section6} proves that the target cannot be removed by a one-particle scheme change and, under its profile and Mellin cyclicity hypotheses, is the full soft-accessible module \(\mathcal M_2\).

\subsection{Relation to previous work and scope}
\label{subsec:intro-relation-scope}

The high-energy hierarchy, dipole-current doublet, Ward identity, soft-hard OPE, transformation law, normal-ordered dipole descendant and higher-spin-current construction are inputs from Banerjee, Mandal and Sahoo \cite{Banerjee:2026soft}; the hard-current setting follows \cite{Ball:2024currents,Liu:2026hard}.  Loop-level hard OPEs, logarithmic celestial multiplets, and conformally soft one-loop currents are developed in \cite{Bhardwaj:2022gluon,Fiorucci:2023log,Bhardwaj:2024qrf}, while the spacetime multipole interpretation of the classical logarithmic theorem is developed in \cite{Compere:2025multipole}.  Choi, Kadhe and Puhm compute conformally soft loop-operator OPEs and derive their long-range exponential \cite{Choi:2026longrange}.  A note added to that work leaves open the relation between this description and the emergent dipole operator of \cite{Banerjee:2026soft}.  The hard-current cocycle theorem, propositions~\ref{prop:canonical-primitive-relative-splitting} and \ref{prop:integrated-hard-current-cocycle}, and proposition~\ref{prop:pairwise-exponentiation} give the connection in the finite-energy hard-current sector.  The finite-energy construction consists of the non-absorbable second-support class, its hard-current cocycle, its finite affine action, and the generation of the analytic two-particle target.  The extension is a statement about the consistency of the filtered celestial OPE, in the sense distinguished from a spacetime symmetry in \cite{Ball:2024currents}.  No additional bulk symmetry is asserted.  One-loop failures of holomorphic chiral-OPE associativity in self-dual gauge theory and gravity instead arise from twistor-space anomalies \cite{Costello:2022wso,Bittleston:2022jeq}.

The residue criterion, hard-current cocycle, canonical primitive, affine action, and scalar-leg matching follow directly from the Ward kernel and the Mellin-difference module.  The global analytic identification $F_2/F_1\simeq\mathcal M_2$, the generation theorem, and the full minimality statement use assumption~\ref{ass:appC-meromorphic-continuation}.  Their stripwise versions are unconditional on each compact pole-free spectral strip covered by proposition~\ref{prop:appC-stripwise-plancherel}.

The analysis is restricted to the first logarithmic photon kernel of an infrared-subtracted abelian QED-like theory and to the quotient through \(F_2/F_1\); the two gradings are fixed in eqs.~\eqref{eq:loop-support-bigrading} and \eqref{eq:formal-support-parameter}.  The momentum-space theorem is regulated in the hierarchy \(\omega_s\ll m\ll E\), and the celestial formulas retain its leading high-energy term.  Nonabelian gauge theory and gravity require additional soft data and are not assumed here \cite{Magnea:2025nonabelian}.

\section{Renormalized celestial hard data}
\label{sec:renormalized-hard-data}

The current algebra acts on infrared-subtracted finite-energy celestial data.  The Mellin transform converts each hard energy into a conformal weight and converts logarithms of the soft energy into higher-order poles in the soft conformal dimension.

We use mostly-plus signature and write all scattering data in an all-outgoing convention.  An external leg has an orientation sign \(\eta_i=+1\) for an outgoing particle and \(\eta_i=-1\) for an incoming particle.  Its physical electric charge is denoted by \(Q_i\), while the charge entering Ward identities in the all-outgoing convention is
\begin{equation}
  e_i = \eta_i Q_i \,.
\end{equation}
The distinction will be useful when comparing celestial Ward identities with the usual momentum-space soft factors.  The regulator and the subtraction scale are denoted by \(\epsilon_{\rm IR}\) and \(\mu\), respectively.  The logarithmic soft theorem is derived with a small charged-particle mass in the hierarchy \(\omega_s\ll m\ll E\).  We retain its leading high-energy coefficient and then use the massless celestial conformal basis; power corrections in \(m/E\) are outside the analysis.  We reserve \(\epsilon_i\) neither for dimensional regularization nor for particle orientation, in order to avoid a clash of notation.

\medskip
\noindent\textbf{Main notation.}
\begin{center}
\small
\renewcommand{\arraystretch}{1.15}
\begin{tabular}{@{}p{0.27\textwidth}p{0.65\textwidth}@{}}
\hline
Symbol & Meaning \\
\hline
\(H[\Phi]\), \(\mathfrak h_{\rm hard}\) & Mellin-difference hard current and its Lie algebra \\
\(X\), \(K_X\), \(\mathfrak s_{\log}\) & logarithmic soft insertion, its ordered-pair kernel, and the one-particle soft closure \\
\(F_1\subset F_2\), \(\sigma_2\) & one-particle and one-pair support layers, with \(\sigma_2:F_2\to F_2/F_1\) \\
\(\mathbb M_X[\Phi]\), \(\mathfrak m_X(\Phi)\) & a depth-two representative and its class \(\sigma_2\mathbb M_X[\Phi]\) \\
\(\kappa_X\), \(\rho_2\) & the class \(\sigma_2(K_X)\) and the diagonal hard-current action on the second layer \\
\(\mathcal M_2\) & analytic two-particle module, including crossed-pole residue channels and the shadow quotient \\
\(T_a^{-1}\), \(\delta_a^-\), \(\nabla_a^-\) & Mellin shift, scalar backward difference, and \((\delta_a^-\Phi_a)T_a^{-1}\) \\
\(L\), \(t\) & physical loop degree and formal support-extension degree \\
\hline
\end{tabular}
\end{center}
\medskip

\subsection{Null directions and the celestial transform}
\label{subsec:kinematics-celestial-transform}

A point \((z,\bar z)\) on the celestial sphere determines the future-directed null vector
\begin{equation}
  q^\mu(z,\bar z)
  = \bigl(1+z\bar z,\, z+\bar z,\, -i(z-\bar z),\, 1-z\bar z\bigr) \,.
  \label{eq:q-vector}
\end{equation}
With our metric conventions
\begin{equation}
  q(z,\bar z)^2 = 0 \, ,
  \qquad
  q(z,\bar z)\cdot q(w,\bar w) = -2 |z-w|^2 \, .
  \label{eq:q-inner-products}
\end{equation}
A massless momentum is parametrized as
\begin{equation}
  p_i^\mu = \eta_i\,\omega_i\,q^\mu(z_i,\bar z_i) \, ,
  \qquad
  \omega_i>0 \, .
  \label{eq:p-parametrization}
\end{equation}
The spinor-helicity variables compatible with \eqref{eq:p-parametrization} may be chosen as
\begin{equation}
  \lambda_{i\alpha}=\sqrt{2\omega_i}\binom{1}{z_i}\, ,
  \qquad
  \tilde\lambda_{i\dot\alpha}=\eta_i\sqrt{2\omega_i}\binom{1}{\bar z_i}\, ,
  \label{eq:spinors}
\end{equation}
so that \(p_{i\alpha\dot\alpha}=\lambda_{i\alpha}\tilde\lambda_{i\dot\alpha}\).  The sign \(\eta_i\) is absorbed in \(\tilde\lambda_i\), which keeps \(\omega_i\) positive for both incoming and outgoing legs.

The Lorentz group \(SL(2,\mathbb C)\) acts on the celestial sphere by Mobius transformations.  If
\begin{equation}
  g = \begin{pmatrix} a & b \\ c & d \end{pmatrix}\in SL(2,\mathbb C)\, ,
  \qquad
  z' = \frac{az+b}{cz+d}\, ,
  \label{eq:mobius}
\end{equation}
then the null vector satisfies
\begin{equation}
  \Lambda(g)^\mu{}_{\nu} q^\nu(z,
  \bar z)= |cz+d|^2 q^\mu(z',\bar z') \, .
  \label{eq:q-lorentz-transform}
\end{equation}
The energy variable transforms as \(\omega'=|cz+d|^{-2}\omega\) if the momentum is held fixed, and as \(\omega'=|cz+d|^{2}\omega\) if the transformed momentum is written again in the form \eqref{eq:p-parametrization}.  The two descriptions are equivalent; in the Mellin transform it is the Jacobian of this rescaling that produces the conformal weights.

For a massless external field of helicity \(J\) and electric charge \(Q\), the celestial operator is the Mellin transform of the corresponding plane-wave creation or annihilation operator,
\begin{equation}
  \mathcal O^{Q,\eta}_{\Delta,J}(z,\bar z)
  =\int_0^\infty d\omega\,\omega^{\Delta-1}
  a^{Q,\eta}_{J}(\omega,z,\bar z) \, .
  \label{eq:celestial-primary-definition}
\end{equation}
For charged scalars \(J=0\).  For photons \(Q=0\) and \(J=\pm1\).  The principal series line is
\begin{equation}
  \Delta = 1+i\lambda\, ,
  \qquad
  \lambda\in\mathbb R \, ,
  \label{eq:principal-series}
\end{equation}
and the two-dimensional weights are
\begin{equation}
  h=\frac{\Delta+J}{2}\, ,
  \qquad
  \bar h=\frac{\Delta-J}{2}\, .
  \label{eq:weights}
\end{equation}
These operators form the standard massless conformal basis for flat-space scattering amplitudes \cite{Pasterski:2017kqt,Pasterski:2016qvg,Pasterski:2017kqt2,Pasterski:2021rjz}.

\begin{lemma}
\label{lem:lorentz-primary}
Let \(\mathcal O^{Q,\eta}_{\Delta,J}\) be defined by \eqref{eq:celestial-primary-definition}.  Under \(g\in SL(2,\mathbb C)\) it transforms as a two-dimensional spin-\(J\) primary,
\begin{equation}
  \mathcal O^{Q,\eta}_{\Delta,J}(z,\bar z)
  \longmapsto
  (cz+d)^{\Delta+J}(\bar c\bar z+\bar d)^{\Delta-J}
  \mathcal O^{Q,\eta}_{\Delta,J}(z',\bar z') \, .
  \label{eq:primary-transform}
\end{equation}
Correlation functions of the Mellin-transformed hard data transform covariantly with weights \((h_i,\bar h_i)\) at each insertion.
\end{lemma}

\begin{proof}
Equation \eqref{eq:q-lorentz-transform} implies that a Lorentz transformation rescales the energy variable of a plane-wave mode and rotates the helicity frame.  For helicity \(J\), the little-group factor is
\begin{equation}
  \left(\frac{cz+d}{\bar c\bar z+\bar d}\right)^J \, .
\end{equation}
We use the passive operator convention, so the mode on the right-hand side is evaluated at \(\Lambda^{-1}p\).  In this convention \(\omega'=|cz+d|^{-2}\omega\), and the change of variables \(\omega=|cz+d|^{2}\omega'\) contributes \(|cz+d|^{2\Delta}\) to the Mellin integral.  Combining it with the helicity factor gives \eqref{eq:primary-transform}; applying the same transformation to every leg proves covariance of the correlator.
\end{proof}

The one-particle celestial space is the direct integral
\begin{equation}
  \mathcal V_1
  =\bigoplus_{Q,J,\eta}\int_{\mathbb R}^{\oplus} d\lambda\,
  \mathbb C\,\mathcal O^{Q,\eta}_{1+i\lambda,J}(z,\bar z) \, ,
  \label{eq:one-particle-space}
\end{equation}
understood as a space of operator-valued distributions in \((\lambda,z,\bar z)\).  Endomorphism kernels on \(\mathcal V_1\) and its tensor powers are filtered below by their hard-leg support.

Denote the one-particle diagonal in the kernel variables by
\begin{equation}
  \Delta_1^{\rm cel}
  =\{(x;x')\in (\mathbb C_\Delta\times\mathbb{CP}^1)^{2}:\, \Delta=\Delta',\ z=z',\ \bar z=\bar z',\ J=J',\ Q=Q',\eta=\eta'\} \, .
  \label{eq:one-particle-diagonal}
\end{equation}
A local one-particle endomorphism has Schwartz kernel supported on this diagonal, up to Mellin shifts in the \(\Delta\)-coordinate.  Pairwise kernels will have singular support on the corresponding diagonals in two distinct celestial variables.

\subsection{Infrared-subtracted hard functions}
\label{subsec:ir-renormalized-hard-functions}

The plane-wave amplitude of a theory with massless photons is not the object on which a celestial current algebra should act.  It contains universal infrared singularities associated with virtual soft photons and, for strictly massless charged external particles, collinear singularities as well.  The algebraic object relevant for finite-energy dynamics is the infrared-subtracted hard function.  The same separation is used in the standard factorization of infrared singularities in gauge theory \cite{Weinberg:1965nx,Catani:1998bh}.  In the abelian theory considered here the factor is scalar in charge space, although it is still a nontrivial function of the pairwise invariants.

We write the dimensionally regulated all-outgoing amplitude as
\begin{equation}
  \mathcal A_n^{\rm bare}
  (\{\eta_i\omega_i q_i,J_i,Q_i\};\epsilon_{\rm IR},\mu)
  = Z_n^{\rm IR}(\{\eta_i\omega_i q_i,Q_i\};\epsilon_{\rm IR},\mu)
  \mathcal H_n^{\rm ren}(\{\eta_i\omega_i q_i,J_i,Q_i\};\mu) \, .
  \label{eq:ir-factorization}
\end{equation}
The renormalized hard function is defined by the finite part
\begin{equation}
  \mathcal H_n^{\rm ren}(\mu)
  =\operatorname{FP}_{\epsilon_{\rm IR}\to0}
  \Bigl[(Z_n^{\rm IR})^{-1}\mathcal A_n^{\rm bare}\Bigr] \, .
  \label{eq:hard-function-finite-part}
\end{equation}
This convention fixes the separation between the universal infrared factor and the finite hard data.  The form of such a factorization is standard in the analysis of infrared singularities and soft anomalous dimensions \cite{Catani:1998bh,Becher:2009cu,Gardi:2009qi}.  A change of subtraction scheme acts by a finite, Lorentz-covariant multiplicative factor on \(\mathcal H_n^{\rm ren}\).  Later logarithmic currents will therefore have a controlled scheme dependence, rather than an ambiguity in their definition.

The celestial hard distribution is obtained by Mellin-transforming \(\mathcal H_n^{\rm ren}\).  When the chosen amplitude convention includes the momentum-conservation delta function, it is included in this transform,
\begin{equation}
  \widetilde{\mathcal H}_n
  \bigl(\{\Delta_i,J_i,Q_i,\eta_i,z_i,\bar z_i\};\mu\bigr)
  =\int_0^\infty \prod_{i=1}^n
  \bigl(d\omega_i\,\omega_i^{\Delta_i-1}\bigr)
  \mathcal H_n^{\rm ren}
  (\{\eta_i\omega_i q_i,J_i,Q_i\};\mu) \, .
  \label{eq:celestial-hard-transform}
\end{equation}
The integral is interpreted distributionally: one first pairs the integrand with a test function in the variables \((\lambda_i,z_i,\bar z_i)\) and then performs the energy integrals.  Momentum conservation fixes four combinations of the \(\omega_i\), conformal covariance fixes the transformation under \(SL(2,\mathbb C)\), and loop logarithms introduce meromorphic dependence on the \(\Delta_i\).  All three features are retained by \eqref{eq:celestial-hard-transform}.

The scale dependence of the hard function is finite.  In a mass-independent subtraction scheme it takes the form
\begin{equation}
  \mu\frac{d}{d\mu}\mathcal H_n^{\rm ren}(\mu)
  =\Gamma_n(\{p_i,Q_i\};\mu)\mathcal H_n^{\rm ren}(\mu) \, ,
  \label{eq:hard-rg}
\end{equation}
where \(\Gamma_n\) is the finite infrared anomalous dimension obtained from \(Z_n^{\rm IR}\).  In an abelian theory its nontrivial dependence is pairwise, through the invariants
\begin{equation}
  s_{ij}= -2 p_i\cdot p_j - i0
  = 4\eta_i\eta_j\omega_i\omega_j |z_i-z_j|^2 - i0 \, .
  \label{eq:pairwise-invariant}
\end{equation}
The pairwise anomalous dimension is the momentum-space precursor of the ordered-pair kernel isolated by the logarithmic soft theorem.

The hard correlator associated to \eqref{eq:celestial-hard-transform} is denoted by
\begin{equation}
  \left\langle
  \prod_{i=1}^n \mathcal O^{Q_i,\eta_i}_{\Delta_i,J_i}(z_i,\bar z_i)
  \right\rangle_{\mathcal H,\mu}
  =\widetilde{\mathcal H}_n
  \bigl(\{\Delta_i,J_i,Q_i,\eta_i,z_i,\bar z_i\};\mu\bigr) \, .
  \label{eq:hard-correlator-notation}
\end{equation}

A running hard amplitude is the photon-exchange block for two distinguishable charged scalars.  Let legs $(1,4)$ carry species charge $Q_A$ and legs $(2,3)$ carry $Q_B$, with $Q_1=-Q_4=Q_A$ and $Q_2=-Q_3=Q_B$.  After removing the overall factor $e^2$ and the phase convention of the Feynman amplitude, the all-outgoing hard block is
\begin{equation}
  \widehat{\mathcal H}_{4}^{\rm ex}
  =\frac{J_A\cdot J_B}{(p_1+p_4)^2+i0},
  \qquad
  J_A^\mu=Q_A(p_1-p_4)^\mu,
  \qquad
  J_B^\mu=Q_B(p_2-p_3)^\mu .
  \label{eq:running-exchange-hard-block}
\end{equation}
This is the scalar-QED photon-exchange graph in the all-outgoing convention \cite{Bern:2021scalarQED}.  Its celestial transform, including the momentum-conservation distribution, is denoted by $C_4^{\rm ex}$.

\subsection{The Mellin-difference hard-current module}
\label{subsec:hard-current-module}

Finite-energy hard currents act by changing Mellin weights and by multiplying by profiles of the celestial coordinates and discrete quantum numbers.  Let \(\mathscr P(X)\) be the unital profile algebra generated, in each affine patch of \(X=\mathbb{CP}^1\), by smooth compactly supported functions and by meromorphic current profiles whose poles are disjoint from the hard insertions.  Its restriction to an affine patch is denoted by \(\mathscr P(U)\).  On the ordered configuration space of distinct hard points, profiles supported in disjoint patches act as independent leg multipliers; successive diagonal hard-current actions therefore give the completed tensor products \(\mathscr P(U_a)\widehat\otimes\mathscr P(U_b)\) used in the generation theorem.  Symmetrization of identical external states is imposed after the ordered-pair calculation.  The corresponding current labels form a difference-differential algebra in the conformal dimension.  Let \(T_a\) denote the shift operator
\begin{equation}
  T_a \mathcal O^{Q,\eta}_{\Delta,J}(z,\bar z)
  =\mathcal O^{Q,\eta}_{\Delta+a,J}(z,\bar z)\, ,
  \qquad
  T_a=e^{a\partial_\Delta}\, .
  \label{eq:mellin-shift}
\end{equation}
For a coefficient \(f(\Delta,z,\bar z,J,Q,\eta)\) that is holomorphic in the chosen Mellin strip and belongs to \(\mathscr P(X)\) in the celestial variables, the basic relation is
\begin{equation}
  T_a f(\Delta,z,\bar z,J,Q,\eta)
  = f(\Delta+a,z,\bar z,J,Q,\eta)T_a \, .
  \label{eq:shift-commutation}
\end{equation}
We define \(\mathscr D_{\rm Mell}\) to be the algebra of finite sums of operators
\begin{equation}
  \Phi
  =\sum_{\alpha} f_\alpha(\Delta,z,\bar z,J,Q,\eta)\,
  \partial_\Delta^{k_\alpha}T_{a_\alpha} \, .
  \label{eq:mellin-difference-operator}
\end{equation}
The product induced by composition will be denoted by \(\star\).  In the case without derivatives it is explicitly
\begin{equation}
  \bigl(f(\Delta)T_a\bigr)\star\bigl(g(\Delta)T_b\bigr)
  =f(\Delta)g(\Delta+a)T_{a+b} \, .
  \label{eq:star-product}
\end{equation}
The derivative terms obey the same rule together with the ordinary Leibniz rule.  This algebra is noncommutative because multiplication by \(\Delta\)-dependent functions does not commute with Mellin shifts.

A hard current labelled by \(\Phi\in\mathscr D_{\rm Mell}\) is denoted by \(H[\Phi]\).  Its action on a one-particle celestial operator is
\begin{equation}
  H[\Phi]\cdot \mathcal O^{Q,\eta}_{\Delta,J}(z,\bar z)
  =\Phi(\Delta,z,\bar z,J,Q,\eta)
  \mathcal O^{Q,\eta}_{\Delta,J}(z,\bar z) \, .
  \label{eq:hard-current-action-on-operator}
\end{equation}
For example, the diagonal electric hard current with profile \(\varphi\) corresponds to
\begin{equation}
  \Phi_{\rm el}=\eta Q\,\varphi(\Delta,z,\bar z,\partial_\Delta) \, .
  \label{eq:electric-hard-current-label}
\end{equation}
A general hard current may also contain helicity projectors and Mellin shifts.

On hard correlators the current insertion is the sum over external legs,
\begin{equation}
\begin{aligned}
  &\left\langle H[\Phi]
  \prod_{i=1}^n \mathcal O^{Q_i,\eta_i}_{\Delta_i,J_i}(z_i,\bar z_i)
  \right\rangle_{\mathcal H,\mu}
  \\
  &\hspace{1.0cm}
  =\sum_{i=1}^n
  \Phi_i
  \left\langle
  \prod_{i=1}^n \mathcal O^{Q_i,\eta_i}_{\Delta_i,J_i}(z_i,\bar z_i)
  \right\rangle_{\mathcal H,\mu} \, ,
\end{aligned}
\label{eq:hard-current-on-correlator}
\end{equation}
where \(\Phi_i\) acts on the variables of the \(i\)-th insertion.  The formula is a Ward identity in the hard theory.  It contains the ordinary global charge Ward identity as the special case \(\Phi_i=e_i\), for which charge conservation gives \(\sum_i e_i=0\) on nonvanishing amplitudes.

\begin{lemma}[Hard-current representation]
\label{lem:hard-current-representation}
The assignment \(\Phi\mapsto H[\Phi]\) gives a representation of the Lie algebra
\begin{equation}
  \mathfrak h_{\rm hard}
  =\bigl(\mathscr D_{\rm Mell},[\,,\,]_{\star}\bigr)\, ,
  \qquad
  [\Phi,\Psi]_{\star}=\Phi\star\Psi-\Psi\star\Phi \, ,
  \label{eq:hard-current-lie-algebra}
\end{equation}
up to the Ward-null ideal of operators that annihilate all hard correlators.  In particular,
\begin{equation}
  [H[\Phi],H[\Psi]]
  = H[[\Phi,\Psi]_{\star}]
  \label{eq:hard-current-bracket}
\end{equation}
as operators on \(\mathcal V_1\) and on the distributional hard correlators \eqref{eq:hard-correlator-notation}.
\end{lemma}

\begin{proof}
On a one-particle operator the statement is the associativity of composition in \(\mathscr D_{\rm Mell}\):
\begin{equation}
  H[\Phi]H[\Psi]\cdot\mathcal O
  = (\Phi\star\Psi)\mathcal O\, ,
  \qquad
  H[\Psi]H[\Phi]\cdot\mathcal O
  = (\Psi\star\Phi)\mathcal O\, .
\end{equation}
Taking the difference gives \eqref{eq:hard-current-bracket}.  For an \(n\)-point hard correlator, \(H[\Phi]\) acts by the sum of the one-particle actions \(\sum_i\Phi_i\).  Operators acting on distinct insertions commute.  The commutator therefore reduces to the sum of the one-particle commutators,
\begin{equation}
  \left[\sum_i\Phi_i,\sum_j\Psi_j\right]
  =\sum_i[\Phi_i,\Psi_i]_{\star}\, .
\end{equation}
It gives the action of \(H[[\Phi,\Psi]_{\star}]\).  If a label produces zero on all hard correlators because of charge conservation, momentum conservation, or the equations of motion, it lies in the Ward-null ideal.  Quotienting by this ideal gives a faithful action.
\end{proof}

Only this Mellin-difference module action is required below.  It does not assume a microscopic construction of all hard currents.  Ball's hard-current insertions provide the CFT consistency language, while the hard-current algebra of Liu and Ma gives a broader framework in which soft current algebras arise from finite-energy hard currents \cite{Ball:2024currents,Liu:2026hard}.  Only the Mellin-difference part of that framework enters here.

\paragraph{Physical loop degree and support degree.}
The normalized logarithmic Ward kernel is used throughout without its overall perturbative coefficient.  Physical loop degree, denoted by \(L\), and support depth are independent gradings:
\begin{equation}
\label{eq:loop-support-bigrading}
\begin{array}{c|c|c}
\text{object} & L & \text{support depth} \\
\hline
H[\Phi] & 0 & 1 \\
D_\alpha & 0 & 1 \\
\mathsf S^0_\mu & 1 & 2\ \text{on hard correlators} \\
\mathbb M_X[\Phi] & 1 & 2 \\
\mathsf S_X^0\mathsf S_Y^0\ \text{bilinear kernel} & 2 & \leq 2\ \bmod F_3
\end{array}
\end{equation}
The dipole current \(D_\alpha\) is normalized by its Ward identity and carries no separate perturbative coefficient.  The logarithmic insertion \(\mathsf S^0_\mu\) is the coefficient of the one-loop soft theorem.  Two logarithmic insertions therefore belong to physical loop degree two, even though their universal bilinear contribution is determined by composing one-loop Ward kernels.  This bilinear term is not the independent two-loop hard amplitude; it is the contribution fixed by the normalized one-loop Ward data.  In the deformation-theoretic discussion a formal parameter \(t\) records the first extension by \(F_2/F_1\),
\begin{equation}
\label{eq:formal-support-parameter}
  m_t=m_0+t\,m_1+O(t^2)\, .
\end{equation}
The parameter \(t\) is unrelated to \(L\).  A term may therefore have physical degree \(L=2\) and formal degree one: this occurs when two normalized one-loop insertions compose on a single connected hard pair and still define one class in \(F_2/F_1\).  Jacobi and Hochschild compatibility refer to this formal support extension and constrain the bilinear part fixed by the one-loop Ward kernels.

\subsection{Logarithmic Mellin distributions}
\label{subsec:logarithmic-mellin-distributions}

The celestial transform converts powers of \(\log \omega_s\) into higher-order poles at the conformally soft value of \(\Delta_s\).  Tree-level conformally soft currents are organized by residues at special conformal weights \cite{Pate:2019lpp,Pasterski:2021rjz}; one-loop soft-current OPEs require logarithmic partners \cite{Bhardwaj:2024qrf}, and the all-loop soft-photon theorem identifies the one-loop \(\log\omega_s\) coefficient with antiholomorphic dipole currents \cite{Banerjee:2026soft}.  The Mellin distribution formulas below depend only on this logarithmic energy behavior.

The formulas are first written at the leading conformally soft point.  More generally, a term \(\omega^r\log(\omega/\mu)\) gives the same Laurent expansion after replacing \(s\) by \(\Delta+r\).  Thus the leading Weinberg logarithmic distribution is centered at \(\Delta=1\), while the one-loop \(O(\log\omega_s)\) dipole theorem used in section~\ref{sec:dipole-currents} is centered at \(\Delta_s=0\).

Let \(f\in C_c^\infty([0,\infty))\) be equal to a smooth test function near \(\omega=0\), and set \(s=\Delta-1\).  Its Mellin transform near \(s=0\) has the meromorphic expansion
\begin{equation}
  M_f(s)=\int_0^\infty d\omega\,\omega^{s-1}f(\omega)
  =\frac{f(0)}{s}+M_f^{\rm reg}(0)+O(s) \, .
  \label{eq:mellin-simple-pole}
\end{equation}
The logarithmic transform is
\begin{equation}
  L_{\mu,f}(s)
  =\int_0^\infty d\omega\,\omega^{s-1}\log\frac{\omega}{\mu}\,f(\omega)
  =\partial_s M_f(s)-\log\mu\,M_f(s) \, .
  \label{eq:mellin-log-transform}
\end{equation}
It follows that
\begin{equation}
  L_{\mu,f}(s)
  =-\frac{f(0)}{s^2}-\frac{f(0)\log\mu}{s}+O(1) \, .
  \label{eq:log-mellin-poles}
\end{equation}
The sign of the double pole is a convention inherited from differentiating \(s^{-1}\).  One may reverse it by defining the logarithmic residue with an additional minus sign.  The invariant content is the existence of a second-order pole and the scale-dependence of the simple-pole part.

It is often more useful to express \eqref{eq:log-mellin-poles} in distributional language.  The renormalized distributions
\begin{equation}
  \left[\frac{1}{\omega}\right]_+\, ,
  \qquad
  \left[\frac{\log(\omega/\mu)}{\omega}\right]_+
  \label{eq:plus-distributions}
\end{equation}
are defined by subtracting the value of the test function at the origin, in the standard finite-part sense of distribution theory \cite{Gelfand:1964gf}.  Their precise definition depends on a choice of smooth cutoff, but two choices differ by a multiple of \(\delta(\omega)\), hence by a local counterterm in the conformally soft operator.  The scale dependence is unambiguous,
\begin{equation}
  \left[\frac{\log(\omega/\mu')}{\omega}\right]_+
  =\left[\frac{\log(\omega/\mu)}{\omega}\right]_+
  -\log\frac{\mu'}{\mu}\left[\frac{1}{\omega}\right]_+ \, .
  \label{eq:scale-dependence-plus-distribution}
\end{equation}
This identity will later become the mixing of a dipole current with the ordinary conformally soft current under a change of renormalization scale.

We use the following notation for residues at the conformally soft point.  If \(F(\Delta)\) has an expansion
\begin{equation}
  F(\Delta)
  =\frac{A_{-2}}{(\Delta-1)^2}
  +\frac{A_{-1}}{\Delta-1}
  +A_0+O(\Delta-1) \, ,
  \label{eq:laurent-expansion}
\end{equation}
then
\begin{equation}
  \operatorname{DRes}_{\Delta=1}F=A_{-2}\, ,
  \qquad
  \operatorname{Res}_{\Delta=1}F=A_{-1}\, ,
  \qquad
  \operatorname{FP}_{\Delta=1}F=A_0 \, .
  \label{eq:residue-notation}
\end{equation}
The ordinary leading conformally soft photon current is associated with \(\operatorname{Res}_{\Delta=1}\).  A logarithmic factor of homogeneity \(r\) contributes through \(\operatorname{DRes}_{\Delta=-r}\) and through a scheme-dependent simple-pole term.  The one-loop dipole current corresponds to \(r=0\), whereas the displayed leading-soft distribution corresponds to \(r=-1\).  This produces the logarithmic extension
\begin{equation}
  0\longrightarrow \mathcal S
  \longrightarrow \mathcal S_{\log}
  \longrightarrow \mathcal S
  \longrightarrow 0 \, ,
  \label{eq:log-current-extension}
\end{equation}
where \(\mathcal S\) is the ordinary conformally soft current module.  Multiplication by \(\log\omega\) in momentum space is represented by \(\partial_\Delta\) in Mellin space,
\begin{equation}
  \partial_\Delta\mathcal O^{Q,\eta}_{\Delta,J}(z,\bar z)
  =\int_0^\infty d\omega\,\omega^{\Delta-1}\log\omega\,
  a^{Q,\eta}_{J}(\omega,z,\bar z) \, .
  \label{eq:delta-derivative-log}
\end{equation}
The logarithmic module is already present in the hard Mellin-difference algebra.  The one-loop input adds pairwise kernels that are not endomorphisms of \(\mathcal V_1\).

Let \(S\) denote the ordinary conformally soft photon current obtained from \([1/\omega]_+\), and let \(D^{(\mu)}\) denote the logarithmic current obtained from \([\log(\omega/\mu)/\omega]_+\).  Their scale mixing follows from \eqref{eq:scale-dependence-plus-distribution}:
\begin{equation}
  D^{(\mu')} = D^{(\mu)}-
  \log\frac{\mu'}{\mu}\,S \, .
  \label{eq:dipole-scale-mixing}
\end{equation}
The renormalized algebra is therefore an algebra over the scale \(\mu\), with changes of \(\mu\) implemented by triangular transformations of the logarithmic current module.  It is the celestial analogue of the finite scheme dependence of the hard function in \eqref{eq:hard-function-finite-part}.

The single-particle hard algebra is built from local Mellin-difference endomorphisms of \(\mathcal V_1\).  Logarithmic Mellin singularities extend this algebra by Jordan partners of conformally soft currents.  Multiparticle data enter only through the pairwise support of the one-loop soft factor, which is represented by two-particle celestial primaries in the flat-hologram spectrum \cite{Kulp:2025flat}.

\section{The Banerjee--Mandal--Sahoo dipole current and the one-loop soft theorem}
\label{sec:dipole-currents}

This section reviews the dipole-current construction of Banerjee, Mandal and Sahoo \cite{Banerjee:2026soft} in the all-outgoing and Mellin conventions used in the rest of the paper.  We keep their high-energy hierarchy, dipole-current doublet, Ward identity, soft-hard OPE, transformation law, spinor notation, local normal-ordered logarithmic insertion and higher-spin-current construction in order to fix the normalization of the ordered-pair kernel that acts on hard currents in sections~\ref{sec:dipole-hard-ope}--\ref{sec:checks-examples-outlook}.

The Weinberg soft factor is a sum of one-leg operators on celestial insertions.  The one-loop logarithmic coefficient instead depends on an ordered pair: the emitting particle carries the holomorphic pole and a second charged particle enters through an antiholomorphic rational kernel.  This dependence has a Cauchy residue on a hard-pair diagonal and hence lives in the next support layer.

Before taking a hard-current OPE, consider the renormalized hard correlators of section~\ref{sec:renormalized-hard-data}.  The long-range photon has already been separated into the universal infrared factor, and the soft theorem is an identity for the finite hard functions.  Following the high-energy hierarchy used by Banerjee, Mandal and Sahoo, eq.~(2.3) of \cite{Banerjee:2026soft}, we take
\begin{equation}
\label{eq:high-energy-hierarchy}
  \omega_s\ll m\ll E \, ,
\end{equation}
where \(\omega_s\) is the soft photon energy, \(m\) is the common mass used to regulate collinear regions of the charged external particles, and \(E\) is a typical hard energy.  At leading order in \(m/E\) the hard momenta are null and may be parametrized as in eq.~\eqref{eq:p-parametrization}.  The mass remains only as the physical regulator that makes the logarithmic theorem well defined before the high-energy limit is taken.  In this regime the one-loop logarithmic soft factor of \cite{Banerjee:2026soft} is the soft-theorem input for the Banerjee--Mandal--Sahoo dipole-current Ward identity.  We keep their Ward-identity normalization because it is the normalization relevant for the later hard-algebra calculation.

\subsection{The logarithmic soft factor in celestial variables}
\label{subsec:log-soft-factor-celestial}

Let the outgoing positive-helicity soft photon have momentum
\begin{equation}
\label{eq:soft-photon-momentum}
  k^\mu=\omega_s q^\mu(w,\bar w) \, ,
  \qquad \omega_s>0 \, ,
\end{equation}
and choose the polarization vector
\begin{equation}
\label{eq:positive-helicity-polarization}
  \varepsilon_+^\mu(w,\bar w)
  =\frac{1}{\sqrt2}\bigl(\bar w,1,-i,-\bar w\bigr) \, .
\end{equation}
With the normalization used below, the precise little-group convention is secondary; the choice in eq.~\eqref{eq:positive-helicity-polarization} fixes the relative factors between the holomorphic soft pole and the antiholomorphic dipole kernel.  We write
\begin{equation}
\label{eq:short-hand-correlator-section3}
  C_n(1,\ldots,n)
  =\left\langle
  \prod_{i=1}^n
  \mathcal O_{\Delta_i,J_i}^{Q_i,\eta_i}(z_i,\bar z_i)
  \right\rangle_{\mathcal H,\mu}
\end{equation}
for the renormalized hard celestial correlator.  The signed all-outgoing charge is
\begin{equation}
\label{eq:signed-charge-section3}
  e_i=\eta_i Q_i \, .
\end{equation}
We use the shift operator on the \(a\)-th Mellin weight,
\begin{equation}
\label{eq:T-minus-one-section3}
  T_a^{-1}C_n(1,\ldots,n)
  =C_n(1,\ldots,(\Delta_a-1,J_a,Q_a,\eta_a,z_a,\bar z_a),\ldots,n) \, .
\end{equation}
The shift by \(-1\) is the celestial image of the inverse hard-energy factor in the logarithmic kernel.

The one-loop logarithmic soft theorem for the finite hard functions may be stated as follows.  Let \(\Gamma_\rho^+(w,\bar w)\) be the Mellin transform of the positive-helicity photon annihilation operator with soft Mellin parameter \(\rho\),
\begin{equation}
\label{eq:photon-soft-family-rho}
  \Gamma_\rho^+(w,\bar w)
  =\int_0^\infty d\omega_s\,\omega_s^{\rho-1}
  a_+(\omega_s,w,\bar w) \, .
\end{equation}
The logarithmic term \(\log(\mu/\omega_s)\) has a double pole at \(\rho=0\), by the same Mellin calculation as in section~\ref{subsec:logarithmic-mellin-distributions} with \(s\) replaced by \(\rho\).  We denote by \(\mathsf S^0_\mu(w,\bar w)\) the renormalized logarithmic soft insertion normalized by the ordered-pair Mellin correlator identity
\begin{equation}
\label{eq:S0-Ward-pairwise}
\begin{aligned}
  &\left\langle
  \mathsf S^0_\mu(w,\bar w)
  \prod_{i=1}^n
  \mathcal O_{\Delta_i,J_i}^{Q_i,\eta_i}(z_i,\bar z_i)
  \right\rangle_{\mathcal H,\mu}
  \\
  &\hspace{0.6cm}
  =\sum_{a=1}^n\sum_{\substack{b=1\\ b\ne a}}^n
  \frac{\eta_a Q_a^2 e_b}{w-z_a}
  \frac{\bar w-\bar z_b}{\bar z_a-\bar z_b}
  \;T_a^{-1}C_n(1,\ldots,n) \, .
\end{aligned}
\end{equation}
Equation~\eqref{eq:S0-Ward-pairwise} is the leading high-energy Mellin form of the ordered-pair part of the one-loop exact \(O(\log\omega_s)\) soft photon theorem.  Its massive origin and the extraction of the logarithmic coefficient go back to the logarithmic soft expansion of Sahoo and Sen \cite{Sahoo:2018lxl}; universality of the logarithmic factors was analyzed in \cite{Krishna:2023fxg}.  Banerjee, Mandal and Sahoo showed that this ordered-pair soft factor can be rewritten as their current-algebra Ward identity in terms of the dipole-current doublet \cite{Banerjee:2026soft}.  Thus eq.~\eqref{eq:S0-Ward-pairwise} is not itself the dipole Ward identity; it is the Mellin-space soft-theorem coefficient whose Banerjee--Mandal--Sahoo Ward-identity form is reviewed in the next subsections.  The hard-current calculation uses this normalization.

In eq.~\eqref{eq:S0-Ward-pairwise}, the holomorphic pole \((w-z_a)^{-1}\) is attached to the emitting leg, whereas the antiholomorphic coefficient depends on a second leg \(b\).  A single sum over \(a\) would define a Mellin-difference endomorphism of \(\mathcal V_1\); the ordered-pair sum has support depth two.

Write the ordered-pair kernel as
\begin{equation}
\label{eq:ordered-pair-kernel-section3}
  K_{w,\bar w}^{ab}
  =\frac{\eta_a Q_a^2 e_b}{w-z_a}
  \frac{\bar w-\bar z_b}{\bar z_a-\bar z_b}
  T_a^{-1} \, ,
  \qquad a\ne b \, .
\end{equation}
Then eq.~\eqref{eq:S0-Ward-pairwise} becomes
\begin{equation}
\label{eq:S0-pairwise-kernel-compact}
  \left\langle
  \mathsf S^0_\mu(w,\bar w)\prod_i\mathcal O_i
  \right\rangle_{\mathcal H,\mu}
  =\sum_{a\ne b}K_{w,\bar w}^{ab}C_n \, .
\end{equation}

\subsection{The antiholomorphic dipole current}
\label{subsec:antiholo-dipole-current}

The pairwise kernel in eq.~\eqref{eq:S0-Ward-pairwise} is the Ward kernel of the Banerjee--Mandal--Sahoo celestial-sphere current for antiholomorphic charge moments.  This operator is not an asymptotic photon field; its Ward identity is defined on renormalized hard correlators by the rational function appearing in the loop theorem.  In the spinor notation of Banerjee, Mandal and Sahoo, let
\begin{equation}
\label{eq:Zbar-spinor-definition}
  \bar Z^\alpha(\bar z)=
  \begin{pmatrix}\bar z\\ 1\end{pmatrix},
  \qquad
  \epsilon_{\alpha\beta}=
  \begin{pmatrix}0&1\\-1&0\end{pmatrix},
  \qquad
  \bar Z_\alpha(\bar z)=\epsilon_{\alpha\beta}\bar Z^\beta(\bar z)
  =\begin{pmatrix}1\\-\bar z\end{pmatrix} \, ,
\end{equation}
with \(\alpha=1,2\).  In the all-outgoing convention used here, the Banerjee--Mandal--Sahoo dipole Ward identity, eq.~(4.3) of \cite{Banerjee:2026soft}, reads
\begin{equation}
\label{eq:dipole-current-Ward-alpha}
  \left\langle
  D_\alpha^{(\mu)}(\bar z)
  \prod_{i=1}^n \mathcal O_i
  \right\rangle_{\mathcal H,\mu}
  =\sum_{i=1}^n e_i
  \frac{\bar Z_{i\alpha}}{\bar z-\bar z_i}
  C_n(1,\ldots,n) \, ,
\end{equation}
where \(\bar Z_{i\alpha}=\bar Z_\alpha(\bar z_i)\).  This is the spinor form of their Eq.~(6.3), equivalently their dipole-current Ward identity Eq.~(4.3), written with the charge and orientation conventions of this paper.  Under a \(\mu\)-rescaling, a finite-part representative may mix with an ordinary soft partner as in eq.~\eqref{eq:dipole-scale-mixing}.  The double-pole residue entering the Ward identity is unchanged, and the Ward kernels are invariant under this triangular change of representative.

For a spinor \(\bar W^\alpha=(\bar W,1)^t\) we set
\begin{equation}
\label{eq:D-W-contracted}
  D^{(\mu)}[\bar W](\bar z)=\bar W^\alpha D_\alpha^{(\mu)}(\bar z) \, .
\end{equation}
Then the contracted form of the Banerjee--Mandal--Sahoo Ward identity is
\begin{equation}
\label{eq:dipole-current-Ward-W}
  \left\langle
  D^{(\mu)}[\bar W](\bar z)
  \prod_{i=1}^n\mathcal O_i
  \right\rangle_{\mathcal H,\mu}
  =\sum_{i=1}^n e_i
  \frac{\bar W-\bar z_i}{\bar z-\bar z_i}
  C_n(1,\ldots,n) \, .
\end{equation}
The numerator is affine in \(\bar W\).  The coefficient of \(\bar W\) is the monopole charge distribution, while the constant term is the antiholomorphic position-weighted charge distribution.  At large \(\bar z\), charge conservation removes the monopole component of a physical correlator, leaving the first antiholomorphic moment as the dipole component.  

The soft-hard OPE follows directly from the Banerjee--Mandal--Sahoo Ward identity, eq.~(4.3) of \cite{Banerjee:2026soft}, and in our conventions reads
\begin{equation}
\label{eq:dipole-hard-primary-OPE-section3}
  D_\alpha^{(\mu)}(\bar z)
  \mathcal O_{\Delta,J}^{Q,\eta}(z_i,\bar z_i)
  \sim
  \frac{e_i\bar Z_{i\alpha}}{\bar z-\bar z_i}
  \mathcal O_{\Delta,J}^{Q,\eta}(z_i,\bar z_i) \, .
\end{equation}
The corresponding charge action is the Banerjee--Mandal--Sahoo charge action, eq.~(4.4) of \cite{Banerjee:2026soft}, written in all-outgoing variables:
\begin{equation}
\label{eq:dipole-charge-action}
  \left[\oint_{\bar z_i}\frac{d\bar z}{2\pi i}
  D^{(\mu)}[\bar W](\bar z),
  \mathcal O_{\Delta,J}^{Q_i,\eta_i}(z_i,\bar z_i)
  \right]
  = e_i(\bar W-\bar z_i)
  \mathcal O_{\Delta,J}^{Q_i,\eta_i}(z_i,\bar z_i) \, .
\end{equation}
The pole is purely antiholomorphic.  This chirality is what permits the logarithmic soft theorem to be written as a holomorphic soft pole multiplying an antiholomorphic dipole composite.

The transformation law and its covariance check are the ones derived in section~V of \cite{Banerjee:2026soft}.  Let
\begin{equation}
\label{eq:SL2R-transform}
  \bar z'=\frac{a\bar z+b}{c\bar z+d} \, ,
  \qquad
  g=\begin{pmatrix}a&b\\c&d\end{pmatrix}\in SL(2,\mathbb R)_R \, .
\end{equation}
The contracted current obeys
\begin{equation}
\label{eq:D-W-transform}
  D^{(\mu)}[\bar W](\bar z)
  \longmapsto
  \frac{c\bar W+d}{c\bar z+d}
  D^{(\mu)}[\bar W'](\bar z') \, ,
  \qquad
  \bar W'=\frac{a\bar W+b}{c\bar W+d} \, .
\end{equation}
Equivalently, in spinor notation
\begin{equation}
\label{eq:D-alpha-transform}
  D_\alpha^{(\mu)}(\bar z)
  \longmapsto
  \frac{1}{c\bar z+d}\,
  g^\beta{}_{\alpha}
  D_\beta^{(\mu)}(\bar z') \, .
\end{equation}
\(D_\alpha^{(\mu)}\) is an antiholomorphic conformal field of weight \(1/2\) tensored with the two-dimensional spin representation.  It is the \(j=1/2\) member of the current tower generated by the higher logarithmic soft photon theorems \cite{Banerjee:2026soft}.

\paragraph{Covariance.}
The Ward identity \eqref{eq:dipole-current-Ward-alpha} is covariant under \(SL(2,\mathbb R)_R\) precisely for the transformation law \eqref{eq:D-alpha-transform}.  The contracted transformation law \eqref{eq:D-W-transform} is the Banerjee--Mandal--Sahoo covariance statement, eq.~(5.1) of \cite{Banerjee:2026soft}.  In our conventions the elementary identity entering the covariance check is
\begin{equation}
\label{eq:kernel-transform-proof}
  \frac{\bar W'-\bar z_i'}{\bar z'-\bar z_i'}
  =\frac{c\bar z+d}{c\bar W+d}
  \frac{\bar W-\bar z_i}{\bar z-\bar z_i} \, .
\end{equation}
The prefactor in \eqref{eq:D-W-transform} therefore leaves the Ward kernel \eqref{eq:dipole-current-Ward-W} invariant.  Conversely, the rational functions \((\bar W-\bar z_i)/(\bar z-\bar z_i)\) separate points and charge assignments, so covariance determines the same current transformation law.

\subsection{The Banerjee--Mandal--Sahoo local dipole Ward identity}
\label{subsec:soft-theorem-local-dipole}

The local form of the Banerjee--Mandal--Sahoo dipole Ward identity, in the normalization used here, rewrites the spectator sum as a local normal-ordered dipole descendant at the emitting leg.  This is the construction of section IV of \cite{Banerjee:2026soft}.

For each hard insertion, the normal-ordered descendant of eq.~(4.6) of \cite{Banerjee:2026soft}, in the all-outgoing convention used here, is
\begin{equation}
\label{eq:normal-ordered-dipole-descendant}
\begin{aligned}
  &\mathcal N_a^{(\mu)}[\bar W]
  (z_a,\bar z_a)
  \\
  &\quad
  =:D^{(\mu)}[\bar W]T_a^{-1}
  \mathcal O_{\Delta_a,J_a}^{Q_a,\eta_a}:
  (z_a,\bar z_a)
  \\
  &\quad
  =\lim_{\bar\xi\to\bar z_a}
  \left[
  D^{(\mu)}[\bar W](\bar\xi)
  T_a^{-1}\mathcal O_{\Delta_a,J_a}^{Q_a,\eta_a}(z_a,\bar z_a)
  -e_a\frac{\bar W-\bar z_a}{\bar\xi-\bar z_a}
  T_a^{-1}\mathcal O_{\Delta_a,J_a}^{Q_a,\eta_a}(z_a,\bar z_a)
  \right] \, .
\end{aligned}
\end{equation}
The subtraction is local and is fixed by the OPE \eqref{eq:dipole-hard-primary-OPE-section3}.  No choice of finite counterterm is involved in this definition.  The composite has the same spin and charge labels as the original hard operator, while its Mellin weight is shifted from \(\Delta_a\) to \(\Delta_a-1\).

\paragraph{Banerjee--Mandal--Sahoo local Ward identity.}
With the normalization of eq.~\eqref{eq:S0-Ward-pairwise}, their local Ward identity, eq.~(4.5) of \cite{Banerjee:2026soft}, reads
\begin{equation}
\label{eq:S0-local-dipole-Ward}
\begin{aligned}
  &\left\langle
  \mathsf S^0_\mu(w,\bar w)
  \prod_{i=1}^n\mathcal O_i
  \right\rangle_{\mathcal H,\mu}
  \\
  &\hspace{0.6cm}
  =\sum_{a=1}^n
  \frac{\eta_a Q_a^2}{w-z_a}
  \left\langle
  \mathcal N_a^{(\mu)}[\bar w]
  (z_a,\bar z_a)
  \prod_{\substack{i=1\\ i\ne a}}^n\mathcal O_i
  \right\rangle_{\mathcal H,\mu} \, .
\end{aligned}
\end{equation}
Here \(\bar w\) in \(\mathcal N_a^{(\mu)}[\bar w]\) denotes the spinor \(\bar W^\alpha=(\bar w,1)^t\).  Inserting the definition \eqref{eq:normal-ordered-dipole-descendant} into the right hand side of eq.~\eqref{eq:S0-local-dipole-Ward} and using the Banerjee--Mandal--Sahoo dipole Ward identity \eqref{eq:dipole-current-Ward-W} gives, before the self-subtraction,
\begin{equation}
\label{eq:normal-order-proof-first}
  \sum_{b=1}^n e_b
  \frac{\bar w-\bar z_b}{\bar z_a-\bar z_b}
  T_a^{-1}C_n(1,\ldots,n) \, ,
\end{equation}
where the expression is understood as the limit \(\bar\xi\to\bar z_a\).  The term with \(b=a\) is the pole subtracted in eq.~\eqref{eq:normal-ordered-dipole-descendant}.  The remaining finite part is
\begin{equation}
\label{eq:normal-order-proof-second}
  \sum_{\substack{b=1\\ b\ne a}}^n e_b
  \frac{\bar w-\bar z_b}{\bar z_a-\bar z_b}
  T_a^{-1}C_n(1,\ldots,n) \, .
\end{equation}
Multiplying by \(\eta_aQ_a^2/(w-z_a)\) and summing over \(a\) reproduces eq.~\eqref{eq:S0-Ward-pairwise}.

The Banerjee--Mandal--Sahoo local identity gives the celestial OPE between the logarithmic soft photon insertion and a hard primary, corresponding to their eq.~(4.7),
\begin{equation}
\label{eq:S0-hard-OPE-section3}
  \mathsf S^0_\mu(w,\bar w)
  \mathcal O_{\Delta,J}^{Q_a,\eta_a}(z_a,\bar z_a)
  \sim
  \frac{\eta_a Q_a^2}{w-z_a}
  :D^{(\mu)}[\bar w]T_a^{-1}
  \mathcal O_{\Delta,J}^{Q_a,\eta_a}:
  (z_a,\bar z_a) \, .
\end{equation}
The coefficient is not a conventional one-particle OPE coefficient.  Through the Ward identity of \(D_\alpha^{(\mu)}\), the normal-ordered composite contains the sum over the remaining charged insertions and reproduces the ordered-pair dependence of the loop soft factor.

\subsection{Pairwise support and the first obstruction}
\label{subsec:pairwise-support-first-obstruction}

The Banerjee--Mandal--Sahoo local form \eqref{eq:S0-local-dipole-Ward} packages the ordered-pair sum into a normal-ordered dipole descendant, but it does not alter the support of the operation.  Let \(F_1\) denote the layer of kernels supported on one-particle current-leg diagonals, and let \(F_2\) also allow one correlated hard-hard pair.  The logarithmic pair kernel
\begin{equation}
\label{eq:B0-operator-section3}
  K_{w,\bar w}=\sum_{a\ne b}K_{w,\bar w}^{ab}
\end{equation}
belongs to \(F_2\).  On the ordered pair \((a,b)\), its Cauchy residue is
\begin{equation}
\label{eq:pairwise-residue-section3}
  \operatorname{Res}_{\bar z_a=\bar z_b}
  (w-z_a)K_{w,\bar w}^{ab}
  =\eta_aQ_a^2e_b(\bar w-\bar z_b)T_a^{-1},
\end{equation}
up to the conventional residue factor.  The residue map vanishes on \(F_1\).  Theorem~\ref{thm:nonabsorption-section4} turns \eqref{eq:pairwise-residue-section3} into a basis-independent criterion for a nonzero class in \(F_2/F_1\).

A four-scalar contact block already has nonzero second support.  Let \(C_4^{\rm cont}=g\mathcal I_4\), where \(\mathcal I_4\) is the Mellin transform of the momentum-conservation distribution.  The ordered pair \((1,2)\) gives
\begin{equation}
\label{eq:section3-running-residue}
  \operatorname{Res}_{\bar z_1=\bar z_2}
  (w-z_1)K_{w,\bar w}^{12}C_4^{\rm cont}
  =\eta_1Q_1^2e_2(\bar w-\bar z_2)g\,T_1^{-1}\mathcal I_4.
\end{equation}
A Mellin-dependent hard label contributes its backward finite difference, multiplying the same residue in the dipole-hard commutator of section~\ref{sec:dipole-hard-ope}.

\section{Dipole-hard OPE and the analytic two-particle module}
\label{sec:dipole-hard-ope}

The logarithmic soft kernel is local at the emitting leg but singular on an ordered hard pair.  Commuting it with a hard current produces an ordinary one-particle term together with a Cauchy residue on the pair diagonal.  The residue defines the next support layer, which is resolved by the meromorphically continued two-particle primary module.

\subsection{Analytic two-particle celestial primaries}
\label{subsec:two-particle-primaries}

For a one-particle species \(\chi=(Q,\eta,J)\), let \(\mathcal V_\chi(\Delta)\) denote the corresponding celestial representation with complex conformal dimension \(\Delta\).  On the unitary principal line, \(\Delta=1+i\lambda\), the completed tensor product of two such modules has the diagonal Lorentz decomposition
\begin{equation}
  \mathcal V_{\chi_a}(1+i\lambda_a)\widehat\otimes
  \mathcal V_{\chi_b}(1+i\lambda_b)
  \simeq
  \bigoplus_{\ell}
  \int_{\mathbb R}^{\oplus}d\nu\,
  \mathcal M_{ab}(\nu,\ell)\otimes\mathcal V_{1+i\nu,\ell} \, .
  \label{eq:plancherel-decomposition}
\end{equation}
For scalar external legs, \(\ell\in\mathbb Z\).  We allow \(\ell\in\frac12\mathbb Z\) in formulas that also cover spinning external legs and the dipole-current channels considered here.  The multiplicity space \(\mathcal M_{ab}(\nu,\ell)\) includes the discrete charge, orientation and helicity data of the ordered pair.

The soft kernel contains \(T_a^{-1}\).  Consequently, even when the original leg lies on \(\operatorname{Re}\Delta_a=1\), the shifted field has \(\operatorname{Re}\Delta_a=0\).  The decomposition required here is therefore not the unitary formula \eqref{eq:plancherel-decomposition} evaluated naively off its domain.  We use its meromorphic continuation in the external dimensions.  Celestial correlators and conformal partial waves admit such continuation as distributions in Mellin space \cite{Dobrev:1977qv,Atanasov:2021cje,Borji:2024distributional,Pacifico:2025cpw,Himwich:2025shadow}.  Concretely, hard correlators are paired with wave packets that are entire in a strip containing both \(\operatorname{Re}\Delta=1\) and \(\operatorname{Re}\Delta=0\), rapidly decreasing on every vertical line in that strip.  The shift \(T_a^{-1}\) is then a continuous isomorphism between the two boundary values.

Let
\begin{equation}
  \mathbb P_{ab}^{\rm an}(\nu,\ell;\Delta_a,\Delta_b)
  \label{eq:analytic-two-particle-projector}
\end{equation}
be the meromorphic continuation of the principal-series projector.  Moving \(\Delta_a\) from \(1+i\lambda_a\) to \(i\lambda_a\) may cross poles of the intertwining kernel.  Cauchy's theorem gives a continuous principal-series integral together with the residues of the crossed poles,
\begin{equation}
\begin{aligned}
  \mathbf 1_{ab}^{\rm an}
  ={}&
  \sum_{\ell}\int_{\mathbb R}d\nu\,
  \rho_\ell(\nu)\,
  \mathbb P_{ab}^{\rm an}(\nu,\ell;\Delta_a-1,\Delta_b)
  \\
  &+\sum_{r\in\mathfrak R_{ab}}
  \mathbb P_{ab}^{\rm res}(r;\Delta_a-1,\Delta_b) \, .
  \label{eq:analytic-two-particle-completeness}
\end{aligned}
\end{equation}
The finite set \(\mathfrak R_{ab}\) depends on the external weights and is empty whenever the continuation crosses no pole.  On a pole-free connected component of the analytic test space, the second line is absent; the four-point check in section~\ref{sec:checks-examples-outlook} is carried out on such a component.  If the contour is enlarged across an intertwiner pole, its residue projector must be restored.  Equation \eqref{eq:analytic-two-particle-completeness} is understood after pairing with the analytic Mellin wave packets just described.  On the unitary line it reduces to the ordinary Plancherel identity
\begin{equation}
  \sum_{\ell}\int_{\mathbb R}d\nu\,
  \rho_\ell(\nu)\,
  \mathbb P_{ab}(\nu,\ell)
  =\mathbf 1_{\mathcal V_{ab}} \, .
  \label{eq:two-particle-completeness}
\end{equation}
The shadow-related channels \((\nu,\ell)\sim(-\nu,-\ell)\) are identified by the meromorphically continued Knapp--Stein intertwiner.  On each compact pole-free spectral strip, its graph relations and finite-dimensional pole kernels form a closed subspace of the analytic test space.  The quotient is complete and Hausdorff, while the cokernel data are the explicit residue projectors in \eqref{eq:analytic-two-particle-completeness}.  The stripwise Plancherel maps are compatible under restriction, and their inverse limit defines the global analytic module.  Proposition~\ref{prop:appC-stripwise-plancherel} and proposition~\ref{prop:appC-F2-M2-identification} give the topological statement used below.

A two-particle primary is the image of the product of two fields under one of the projectors in \eqref{eq:analytic-two-particle-completeness}.  For a continuous channel we write
\begin{equation}
\begin{aligned}
  [\mathcal O_a\mathcal O_b]_{\nu,\ell}(y,\bar y)
  ={}&\int d^2x_1\,d^2x_2\,
  \mathcal K_{\nu,\ell}^{ab}
  (y,\bar y;x_1,\bar x_1,x_2,\bar x_2)
  \\
  &\hspace{2.5cm}\times
  \mathcal O_a(x_1,\bar x_1)\mathcal O_b(x_2,\bar x_2) \, ,
  \label{eq:two-particle-primary-kernel}
\end{aligned}
\end{equation}
with weights
\begin{equation}
  h_{\nu,\ell}=\frac{1+i\nu+\ell}{2}\, ,\qquad
  \bar h_{\nu,\ell}=\frac{1+i\nu-\ell}{2} \, .
  \label{eq:two-particle-primary-weights}
\end{equation}
Residue channels are denoted by \([\mathcal O_a\mathcal O_b]_r^{\rm res}\).  The analytic two-particle module used in this paper is
\begin{equation}
\begin{aligned}
  \mathcal M_2^{\rm an}
  =\bigoplus_{a<b}\biggl[
  &\bigoplus_{\ell}\int_{\mathbb R}^{\oplus}d\nu\,
  \mathcal M_{ab}(\nu,\ell)\otimes\mathcal V_{1+i\nu,\ell}
  \\
  &\oplus\bigoplus_{r\in\mathfrak R_{ab}}
  \mathcal M_{ab}^{\rm res}(r)\biggr] \Big/\!\sim_{\rm sh} \, .
  \label{eq:M2-definition}
\end{aligned}
\end{equation}
We write \(\mathcal M_2\equiv\mathcal M_2^{\rm an}\); residue channels and analytically continued external weights are included in this notation.

Let \(F_1\) be the layer of kernels supported on one-particle current-leg diagonals, including Mellin shifts and local counterterms, and let \(F_2\) also allow one connected hard-hard pair.  For the meromorphic Cauchy kernels generated by the logarithmic soft theorem, the second associated-graded quotient is
\begin{equation}
  \operatorname{gr}_2F=F_2/F_1\simeq\mathcal M_2 \, .
  \label{eq:F2-F1-M2}
\end{equation}
The global identification in \eqref{eq:F2-F1-M2} uses assumption~\ref{ass:appC-meromorphic-continuation}.  On any fixed compact pole-free strip for which the continued projectors are defined, the algebraic arguments below require only the corresponding stripwise resolution of proposition~\ref{prop:appC-stripwise-plancherel}.  The residue criterion, the cocycle identity and the affine action can therefore be read stripwise.  Statements involving the full global module \(\mathcal M_2\), including the global generation theorem and the full minimality statement, invoke assumption~\ref{ass:appC-meromorphic-continuation}.  The identification combines the pairwise residue map with \eqref{eq:analytic-two-particle-completeness}.  The quotient map annihilates \(F_1\), whose representatives carry no hard-pair residue.  A nonzero pairwise residue gives a nonzero element of the left-hand side, and the continued Plancherel transform resolves it into the continuous and discrete channels on the right.

\begin{lemma}[One-particle scheme invariance]
\label{lem:one-particle-scheme-invariance}
Let \(K\in F_2\), \(r\in F_1\), and \(\Phi\in\mathfrak h_{\rm hard}\).  Then
\begin{equation}
  \sigma_2(K+r)=\sigma_2(K),
  \qquad
  \sigma_2([K+r,\Phi])=\sigma_2([K,\Phi]).
  \label{eq:one-particle-scheme-invariance}
\end{equation}
\end{lemma}

\begin{proof}
The first identity is the definition of the quotient map.  Hard-current endomorphisms preserve one-particle support, so \([F_1,\mathfrak h_{\rm hard}]\subset F_1\).  Applying \(\sigma_2\) to \([r,\Phi]\) therefore gives zero, which proves the second identity.
\end{proof}

\begin{proposition}[Analytic second-support resolution]
\label{prop:pairwise-plancherel-resolution}
Let \(B_{ab}\) be a distributional kernel with a conormal singularity on the pairwise diagonal and with analytic Mellin dependence in the strip \(0\leq\operatorname{Re}\Delta_a\leq1\).  Its class in \(F_2/F_1\) has the unique expansion
\begin{equation}
\begin{aligned}
  [B_{ab}]={}&
  \sum_{\ell}\int_{\mathbb R}d\nu\,\rho_\ell(\nu)\,
  B_{ab}^{\nu,\ell}
  [\mathcal O_a\mathcal O_b]_{\nu,\ell}
  \\
  &+\sum_{r\in\mathfrak R_{ab}}
  B_{ab}^{r,\rm res}
  [\mathcal O_a\mathcal O_b]_{r}^{\rm res}
  \quad \bmod F_1 \, ,
  \label{eq:pairwise-expansion-M2}
\end{aligned}
\end{equation}
modulo the shadow relation and the intrinsic null submodule of the intertwiners.
\end{proposition}

\begin{proof}
Smear the external Mellin variables with analytic wave packets and the celestial coordinates with compactly supported test functions.  On a compact pole-free spectral strip, the ordinary principal-series Plancherel transform continues to a topological isomorphism after quotienting by the closed Knapp--Stein graph relations and adjoining the finite-dimensional residues of crossed poles.  This follows by contour deformation; the continued transform and its inverse are the meromorphic continuations of the two compositions that equal the identity on the principal line.  Applying the stripwise identity to \(B_{ab}\), and then passing to the compatible inverse limit of strips, gives \eqref{eq:pairwise-expansion-M2}.  If every continuous coefficient and crossed-pole residue vanishes, the class of \(B_{ab}\) lies in the shadow-null subspace; conormal separation then leaves only an \(F_1\) representative.
\end{proof}

\subsection{The dipole-resolved hard OPE}
\label{subsec:dipole-hard-OPE-theorem}

For \(\Phi\in\mathscr D_{\rm Mell}\), write \(\Phi_a\) for its action on the \(a\)-th insertion.  For a diagonal Mellin label define its scalar backward difference and shifted commutator by
\begin{equation}
  \delta_a^-\Phi_a
  :=\Phi_a(\Delta_a-1)-\Phi_a(\Delta_a),
  \qquad
  \nabla_a^-\Phi_a
  :=T_a^{-1}\Phi_a-\Phi_aT_a^{-1}
  =(\delta_a^-\Phi_a)T_a^{-1}.
  \label{eq:backward-difference-Phi}
\end{equation}
Matrix-valued labels obey the same formula with matrix composition.  Constant global charges have \(\delta^-\Phi=\nabla^-\Phi=0\); Mellin-dependent hard currents generally do not.

A general subtraction scheme may add a one-particle representative to the ordered-pair kernel.  Modulo central contact terms, write
\begin{equation}
  \mathsf S^0_\mu(w,\bar w)
  =H[s^0_{\mu,1}(w,\bar w)]
   +\sum_{a\ne b}K_{w,\bar w}^{ab},
  \qquad H[s^0_{\mu,1}]\in F_1,
  \label{eq:S0-one-particle-pair-decomposition}
\end{equation}
and define the local hard-current derivation by
\begin{equation}
  \mathcal L^0_{w,\bar w}\Phi
  :=[s^0_{\mu,1}(w,\bar w),\Phi]_{\star}.
  \label{eq:L0-explicit-definition}
\end{equation}
The pairwise representative fixed by eq.~\eqref{eq:S0-Ward-pairwise} has \(s^0_{\mu,1}=0\).  Other infrared subtraction schemes change only the \(F_1\)-valued term in eq.~\eqref{eq:L0-explicit-definition}.  The ordered-pair kernel \(K_{w,\bar w}^{ab}\) is defined in eq.~\eqref{eq:ordered-pair-kernel-section3}; its commutator with the hard-current label is
\begin{equation}
  [K_{w,\bar w}^{ab},\Phi_a+\Phi_b]
  =\frac{\eta_aQ_a^2e_b}{w-z_a}
   \frac{\bar w-\bar z_b}{\bar z_a-\bar z_b}
   \nabla_a^-\Phi_a
  \quad \bmod F_1 \, .
  \label{eq:pair-kernel-hard-commutator}
\end{equation}

\begin{theorem}[Dipole-hard OPE]
\label{thm:dipole-hard-ope}
The singular commutator of the normalized logarithmic soft insertion with a hard current is
\begin{equation}
  [\mathsf S^0_\mu(w,\bar w),H[\Phi]]
  =H[\mathcal L^0_{w,\bar w}\Phi]
   +\mathbb M^0_{w,\bar w}[\Phi]
  \quad \bmod F_3 \, ,
  \label{eq:dipole-hard-OPE-filtered}
\end{equation}
where the second-support term acts on an \(n\)-point hard correlator by
\begin{equation}
\begin{aligned}
  \left\langle \mathbb M^0_{w,\bar w}[\Phi]
  \prod_{i=1}^n\mathcal O_i\right\rangle_{\mathcal H,\mu}
  =\sum_{a\ne b}
  \frac{\eta_aQ_a^2e_b}{w-z_a}
  \frac{\bar w-\bar z_b}{\bar z_a-\bar z_b}
  (\nabla_a^-\Phi_a)C_n \, .
  \label{eq:M-Phi-explicit}
\end{aligned}
\end{equation}
Its associated-graded class has the analytic primary expansion
\begin{equation}
\begin{aligned}
  [\mathbb M^0_{w,\bar w}[\Phi]]
  ={}&\sum_{a\ne b}\sum_\ell\int_{\mathbb R}d\nu\,
  \rho_\ell(\nu)\,
  \mathcal C^{ab}_{\Phi;\nu,\ell}(w,\bar w)
  [\mathcal O_a\mathcal O_b]_{\nu,\ell}
  \\
  &+\sum_{a\ne b}\sum_{r\in\mathfrak R_{ab}}
  \mathcal C^{ab,r}_{\Phi}(w,\bar w)
  [\mathcal O_a\mathcal O_b]_{r}^{\rm res} \, .
  \label{eq:M-Phi-M2-expansion}
\end{aligned}
\end{equation}
The coefficients are the meromorphic Plancherel pairings of the fixed kernel \eqref{eq:pair-kernel-hard-commutator}.
\end{theorem}

\begin{proof}
Apply the hard current and the logarithmic soft Ward operator in the two possible orders.  The one-particle representative in eq.~\eqref{eq:S0-one-particle-pair-decomposition} contributes
\(H[[s^0_{\mu,1},\Phi]_{\star}]=H[\mathcal L^0_{w,\bar w}\Phi]\) by lemma~\ref{lem:hard-current-representation}.  Operators on legs other than \(a\) and \(b\) commute with the ordered-pair kernel.  On the emitting leg the two orders differ by \(T_a^{-1}\Phi_a-\Phi_aT_a^{-1}=\nabla_a^-\Phi_a\), which gives \eqref{eq:M-Phi-explicit}.  The resulting kernel belongs to \(F_2\).  Proposition~\ref{prop:pairwise-plancherel-resolution} gives \eqref{eq:M-Phi-M2-expansion}, including any residue channels crossed by the shift of the external weight.
\end{proof}

Define the second-support coefficient map
\begin{equation}
  \mathfrak m^0_{w,\bar w}(\Phi)
  :=[\mathbb M^0_{w,\bar w}[\Phi]]
  \in\mathcal M_2 \, .
  \label{eq:m-cocycle-section4}
\end{equation}
The hard-current algebra acts on \(\mathcal M_2\) by the diagonal action on the two external legs.  We denote this action by \(\Phi\cdot U\).

\begin{theorem}[Hard-current one-cocycle]
\label{thm:hard-current-one-cocycle}
For every normalized logarithmic soft kernel \(X\), the map
\(\mathfrak m_X:\mathfrak h_{\rm hard}\to\mathcal M_2\) satisfies
\begin{equation}
  \mathfrak m_X([\Phi,\Psi]_{\star})
  =\Phi\cdot\mathfrak m_X(\Psi)
   -\Psi\cdot\mathfrak m_X(\Phi) \, .
  \label{eq:hard-current-one-cocycle}
\end{equation}
Thus \(\mathfrak m_X\in Z^1(\mathfrak h_{\rm hard},\mathcal M_2)\).
\end{theorem}

\begin{proof}
On every ordered pair, \(\mathfrak m_X(\Phi)\) is the second-support symbol of \([K_X,\Phi]\).  The ordinary commutator identity
\begin{equation}
  [K_X,[\Phi,\Psi]]
  =[\Phi,[K_X,\Psi]]-[\Psi,[K_X,\Phi]]
\end{equation}
holds in the finite-part kernel algebra.  Passing to \(F_2/F_1\) turns the two terms on the right into the diagonal hard-current action on \(\mathcal M_2\), which gives \eqref{eq:hard-current-one-cocycle}.
\end{proof}

\begin{proposition}[Canonical primitive and relative non-splitting]
\label{prop:canonical-primitive-relative-splitting}
Let
\begin{equation}
\label{eq:canonical-kappa-X}
  \kappa_X:=\sigma_2(K_X)\in\mathcal M_2
\end{equation}
be the analytic two-particle class of the normalized ordered-pair kernel.  Then
\begin{equation}
\label{eq:m-is-relative-coboundary}
  \mathfrak m_X(\Phi)=-\rho_2(\Phi)\kappa_X
  =-(d_{\rm CE}\kappa_X)(\Phi).
\end{equation}
No $F_1$-valued zero-cochain has the same differential whenever the pairwise residue of $K_X$ is nonzero.  Under the cyclicity hypotheses of theorem~\ref{thm:minimal-closure-section6}, the closed module generated by the family $\{\kappa_X\}$ under hard-current profiles, soft-point smearing and the diagonal Lorentz action is $\mathcal M_2$.
\end{proposition}

\begin{proof}
The hard-current action on $\mathcal M_2$ is induced by the commutator, so
\begin{equation}
  \rho_2(\Phi)\kappa_X
  =\sigma_2[\Phi,K_X]
  =-\sigma_2[K_X,\Phi]
  =-\mathfrak m_X(\Phi).
\end{equation}
If $u\in F_1$, filtration compatibility gives $\rho_2(\Phi)u\in F_1$, so its pairwise residue vanishes.  It cannot reproduce the nonzero residue of theorem~\ref{thm:nonabsorption-section4}.  The generation statement follows from lemmas~\ref{lem:angular-cyclicity-section4} and \ref{lem:mellin-cyclicity-section4} together with soft-point smearing and the analytic Plancherel resolution.
\end{proof}

\begin{proposition}[Integrated hard-current cocycle]
\label{prop:integrated-hard-current-cocycle}
Let $G_{\rm hard}^{\rm for}=\exp(\mathfrak h_{\rm hard})$ be the formal BCH group of a complete filtered hard-current subalgebra, and let $\rho_2$ be its infinitesimal action on $\mathcal M_2$.  The map
\begin{equation}
\label{eq:integrated-hard-current-cocycle}
  \mathfrak M_X(e^\Phi)
  =\int_0^1 ds\,e^{s\rho_2(\Phi)}\mathfrak m_X(\Phi)
  =\kappa_X-e^{\rho_2(\Phi)}\kappa_X
\end{equation}
defines a formal group one-cocycle,
\begin{equation}
\label{eq:group-one-cocycle}
  \mathfrak M_X(g_1g_2)
  =\mathfrak M_X(g_1)+g_1\cdot\mathfrak M_X(g_2)\, .
\end{equation}
It is the translation part of the affine action obtained by changing the origin from $0$ to $\kappa_X$ in $\mathcal M_2$.
\end{proposition}

\begin{proof}
The block operator
\begin{equation}
\label{eq:affine-block-representation}
  \widehat\rho_X(\Phi)
  =\begin{pmatrix}
     \rho_2(\Phi)&\mathfrak m_X(\Phi)\\
     0&0
   \end{pmatrix}
\end{equation}
is a Lie-algebra representation because eq.~\eqref{eq:hard-current-one-cocycle} gives
\begin{equation}
\label{eq:affine-block-commutator}
  [\widehat\rho_X(\Phi),\widehat\rho_X(\Psi)]
  =\widehat\rho_X([\Phi,\Psi]_\star) \, .
\end{equation}
Its formal exponential is
\begin{equation}
\label{eq:affine-block-exponential}
  e^{\widehat\rho_X(\Phi)}
  =\begin{pmatrix}
      e^{\rho_2(\Phi)}&
      \displaystyle\int_0^1ds\,e^{s\rho_2(\Phi)}\mathfrak m_X(\Phi)\\
      0&1
    \end{pmatrix}.
\end{equation}
Equation~\eqref{eq:m-is-relative-coboundary} evaluates the integral as
$\kappa_X-e^{\rho_2(\Phi)}\kappa_X$.  Multiplication of the block matrices gives \eqref{eq:group-one-cocycle}.
\end{proof}

\begin{proposition}[Exponentiated ordered-pair action]
\label{prop:pairwise-exponentiation}
On a fixed ordered pair, suppose
$K_X^{ab}=L_X^{ab}T_a^{-1}$ with $L_X^{ab}$ independent of $\Delta_a$.  For a scalar Mellin label $f_a(\Delta_a)$ and every $n\geq1$,
\begin{equation}
\label{eq:nth-pairwise-commutator}
  \operatorname{ad}_{K_X^{ab}}^n f_a
  =(L_X^{ab})^n(\delta_a^-)^nf_a\,T_a^{-n},
  \qquad
  \delta_a^-f_a=f_a(\Delta_a-1)-f_a(\Delta_a).
\end{equation}
Consequently, as a formal series in $\tau$ in the fixed-pair finite-part algebra,
\begin{equation}
\label{eq:pairwise-conjugation-series}
  e^{\tau K_X^{ab}}f_a e^{-\tau K_X^{ab}}
  =\sum_{n=0}^{\infty}\frac{\tau^n}{n!}
   (L_X^{ab})^n(\delta_a^-)^nf_a\,T_a^{-n}.
\end{equation}
For the exponential Mellin label $f_{a,s}=e^{s\Delta_a}$ this series resums to
\begin{equation}
\label{eq:pairwise-exponential-label}
  e^{\tau K_X^{ab}}e^{s\Delta_a}e^{-\tau K_X^{ab}}
  =e^{s\Delta_a}
   \exp\!\left[\tau L_X^{ab}(e^{-s}-1)T_a^{-1}\right].
\end{equation}
\end{proposition}

\begin{proof}
For $n=1$, eq.~\eqref{eq:nth-pairwise-commutator} is the shift relation \eqref{eq:backward-difference-Phi}.  If it holds at order $n$, commuting once more with $L_X^{ab}T_a^{-1}$ replaces $(\delta_a^-)^nf_a$ by $(\delta_a^-)^{n+1}f_a$ and adds one factor of $L_X^{ab}T_a^{-1}$.  Induction proves the first formula, and the remaining equations follow from the exponential series and $(\delta_a^-)^ne^{s\Delta_a}=(e^{-s}-1)^ne^{s\Delta_a}$.
\end{proof}

To compare directly with the long-range exponential, factor the leading soft pole from the ordered-pair kernel,
\begin{equation}
\label{eq:reduced-fixed-leg-kernel}
  K_{w,\bar w}^{ab}
  =\frac{Q_a}{w-z_a}\,\widehat K_{ab}(\bar w),
  \qquad
  \widehat K_{ab}(\bar w)
  =\eta_aQ_a e_b
   \frac{\bar w-\bar z_b}{\bar z_a-\bar z_b}T_a^{-1},
\end{equation}
and set
\begin{equation}
\label{eq:reduced-fixed-leg-sum}
  \widehat K_a(\bar w)=\sum_{b\ne a}\widehat K_{ab}(\bar w)
  =\widehat L_a(\bar w)T_a^{-1}.
\end{equation}

\begin{proposition}[Scalar-leg matching with the long-range exponential]
\label{prop:ckp-coefficient-matching}
Consider scalar hard legs and cross them to the outgoing sheet, so that $\eta_i=1$ and $e_i=Q_i$.  Let $\lambda_{\rm em}=-\gamma_{\rm em}$ be the loop coefficient of \cite{Choi:2026longrange}; the unindexed $e$ denotes the gauge coupling, while $e_i$ denotes a signed leg charge elsewhere in this paper.  After the leading soft pole is factored, the scalar-QED long-range tower of that work is generated by the reduced fixed-leg operator in \eqref{eq:reduced-fixed-leg-sum}.  In our notation the complete symmetry-governed tower is the long-range exponential of \cite{Choi:2026longrange},
\begin{equation}
\label{eq:ckp-exponential-matching}
  \left\langle
  \sum_{\ell=0}^{\infty}
  \mathcal O^{\ell\text{-loop}}_{\ell-1,+1}(w,\bar w)
  \prod_i\mathcal O_i
  \right\rangle_{\rm univ}
  =\sqrt2 e\sum_a\frac{Q_a}{w-z_a}
   \exp\!\bigl(\lambda_{\rm em}\widehat K_a(\bar w)\bigr)C_n.
\end{equation}
The coefficient at loop order $\ell$ is
\begin{equation}
\label{eq:ckp-loop-coefficient-matching}
  \frac{\sqrt2 e}{\ell!}
  \sum_a\frac{Q_a}{w-z_a}
  \bigl(\lambda_{\rm em}\widehat K_a(\bar w)\bigr)^\ell C_n.
\end{equation}
For every scalar Mellin label $f_a(\Delta_a)$, the corresponding finite-energy action on the hard-current module is
\begin{equation}
\label{eq:ckp-hard-current-lift}
  e^{\lambda_{\rm em}\widehat K_a}
  f_a(\Delta_a)
  e^{-\lambda_{\rm em}\widehat K_a}
  =\sum_{n=0}^{\infty}\frac{\lambda_{\rm em}^n}{n!}
   \widehat L_a^n(\delta_a^-)^nf_a\,T_a^{-n}.
\end{equation}
Undoing the crossing restores the signed spectator charge $e_b$ in \eqref{eq:reduced-fixed-leg-kernel} and leaves the operator identity unchanged.
\end{proposition}

\begin{proof}
The convention dictionary is
\begin{center}
\small
\renewcommand{\arraystretch}{1.12}
\begin{tabular}{@{}c c c@{}}
\hline
CKP notation & present notation & role \\
\hline
$q_{\rm CKP}$ & $q/2$ & null-vector normalization \\
$e^{-\partial_{\Delta_a}}$ & $T_a^{-1}$ & Mellin shift of the emitting scalar leg \\
$-\gamma_{\rm em}$ & $\lambda_{\rm em}$ & long-range loop coefficient \\
\hline
\end{tabular}
\end{center}
For fixed physical momentum the null-vector rescaling only changes the energy coordinate, so the comparison is most transparent after the leading soft pole has been factored.  Their scalar leading factor becomes
\begin{equation}
  S_{-1,a}^{\rm em,+}=\frac{\sqrt2 eQ_a}{w-z_a}.
\end{equation}
Their reduced scalar operator becomes
\begin{equation}
\label{eq:ckp-reduced-operator-dictionary}
  \frac{S_{0,a}^{+}}{S_{-1,a}^{+}}k_{\rm em}
  =-\gamma_{\rm em}\sum_{b\ne a}
   Q_aQ_b\frac{\bar w-\bar z_b}{\bar z_a-\bar z_b}
   e^{-\partial_{\Delta_a}}.
\end{equation}
The shift $e^{-\partial_{\Delta_a}}$ is our $T_a^{-1}$, so the right-hand side of \eqref{eq:ckp-reduced-operator-dictionary} is $\lambda_{\rm em}\widehat K_a$.  The covariant scalar expression of \cite{Choi:2026longrange} reduces to the same operator by lemma~\ref{lem:appB-spinor-to-celestial}; the factor $1/\omega_a$ in \eqref{eq:appB-celestial-reduction-lemma} becomes $T_a^{-1}$.  The ratio left after removal of the leading pole is independent of the normalization chosen for the null vector.

The symmetry-governed $\ell$-loop coefficient in \cite{Choi:2026longrange} is the $\ell$-th power of this reduced operator divided by $\ell!$, and summing these coefficients gives the long-range exponential.  The summands of $\widehat K_a=\widehat L_aT_a^{-1}$ commute for fixed emitting leg because $\widehat L_a$ is independent of $\Delta_a$.  Equations~\eqref{eq:ckp-exponential-matching} and \eqref{eq:ckp-loop-coefficient-matching} follow, and proposition~\ref{prop:pairwise-exponentiation} gives \eqref{eq:ckp-hard-current-lift}.
\end{proof}

The statement concerns scalar hard legs.  For spinning matter the angular-momentum operator contains a spin term and the reduced kernel must be enlarged accordingly.  The normalization in \eqref{eq:ckp-exponential-matching} restores the overall coupling stripped from the Ward kernel used elsewhere in the paper; the comparison is restricted to the universal symmetry-governed tower.

The cocycle identity gives a recursion for the primary coefficients.  Let 
\(\rho^{ab}_{\nu,\ell}(\Phi)\) denote the hard-current action in a continuous two-particle channel.  Projecting eq.~\eqref{eq:hard-current-one-cocycle} gives
\begin{equation}
\begin{aligned}
  \mathcal C^{ab}_{[\Phi,\Psi]_{\star};\nu,\ell}
  ={}&\rho^{ab}_{\nu,\ell}(\Phi)
      \mathcal C^{ab}_{\Psi;\nu,\ell}
      -\rho^{ab}_{\nu,\ell}(\Psi)
      \mathcal C^{ab}_{\Phi;\nu,\ell} \, ,
  \label{eq:hard-current-coefficient-recursion}
\end{aligned}
\end{equation}
with the same formula in each residue channel.  Thus the coefficients for every iterated hard-current commutator are fixed by their values on a generating set of \(\mathfrak h_{\rm hard}\).

\subsection{Pairwise residues and generation}
\label{subsec:nonabsorption-section4}

A local one-particle redefinition has the form
\begin{equation}
  H[\Phi]\longmapsto H[\Phi]+t\,H[R(\Phi)] \, ,
  \label{eq:one-particle-counterterm-hard}
\end{equation}
where the formal parameter \(t\) records support-extension degree.  It changes the OPE by an element of \(F_1\), leaving the second-support symbol unchanged.  In theorem~\ref{thm:nonabsorption-section4}, generic charged data means that there is an ordered pair \((a,b)\) with \(\eta_aQ_a^2e_b\neq0\), the soft point is separated from \(z_a\), and \(\delta_a^-\Phi_a\) is not identically zero.  Configurations on which all pairwise residues cancel form the excluded proper algebraic locus.

\begin{theorem}[Pairwise residue criterion]
\label{thm:nonabsorption-section4}
Assume that \(\nabla^-\Phi\neq0\) on at least one charged species.  For generic charged data in the preceding sense,
\begin{equation}
  \mathfrak m^0_{w,\bar w}(\Phi)
  \in F_2/F_1\simeq\mathcal M_2
\end{equation}
is nonzero.  No local one-particle redefinition of \(H[\Phi]\) removes \(\mathbb M^0_{w,\bar w}[\Phi]\).
\end{theorem}

\begin{proof}
Choose an emitting leg \(a\) for which \(\nabla_a^-\Phi_a\neq0\) and a charged spectator \(b\) with \(e_b\neq0\).  The residue of the \((a,b)\) summand of \eqref{eq:M-Phi-explicit} on \(\bar z_a=\bar z_b\) is
\begin{equation}
  \frac{\eta_aQ_a^2e_b}{w-z_a}
  (\bar w-\bar z_b)\,\nabla_a^-\Phi_a \, .
  \label{eq:nonabsorption-pair-term}
\end{equation}
It is nonzero for generic data.  Such a pole cannot be produced by an element of \(F_1\).  Varying \(\bar z_b\) separates it from the poles of all other ordered pairs, so cancellation between different spectators cannot hold identically.  The class in \(F_2/F_1\) is therefore nonzero and is unchanged by \eqref{eq:one-particle-counterterm-hard}.
\end{proof}

The generation statement uses both the profile algebra and the analytic Mellin test space.  For a charged ordered pair \((a,b)\), let \(\mathscr C_{ab}^{-1}\) be the closure, modulo smooth and one-particle kernels, of Cauchy conormal symbols
\begin{equation}
  \frac{A(z_a,\bar z_a,z_b,\bar z_b)}{\bar z_a-\bar z_b},
  \qquad A\in C_c^\infty(U_a\times U_b),
  \label{eq:cauchy-conormal-symbol-space}
\end{equation}
with \(U_a,U_b\) affine patches.  Let \(\mathscr W_{\rm an}\) be the nuclear Fr\'echet space of Mellin wave packets that are holomorphic in a strip containing \(0\leq\operatorname{Re}\Delta\leq1\) and Schwartz on every closed vertical substrip.

\begin{lemma}[Angular cyclicity]
\label{lem:angular-cyclicity-section4}
Assume \(\eta_aQ_a^2e_b\neq0\).  The closed span of the angular factors in the logarithmic kernel, under soft-point smearing and multiplication by profiles from \(\mathscr P(U_a)\widehat\otimes\mathscr P(U_b)\), is \(\mathscr C_{ab}^{-1}\).
\end{lemma}

\begin{proof}
Fix a compact subset of \(U_a\times U_b\) and choose a soft point \((w,\bar w)\) outside its two coordinate projections.  The factor
\begin{equation}
  c_{ab}(w,\bar w;z_a,\bar z_b)
  =\frac{\eta_aQ_a^2e_b}{w-z_a}(\bar w-\bar z_b)
\end{equation}
is smooth and nowhere zero on that compact set.  Each coefficient \(A\) in \eqref{eq:cauchy-conormal-symbol-space} can be written locally as \(c_{ab}B\).  Finite sums of decomposable profiles are dense because
\begin{equation}
  C_c^\infty(U_a)\widehat\otimes C_c^\infty(U_b)
  \simeq C_c^\infty(U_a\times U_b).
\end{equation}
A partition of unity completes the argument on the sphere.
\end{proof}

\begin{lemma}[Mellin cyclicity]
\label{lem:mellin-cyclicity-section4}
For \(\delta^-p(\Delta)=p(\Delta-1)-p(\Delta)\), the family
\begin{equation}
  \{(\delta^-p)T^{-1}F:
  p\in\mathbb C[\Delta],\ F\in\mathscr W_{\rm an}\}
\end{equation}
contains \(T^{-1}\mathscr W_{\rm an}\), and its closed span is \(T^{-1}\mathscr W_{\rm an}\).
\end{lemma}

\begin{proof}
For \(p(\Delta)=-\Delta\), one has \(\delta^-p=1\), so every shifted wave packet \(T^{-1}F\) occurs in the displayed family.  The reverse inclusion follows from continuity of polynomial multiplication and of \(T^{-1}\) on \(\mathscr W_{\rm an}\).
\end{proof}

\begin{lemma}[Triviality of the channel annihilator]
\label{lem:channel-annihilator-section4}
Let \(\Lambda\) be a continuous linear functional on the pairwise analytic test space.  Suppose that \(\Lambda\) vanishes on
\(\mathscr C_{ab}^{-1}\widehat\otimes T^{-1}\mathscr W_{\rm an}\) and on every crossed-pole residue vector.  Then \(\Lambda\) represents the zero functional on the Hausdorff shadow quotient \(\mathcal M_2\).
\end{lemma}

\begin{proof}
The tensor product in the hypothesis is dense by lemmas~\ref{lem:angular-cyclicity-section4} and \ref{lem:mellin-cyclicity-section4}.  Continuity extends the vanishing to its closure.  Analytic Plancherel completeness then leaves only the kernel and cokernel of the continued Knapp--Stein intertwiners.  These are the closed shadow-null relations and the explicit crossed-pole residue channels; the latter vanish by assumption.  Hence the induced functional on \(\mathcal M_2\) is zero.
\end{proof}

\begin{theorem}[Generation of the soft-accessible module]
\label{thm:soft-accessible-generation-section4}
Let \(\mathcal M_2^{\rm soft}\) be the closed submodule generated by the coefficients in \eqref{eq:M-Phi-M2-expansion} as the soft point, the charged ordered pair, polynomial Mellin labels, profiles in \(\mathscr P(U_a)\widehat\otimes\mathscr P(U_b)\), and the diagonal \(SL(2,\mathbb C)\) action vary.  If every species component has a pair with \(\eta_aQ_a^2e_b\neq0\), then
\begin{equation}
  \mathcal M_2^{\rm soft}=\mathcal M_2 \, ,
  \label{eq:soft-accessible-generation}
\end{equation}
where both sides include the discrete residue channels and are quotiented by the intrinsic shadow-null submodule.
\end{theorem}

\begin{proof}
Lemma~\ref{lem:angular-cyclicity-section4} gives the full Cauchy-conormal coefficient space on every charged pairwise diagonal, and lemma~\ref{lem:mellin-cyclicity-section4} gives every shifted analytic Mellin wave packet.  Their completed tensor product is dense in the pairwise analytic test space.  Lemma~\ref{lem:channel-annihilator-section4} shows that no nonzero continuous or crossed-pole channel annihilates the generated family in the Hausdorff shadow quotient.  This establishes \eqref{eq:soft-accessible-generation}.
\end{proof}

\begin{corollary}[Universal second-support map]
\label{cor:universal-second-support-map}
For a filtered kernel \(T\), let
\begin{equation}
  \operatorname{depth}(T)=\min\{k\geq1:T\in F_k\} \, .
  \label{eq:support-depth-definition}
\end{equation}
Modulo \(F_1\), the dipole-hard commutator factors as
\begin{equation}
  \mathscr D_{\rm Mell}
  \xrightarrow{\ \nabla^-\ }
  \mathscr D_{\rm Mell}
  \xrightarrow{\ K_{\log}\ }
  F_2/F_1
  \xrightarrow{\ \mathcal P_2^{\rm an}\ }
  \mathcal M_2 \, .
  \label{eq:universal-second-support-factorization}
\end{equation}
For generic charged data and \(\nabla^-\Phi\neq0\), the image has depth two.  Its continuous and discrete primary coefficients are fixed by the universal logarithmic soft kernel.
\end{corollary}

\section{Products of dipole currents and the two-particle action}
\label{sec:dipole-dipole-ope}

At the level of the Ward kernels used in this paper, the ordered product of two Banerjee--Mandal--Sahoo dipole insertions contains a two-particle component when the two Ward operators act on distinct charged legs.  This component is symmetric under exchange of the complete current insertions and defines the two-particle module action.  The local same-leg current in the symmetric channel is the first higher-spin current in the Banerjee--Mandal--Sahoo construction, eqs.~(7.5)--(7.7) of \cite{Banerjee:2026soft}.  This section decomposes the ordered product into the one-particle local channel and the distinct-leg two-particle channel, and computes the induced action on \(\mathcal M_2\).

\subsection{Local and pairwise channels}
\label{subsec:normal-ordering-dipole-dipole}

Write the dipole Ward operator as
\begin{equation}
  \mathscr D_\alpha(\bar u)
  =\sum_{a=1}^n e_a
  \frac{\bar Z_{a\alpha}}{\bar u-\bar z_a} \, .
  \label{eq:dipole-ward-operator-section5}
\end{equation}
Its ordered product decomposes according to whether the two multipliers act on the same external leg or on distinct legs,
\begin{equation}
  \mathscr D_\alpha(\bar u)\mathscr D_\beta(\bar v)
  =\mathscr D_{\alpha\beta}^{\rm diag}(\bar u,\bar v)
   +\mathscr D_{\alpha\beta}^{\rm pair}(\bar u,\bar v) \, ,
  \label{eq:two-dipole-naive-action}
\end{equation}
where
\begin{equation}
  \mathscr D_{\alpha\beta}^{\rm diag}(\bar u,\bar v)
  =\sum_a e_a^2
  \frac{\bar Z_{a\alpha}\bar Z_{a\beta}}
       {(\bar u-\bar z_a)(\bar v-\bar z_a)}
  \label{eq:diagonal-dipole-product}
\end{equation}
and
\begin{equation}
  \mathscr D_{\alpha\beta}^{\rm pair}(\bar u,\bar v)
  =\sum_{a\ne b}e_ae_b
  \frac{\bar Z_{a\alpha}\bar Z_{b\beta}}
       {(\bar u-\bar z_a)(\bar v-\bar z_b)} \, .
  \label{eq:pair-dipole-product}
\end{equation}
The diagonal part belongs to \(F_1\).  Expanding it about \(\bar u=\bar v\) gives
\begin{equation}
  \mathscr D_{\alpha\beta}^{\rm diag}(\bar u,\bar v)
  =\sum_{r\geq0}(-1)^r(\bar u-\bar v)^r
  \sum_a e_a^2
  \frac{\bar Z_{a\alpha}\bar Z_{a\beta}}
       {(\bar v-\bar z_a)^{r+2}} \, ,
  \label{eq:diagonal-local-expansion}
\end{equation}
so its leading coefficient is the Banerjee--Mandal--Sahoo local spin-one current of eq.~(7.5) in \cite{Banerjee:2026soft}
\begin{equation}
  \left\langle J_{\alpha\beta}^{(1)}(\bar v)
  \prod_i\mathcal O_i\right\rangle_{\mathcal H,\mu}
  =\sum_a e_a^2
  \frac{\bar Z_{a\alpha}\bar Z_{a\beta}}
       {(\bar v-\bar z_a)^2}C_n \, .
  \label{eq:J1-Ward}
\end{equation}
The remaining terms in \eqref{eq:diagonal-local-expansion} are antiholomorphic descendants in \(F_1\).

The product of two operator-valued distributions can also contain local scheme-dependent terms supported on \(u=v\).  We collect them in
\begin{equation}
  C_{\alpha\beta}^{\rm loc}(\bar u,\bar v;\mu)\in F_1 \, .
  \label{eq:singular-dipole-contraction}
\end{equation}
Their detailed coefficients depend on the infrared subtraction and normal-ordering convention.  No second-support statement uses them.

The renormalized ordered product is
\begin{equation}
  :D_\alpha(\bar u)D_\beta(\bar v):_\mu
  :=D_\alpha(\bar u)D_\beta(\bar v)
  -C_{\alpha\beta}^{\rm loc}(\bar u,\bar v;\mu) \, .
  \label{eq:DD-normal-product-definition}
\end{equation}
Away from \(u=v\), its action is fixed by applying the two Ward operators.  The extension across the diagonal is defined by the selected local subtraction.  Changes of subtraction alter only \(F_1\).

The distinct-leg term defines a two-particle distribution by
\begin{equation}
\begin{aligned}
  \left\langle\mathbb B_{\alpha\beta}^{(2)}(\bar u,\bar v)
  \prod_i\mathcal O_i\right\rangle_{\mathcal H,\mu}
  =\sum_{a\ne b}e_ae_b
  \frac{\bar Z_{a\alpha}\bar Z_{b\beta}}
       {(\bar u-\bar z_a)(\bar v-\bar z_b)}C_n \, .
  \label{eq:B-alpha-beta-Ward}
\end{aligned}
\end{equation}
Its analytic primary resolution is obtained from \eqref{eq:analytic-two-particle-completeness}; in shorthand,
\begin{equation}
  [\mathbb B_{\alpha\beta}^{(2)}]
  =\sum_{a\ne b}\mathcal P_{2,ab}^{\rm an}
  \left[
  e_ae_b
  \frac{\bar Z_{a\alpha}\bar Z_{b\beta}}
       {(\bar u-\bar z_a)(\bar v-\bar z_b)}
  \right] \, .
  \label{eq:B-alpha-beta-M2-expansion}
\end{equation}

\begin{proposition}[Ordered dipole product]
\label{prop:DD-OPE}
Modulo one-particle descendants and local contact terms, the ordered product has the filtered expansion
\begin{equation}
  D_\alpha(\bar u)D_\beta(\bar v)
  =C_{\alpha\beta}^{\rm loc}(\bar u,\bar v;\mu)
   +J_{\alpha\beta}^{(1)}(\bar v)
   +\mathbb B_{\alpha\beta}^{(2)}(\bar u,\bar v)
  \quad \bmod F_1^{\rm desc} \, .
  \label{eq:DD-OPE-filtered}
\end{equation}
The second-support term obeys
\begin{equation}
  \mathbb B_{\alpha\beta}^{(2)}(\bar u,\bar v)
  =\mathbb B_{\beta\alpha}^{(2)}(\bar v,\bar u) \, .
  \label{eq:B-pair-symmetry}
\end{equation}
Consequently, the distinct-leg component of the abelian dipole-current commutator vanishes in \(F_2/F_1\).
\end{proposition}

\begin{proof}
Equation \eqref{eq:two-dipole-naive-action} separates the same-leg and distinct-leg contributions.  The former gives \(J_{\alpha\beta}^{(1)}\), its descendants, and the chosen local term.  The latter is \eqref{eq:B-alpha-beta-Ward}.  Relabeling \(a\leftrightarrow b\) in the sum proves \eqref{eq:B-pair-symmetry}.  The analytic primary expansion follows from proposition~\ref{prop:pairwise-plancherel-resolution}.
\end{proof}

The decomposition
\begin{equation}
  V_{1/2}\otimes V_{1/2}
  =\operatorname{Sym}^2V_{1/2}\oplus\wedge^2V_{1/2}
  \simeq V_1\oplus V_0
  \label{eq:spinor-product-decomposition}
\end{equation}
places \(J_{\alpha\beta}^{(1)}\) in the symmetric spin-one channel.  The pairwise primary module contains its own scalar and spin-one projections.  Equation \eqref{eq:B-pair-symmetry} concerns exchange of the two complete current insertions and does not identify those representation channels.

\subsection{Action on the two-particle module}
\label{subsec:two-particle-module-action}

Let \(f^\alpha(\bar z)\) be a meromorphic test spinor with poles away from the hard insertions and define
\begin{equation}
  D[f]=\oint\frac{d\bar z}{2\pi i}\,
  f^\alpha(\bar z)D_\alpha(\bar z) \, .
  \label{eq:integrated-dipole-charge}
\end{equation}
On the \(a\)-th one-particle module it acts by the character
\begin{equation}
  q_a(f)=e_a f^\alpha(\bar z_a)\bar Z_{a\alpha} \, .
  \label{eq:q-character-definition}
\end{equation}
The induced action on a two-particle channel is diagonal,
\begin{equation}
  \rho_2(D[f])U_{ab}
  =\bigl(q_a(f)+q_b(f)\bigr)U_{ab} \, ,
  \qquad U_{ab}\in\mathcal M_{2,ab} \, .
  \label{eq:dipole-action-M2}
\end{equation}

The characters \(q_a(f)\) are commuting scalar multipliers, so
\begin{equation}
  [\rho_2(D[f]),\rho_2(D[g])]=0 \, .
\end{equation}
Thus \(\rho_2(D[f])\) is a representation of the abelian dipole-current algebra.  The distinct-leg terms in the product of \(q_a(f)+q_b(f)\) and \(q_a(g)+q_b(g)\) are \(q_a(f)q_b(g)+q_b(f)q_a(g)\); they reproduce the symmetric pairwise Ward kernel, while the same-leg terms lie in the local channel.

For a normalized logarithmic soft insertion \(X\), let \(K_X^{ab}\) denote its finite-part kernel on the ordered pair \((a,b)\).  Its action on an analytic two-particle distribution is
\begin{equation}
  \rho_2(X)U_{ab}:=[K_X^{ab},U_{ab}]_{\rm fp} \, ,
  \label{eq:log-soft-action-M2}
\end{equation}
where the product is defined by analytic regularization in the relative Cauchy coordinate.

On each ordered pair, associativity of the analytically regularized finite-part product gives
\begin{equation}
  [\rho_2(X),\rho_2(Y)]U
  =\rho_2([X,Y]_1)U
  \qquad \bmod F_3 \, .
  \label{eq:log-soft-representation-M2}
\end{equation}
The left-hand side is
\begin{equation*}
  [K_X,[K_Y,U]]_{\rm fp}-[K_Y,[K_X,U]]_{\rm fp} \, .
\end{equation*}
Finite-part associativity identifies it with \([[K_X,K_Y]_{\rm fp},U]_{\rm fp}\).  Local ambiguities lie in \(F_1\), and a product involving an additional independent pair lies in \(F_3\).  In the diagonal scalar sector the pairwise kernels commute.

The local term \(C_{\alpha\beta}^{\rm loc}\) lies in \(F_1\) and drops out of the second-support symbol.

\section{Hard-current cocycles, Jacobi compatibility and minimal closure}
\label{sec:associativity-minimal-closure}

The nonzero second-support term in the dipole-hard OPE is a module-valued correction to the one-particle current algebra.  Compatibility with the hard-current bracket is the one-cocycle theorem of section~\ref{sec:dipole-hard-ope}.  Compatibility of two logarithmic soft insertions follows from associativity of the finite-part kernel algebra.  Minimality imposes the substantive restriction: the pairwise image cannot be represented in the one-particle quotient, and closure for all hard-current labels generates the analytic module \(\mathcal M_2\).

\subsection{The support-graded current module}
\label{subsec:support-graded-current-complex}

Let
\begin{equation}
  \mathfrak g_1
  =\mathfrak h_{\rm hard}\rtimes\mathfrak s_{\log}
  \label{eq:g1-naive-algebra}
\end{equation}
be the one-particle current algebra.  The bracket on \(\mathfrak h_{\rm hard}\) is \([\Phi,\Psi]_{\star}\).  The algebra \(\mathfrak s_{\log}\) is the one-particle closure generated by the normalized logarithmic insertions, the ordinary dipole charges \(D[f]\), and the local currents \(J_{\alpha\beta}^{(1)}\) with their antiholomorphic descendants arising from same-leg products; the chosen local contact ideal is included in this closure.  Its distinct-leg quotient is abelian.  Local one-particle contact brackets, when present, are retained in \([X,Y]_1\in\mathfrak g_1\).  Their action on hard-current labels is denoted by
\begin{equation}
  [X,H[\Phi]]_1=H[\mathcal L_X\Phi] \, .
  \label{eq:one-particle-soft-hard-bracket}
\end{equation}
The bracket in \eqref{eq:one-particle-soft-hard-bracket} is the projection to \(F_1\).

The second layer is the analytic two-particle module \(\mathcal M_2\simeq F_2/F_1\).  Hard currents act diagonally on its two external legs.  Ordinary dipole charges act by \eqref{eq:dipole-action-M2}, and logarithmic soft insertions act through the finite-part kernel commutator \eqref{eq:log-soft-action-M2}.  These actions preserve the continuous principal-series component, its analytically continued boundary values, and the crossed-pole residue channels.

The dipole-hard OPE defines an alternating two-cochain on \(\mathfrak g_1\) with values in \(\mathcal M_2\),
\begin{equation}
\begin{aligned}
  \nu(X,H[\Phi])&=\mathfrak m_X(\Phi)\, ,
  &\nu(H[\Phi],X)&=-\mathfrak m_X(\Phi)\, ,
  \\
  \nu(H[\Phi],H[\Psi])&=0\, ,
  &\nu(X,Y)&=0
\end{aligned}
\label{eq:nu-definition-section6}
\end{equation}
for the abelian distinct-leg sector.  The last equality expresses proposition~\ref{prop:DD-OPE}: the pairwise part of the ordinary dipole product is symmetric and contributes no module-valued soft-soft commutator.

\subsection{Mixed Jacobi compatibility}
\label{subsec:second-layer-jacobi}

The Chevalley--Eilenberg equation for \(\nu\) on a triple \((X,Y,H[\Phi])\) is
\begin{equation}
\begin{aligned}
  0={}&(d_{\rm CE}\nu)(X,Y,H[\Phi])
  \\
  ={}&\rho_2(X)\mathfrak m_Y(\Phi)
      -\rho_2(Y)\mathfrak m_X(\Phi)
      -\mathfrak m_{[X,Y]_1}(\Phi)
  \\
  &\quad-\mathfrak m_Y(\mathcal L_X\Phi)
      +\mathfrak m_X(\mathcal L_Y\Phi) \, .
  \label{eq:mixed-CE-condition-section6}
\end{aligned}
\end{equation}
All terms take values in the same analytic two-particle module.

For a fixed ordered pair \((a,b)\), let \(K_X^{ab}\) and \(K_Y^{ab}\) be the finite-part kernels representing the two normalized logarithmic insertions, and let \(\Phi_{ab}=\Phi_a+\Phi_b\).  The second-support symbol of \eqref{eq:mixed-CE-condition-section6} is
\begin{equation}
\begin{aligned}
  &[K_X^{ab},[K_Y^{ab},\Phi_{ab}]]_{\rm fp}
  -[K_Y^{ab},[K_X^{ab},\Phi_{ab}]]_{\rm fp}
  \\
  &\hspace{4.0cm}
  -[[K_X^{ab},K_Y^{ab}]_{\rm fp},\Phi_{ab}]_{\rm fp} \, .
  \label{eq:mixed-CE-symbol-section6}
\end{aligned}
\end{equation}
Products on the same pairwise diagonal are defined by analytic regularization in the relative coordinate \(\xi=\bar z_a-\bar z_b\): replace every Cauchy monomial by \(\xi^{-r+\varepsilon}\), multiply in the convergent half-plane, and take the finite part at \(\varepsilon=0\).  The associator is a local distribution supported at \(\xi=0\); it lies in \(F_1\), so the product is associative in the quotient \(F_2/F_1\).  Different local extensions define the same second-support class.

\begin{theorem}[Mixed Jacobi compatibility]
\label{thm:second-layer-jacobi}
The cochain \(\nu\) in \eqref{eq:nu-definition-section6} satisfies
\begin{equation}
  d_{\rm CE}\nu=0
  \label{eq:second-layer-jacobi-closed}
\end{equation}
on all triples containing one hard current and two normalized logarithmic soft insertions, modulo \(F_3\); this is eq.~\eqref{eq:mixed-CE-condition-section6} in \(\mathcal M_2\).
\end{theorem}

\begin{proof}
On each ordered pair, \eqref{eq:mixed-CE-symbol-section6} is the Jacobi identity for three elements of the associative finite-part kernel algebra.  It therefore vanishes before the Plancherel transform.  Terms in which a soft current acts on the one-particle label of the other current give the \(\mathcal L_X\Phi\), \(\mathcal L_Y\Phi\), and \([X,Y]_1\) contributions in \eqref{eq:mixed-CE-condition-section6}.  Passing to \(F_2/F_1\) and then applying the analytic resolution preserves the equality.  Since the calculation is pairwise, no third-support term is produced.
\end{proof}

Associativity enters through the Jacobi identity of the ambient finite-part kernel algebra.  The substantive input is the nonzero second-support symbol of the individual commutators: it has no representative in the one-particle quotient, so the projection \(F_2\to F_1\) is not a homomorphism of the filtered bracket.

By the bigrading \eqref{eq:loop-support-bigrading}, theorem~\ref{thm:second-layer-jacobi} lies at physical loop degree two and formal support degree one; it concerns the bilinear contribution fixed by the one-loop Ward kernels.

\subsection{The filtered abelian extension}
\label{subsec:renormalized-extension-section6}

Introduce a formal parameter \(t\) for support-extension degree and set
\begin{equation}
  \mathfrak g_{\rm ren}^{\log}
  =\mathfrak g_1\oplus_{\nu}t\mathcal M_2\, ,
  \qquad t^2=0 \, .
  \label{eq:renormalized-algebra-section6}
\end{equation}
For \(x,y\in\mathfrak g_1\) and \(u,v\in\mathcal M_2\), define
\begin{equation}
  [(x,tu),(y,tv)]_{\rm ren}
  =\bigl([x,y]_1,\,
  t(x\cdot v-y\cdot u+\nu(x,y))\bigr) \, .
  \label{eq:renormalized-bracket-section6}
\end{equation}
The parameter \(t\) is the formal variable introduced in \eqref{eq:formal-support-parameter}; it is unrelated to physical loop degree.  All brackets are also evaluated modulo \(F_3\).

\begin{corollary}[Filtered closure]
\label{cor:renormalized-closure-section6}
The bracket \eqref{eq:renormalized-bracket-section6} satisfies the Jacobi identity modulo \(F_3\).  Its restriction to one hard current and one logarithmic soft insertion reproduces theorem~\ref{thm:dipole-hard-ope}, and its one-particle quotient is \(\mathfrak g_1\).
\end{corollary}

\begin{proof}
Equation~\eqref{eq:hard-current-one-cocycle} is the Chevalley--Eilenberg equation on two hard currents and one soft insertion; theorem~\ref{thm:second-layer-jacobi} gives the component with two soft insertions and one hard current.  The all-hard component is the Jacobi identity of \(\mathfrak h_{\rm hard}\); the all-soft component is abelian in the distinct-leg sector, with local terms contained in \(F_1\).  Hence \(d_{\rm CE}\nu=0\) on every triple relevant through support depth two.  The standard abelian-extension calculation then proves the Jacobi identity for \eqref{eq:renormalized-bracket-section6}.
\end{proof}

\begin{corollary}[Relative splitting]
\label{cor:relative-splitting-section6}
Define a one-cochain $\alpha:\mathfrak g_1\to\mathcal M_2$ by
$\alpha(H[\Phi])=0$ and $\alpha(X)=\kappa_X$.  Through support depth two,
\begin{equation}
\label{eq:nu-relative-coboundary}
  \nu=d_{\rm CE}\alpha.
\end{equation}
Consequently the change of generators
\begin{equation}
\label{eq:relative-splitting-generator-change}
  X\longmapsto X-t\kappa_X
\end{equation}
splits the extension after adjoining $\mathcal M_2$.  No change of generators valued in $F_1$ produces such a splitting when the pairwise residue is nonzero.
\end{corollary}

\begin{proof}
On mixed pairs, eq.~\eqref{eq:nu-relative-coboundary} is proposition~\ref{prop:canonical-primitive-relative-splitting}.  On hard-hard pairs both sides vanish.  In the abelian distinct-leg sector the scalar pair kernels commute, hence $\rho_2(X)\kappa_Y-\rho_2(Y)\kappa_X-\kappa_{[X,Y]_1}=0$ modulo $F_1$; this is the soft-soft component.  The standard change-of-splitting formula gives \eqref{eq:relative-splitting-generator-change}, and the final assertion follows from the pairwise residue criterion.
\end{proof}

\begin{corollary}[Universal finite-energy recursion]
\label{cor:universal-finite-energy-recursion}
Let \(\{\Phi_r\}\) generate \(\mathfrak h_{\rm hard}\).  The second-support coefficients of every iterated hard-current commutator with a fixed logarithmic insertion \(X\) are determined by the seed values \(\mathfrak m_X(\Phi_r)\), the module action, and the cocycle identity \eqref{eq:hard-current-one-cocycle}.  The analytic Plancherel transform propagates this recursion to every continuous and crossed-pole channel in the symmetry-governed logarithmic sector.
\end{corollary}

\begin{proof}
Apply \eqref{eq:hard-current-one-cocycle} recursively to the Lie words in the generators \(\Phi_r\).  At each step the new coefficient is a sum of module actions on coefficients of shorter Lie words.  The analytic Plancherel transform is linear and therefore preserves this recursion.
\end{proof}

\subsection{Minimality}
\label{subsec:minimality-section6}

\begin{theorem}[Minimal closure]
\label{thm:minimal-closure-section6}
Let \(\mathfrak B\) be a filtered, \(SL(2,\mathbb C)\)-covariant current algebra acting on the infrared-subtracted hard correlators.  Assume that
\begin{enumerate}
  \item its one-particle quotient contains \(\mathfrak h_{\rm hard}\rtimes\mathfrak s_{\log}\);
  \item its mixed soft-hard bracket reproduces the logarithmic Ward kernel \eqref{eq:M-Phi-explicit};
  \item it is closed under polynomial Mellin-difference hard currents, local smooth coordinate profiles on each hard leg, soft-point smearing, and the diagonal Lorentz action;
  \item its bracket is Jacobi compatible through support depth two.
\end{enumerate}
Then \(\operatorname{gr}_2\mathfrak B\) contains a subquotient isomorphic to \(\mathcal M_2\), including the residue channels crossed by the soft Mellin shift.  Consequently, \(\mathfrak g_{\rm ren}^{\log}\) is minimal among such extensions.
\end{theorem}

\begin{proof}
The pairwise residue criterion gives a nonzero class \(\mathfrak m_X(\Phi)\in\operatorname{gr}_2\mathfrak B\) whenever \(\nabla^-\Phi\neq0\).  Closure under hard-current brackets contains the full cocycle orbit of that class.  Closure under local coordinate profiles and soft-point smearing gives the angular conormal symbols by lemma~\ref{lem:angular-cyclicity-section4}, while lemma~\ref{lem:mellin-cyclicity-section4} gives all analytically shifted Mellin wave packets.  The generation theorem \eqref{eq:soft-accessible-generation} identifies the resulting closed module with \(\mathcal M_2\) modulo its intrinsic shadow-null submodule.  A one-particle scheme change lies in \(F_1\) and therefore leaves this generated subquotient unchanged.
\end{proof}

\section{Four-point realization and exponentiated action}
\label{sec:checks-examples-outlook}

The four-point analysis includes a tree-level scalar-QED photon-exchange block and the contact block used in section~\ref{subsec:pairwise-support-first-obstruction}.  Both test the residue criterion, the hard-current cocycle and the mixed Jacobi identity; the exchange block also retains a genuine propagator pole.

Let
\begin{equation}
  \mathcal O_i=
  \mathcal O^{Q_i,\eta_i}_{\Delta_i,0}(z_i,\bar z_i),
  \qquad i=1,\ldots,4,
  \label{eq:section7-four-scalar-operators}
\end{equation}
with \(\sum_i e_i=0\), and write
\begin{equation}
\begin{aligned}
  C_4(1,2,3,4)={}&
  \int_0^\infty\prod_{i=1}^4d\omega_i\,
  \omega_i^{\Delta_i-1}
  \delta^{(4)}\!\left(\sum_i\eta_i\omega_iq_i\right)
  \\
  &\hspace{1.7cm}\times
  \mathcal H_4^{\rm ren}
  (\omega_i,z_i,\bar z_i;Q_i,\eta_i;\mu) \, .
  \label{eq:section7-C4-definition}
\end{aligned}
\end{equation}

For the running exchange block \eqref{eq:running-exchange-hard-block} and the Mellin-counting current \(\Phi_1(\Delta_1)=\Delta_1\), the ordered pair \((1,2)\) obeys
\begin{equation}
  \sigma_2\mathbb M^0_{12}[\Phi_1]C_4^{\rm ex}
  =-\eta_1Q_1^2e_2
  \frac{\bar w-\bar z_2}
       {(w-z_1)(\bar z_1-\bar z_2)}
  T_1^{-1}C_4^{\rm ex} .
  \label{eq:section7-exchange-M12}
\end{equation}
The propagator pole in \eqref{eq:running-exchange-hard-block} is untouched by the universal Ward operator.  The associated two-particle coefficient is the universal linear functional
\begin{equation}
\begin{aligned}
  \mathcal C_{\rm ex;\nu,\ell}^{12}(w,\bar w)
  ={}&-\eta_1Q_1^2e_2
  \left\langle
  \mathbb P_{12}^{\rm an}(\nu,\ell;\Delta_1-1,\Delta_2),
  \right.\\
  &\left.\hspace{1.2cm}
  \frac{\bar w-\bar z_2}
       {(w-z_1)(\bar z_1-\bar z_2)}
  T_1^{-1}C_4^{\rm ex}
  \right\rangle ,
  \label{eq:section7-exchange-primary-coefficient}
\end{aligned}
\end{equation}
with the corresponding residue projections included when the continued contour crosses a pole.

\subsection{Pairwise residues and primary coefficients}
\label{subsec:section7-low-point}

The normalized logarithmic soft insertion acts by
\begin{equation}
\begin{aligned}
  &\left\langle
  \mathsf S^0_\mu(w,\bar w)
  \mathcal O_1\mathcal O_2\mathcal O_3\mathcal O_4
  \right\rangle_{\mathcal H,\mu}
  \\
  &\qquad=
  \sum_{a=1}^{4}\sum_{\substack{b=1\\ b\ne a}}^{4}
  \frac{\eta_aQ_a^2e_b}{w-z_a}
  \frac{\bar w-\bar z_b}{\bar z_a-\bar z_b}
  T_a^{-1}C_4 \, .
  \label{eq:section7-S0-four-point}
\end{aligned}
\end{equation}
The residue on the unordered diagonal \(\bar z_1=\bar z_2\) is
\begin{equation}
\begin{aligned}
  \operatorname{Res}_{\bar z_1=\bar z_2}
  \sigma_2(\mathsf S^0_\mu C_4)
  ={}&
  \frac{\eta_1Q_1^2e_2}{w-z_1}
  (\bar w-\bar z_2)T_1^{-1}C_4
  \\
  &-
  \frac{\eta_2Q_2^2e_1}{w-z_2}
  (\bar w-\bar z_1)T_2^{-1}C_4 \, .
  \label{eq:section7-residue-12}
\end{aligned}
\end{equation}
For generic charged data in the sense defined before theorem~\ref{thm:nonabsorption-section4}, this distribution is nonzero.  It is therefore excluded from \(F_1\) by theorem~\ref{thm:nonabsorption-section4} and has the analytic primary decomposition of proposition~\ref{prop:pairwise-plancherel-resolution}.

For a concrete hard-current insertion, take the contact block
\begin{equation}
  C_4^{\rm cont}=g\,\mathcal I_4
  \label{eq:section7-contact-block}
\end{equation}
and retain the Mellin-counting label \(\Phi_1(\Delta_1)=\Delta_1\).  Its scalar backward difference is \(\delta_1^-\Phi_1=-1\), hence \(\nabla_1^-\Phi_1=-T_1^{-1}\).  The ordered pair \((1,2)\) gives
\begin{equation}
  \sigma_2\mathbb M^0_{12}[\Phi_1]C_4^{\rm cont}
  =-\eta_1Q_1^2e_2
  \frac{\bar w-\bar z_2}
       {(w-z_1)(\bar z_1-\bar z_2)}
  T_1^{-1}C_4^{\rm cont} \, .
  \label{eq:section7-explicit-M12}
\end{equation}
Its pairwise residue is
\begin{equation}
  \operatorname{Res}_{\bar z_1=\bar z_2}
  (w-z_1)\sigma_2\mathbb M^0_{12}[\Phi_1]C_4^{\rm cont}
  =-\eta_1Q_1^2e_2(\bar w-\bar z_2)g\,
  T_1^{-1}\mathcal I_4 \, .
  \label{eq:section7-explicit-residue}
\end{equation}
The continuous two-particle coefficient is the meromorphic Plancherel pairing
\begin{equation}
\begin{aligned}
  \mathcal C^{12}_{\Phi_1;\nu,\ell}(w,\bar w)
  ={}&-\eta_1Q_1^2e_2g
  \left\langle
  \mathbb P_{12}^{\rm an}(\nu,\ell;\Delta_1-1,\Delta_2),
  \right.
  \\
  &\left.\hspace{1.6cm}
  \frac{\bar w-\bar z_2}
       {(w-z_1)(\bar z_1-\bar z_2)}
  T_1^{-1}\mathcal I_4
  \right\rangle \, ,
  \label{eq:section7-explicit-C12}
\end{aligned}
\end{equation}
On the pole-free component chosen for this check, \(\mathfrak R_{12}=\varnothing\).  If the Mellin contour is deformed across an intertwiner pole, the corresponding coefficient is obtained from the same pairing with \(\mathbb P_{12}^{\rm res}\).  Formula \eqref{eq:section7-explicit-C12} then fixes the multiparticle coefficient from the contact partial-wave coefficient.

For the contact block \eqref{eq:section7-contact-block}, the map \(\Phi\mapsto\sigma_2\mathbb M^0[\Phi]C_4^{\rm cont}\) obeys the hard-current one-cocycle identity \eqref{eq:hard-current-one-cocycle}.  Equation~\eqref{eq:section7-explicit-residue} shows that it is nonzero for \(\Phi_1=\Delta_1\) whenever \(\eta_1Q_1^2e_2g\neq0\).

For the exponential Mellin label $\Phi_{1,s}=e^{s\Delta_1}$, proposition~\ref{prop:pairwise-exponentiation} gives a closed tower on the ordered pair $(1,2)$.  Writing
\begin{equation}
\label{eq:section7-L12-definition}
  L_{12}(w,\bar w)
  =\eta_1Q_1^2e_2
   \frac{\bar w-\bar z_2}{(w-z_1)(\bar z_1-\bar z_2)},
\end{equation}
one obtains
\begin{equation}
\label{eq:section7-exponentiated-contact-action}
  e^{\tau K_{12}}e^{s\Delta_1}e^{-\tau K_{12}}C_4^{\rm cont}
  =e^{s\Delta_1}
   \exp\!\left[\tau L_{12}(e^{-s}-1)T_1^{-1}\right]
   C_4^{\rm cont}.
\end{equation}
Writing the expansion as $\sum_{n\geq0}\tau^n\mathcal C_{s;\nu,\ell}^{12,(n)}/n!$, the continuous primary coefficients are
\begin{equation}
\label{eq:section7-exponentiated-primary-coefficient}
\begin{aligned}
  \mathcal C_{s;\nu,\ell}^{12,(n)}
  ={}&(e^{-s}-1)^n
  \left\langle
  \mathbb P_{12}^{\rm an}(\nu,\ell;\Delta_1-n,\Delta_2),
  \right.\\
  &\left.\hspace{2.0cm}
  L_{12}^{n}e^{s\Delta_1}T_1^{-n}C_4^{\rm cont}
  \right\rangle.
\end{aligned}
\end{equation}
The residue channels obey the same formula with $\mathbb P_{12}^{\rm res}$.  Factoring the leading soft pole as in eq.~\eqref{eq:reduced-fixed-leg-kernel} and setting $\tau=\lambda_{\rm em}$ reproduces the hard-current lift \eqref{eq:ckp-hard-current-lift}; hence the coefficients are fixed by the same one-loop ordered-pair operator that appears in the CKP exponential.

\subsection{The mixed soft-soft-hard identity}
\label{subsec:section7-soft-soft-hard-check}

Let \(X=(w_1,\bar w_1)\), \(Y=(w_2,\bar w_2)\), and let \(K_X^{ab}\), \(K_Y^{ab}\) be the corresponding ordered-pair kernels.  On the pair \((a,b)\), the universal bilinear contribution to the mixed Jacobiator is
\begin{equation}
\begin{aligned}
  \Omega_{XY;\Phi}^{ab}={}&
  [K_X^{ab},[K_Y^{ab},\Phi_a+\Phi_b]]_{\rm fp}
  -[K_Y^{ab},[K_X^{ab},\Phi_a+\Phi_b]]_{\rm fp}
  \\
  &-[[K_X^{ab},K_Y^{ab}]_{\rm fp},\Phi_a+\Phi_b]_{\rm fp} \, .
  \label{eq:section7-Omega-four-point}
\end{aligned}
\end{equation}
For a scalar label \(\Phi_a=f_a(\Delta_a)\), write
\(K_X^{ab}=L_X^{ab}T_a^{-1}\) and
\(K_Y^{ab}=L_Y^{ab}T_a^{-1}\).  Direct composition gives
\begin{equation}
\begin{aligned}
  [K_X^{ab},[K_Y^{ab},\Phi_a+\Phi_b]]_{\rm fp}
  =L_X^{ab}L_Y^{ab}
  \bigl(&f_a(\Delta_a-2)-2f_a(\Delta_a-1)
  \\
  &+f_a(\Delta_a)\bigr)T_a^{-2} \, .
  \label{eq:section7-explicit-double-commutator}
\end{aligned}
\end{equation}
The expression is symmetric under \(X\leftrightarrow Y\).  In the diagonal scalar sector \([K_X^{ab},K_Y^{ab}]_{\rm fp}=0\), and hence \(\Omega_{XY;\Phi}^{ab}=0\).  Matrix-valued kernels obey the full three-commutator identity of theorem~\ref{thm:second-layer-jacobi}.

Summing \eqref{eq:section7-Omega-four-point} over ordered pairs gives the second-support symbol of the Jacobiator.  The finite-part commutator Jacobi identity sets each pairwise summand to zero, while products involving two independent hard pairs have support depth at least three.  Therefore
\begin{equation}
  \sigma_2\mathcal J_{XY;\Phi}^{(4)}=0
  \qquad \bmod F_3 \, .
  \label{eq:section7-Jacobiator-vanishes}
\end{equation}

Projecting to \(F_1\) before forming the iterated brackets deletes the module-valued terms \(\mathfrak m_X(\Phi)\).  The resulting quotient therefore does not retain the filtered action, although the ambient Jacobi identity remains valid.  The extension of section~\ref{subsec:renormalized-extension-section6} restores the missing target.

\subsection{Representation content and further directions}
\label{subsec:section7-representation-outlook}

The symmetric dipole product contains the local spin-one current \(J_{\alpha\beta}^{(1)}\) and the symmetric projection of the two-particle distribution,
\begin{equation}
  :D_{(\alpha}(\bar u)D_{\beta)}(\bar v):_\mu
  =J_{\alpha\beta}^{(1)}(\bar v)
   +\mathbb B_{(\alpha\beta)}^{(2)}(\bar u,\bar v)
  \quad \bmod F_1^{\rm desc} \, .
  \label{eq:section7-symmetric-DD}
\end{equation}
The first term transforms in \(\operatorname{Sym}^2V_{1/2}\simeq V_1\), while the second is resolved by the analytic two-particle projectors.  This verifies the first higher-spin channel of the logarithmic current tower without asserting closure of the complete tower \cite{Banerjee:2026soft,Guevara:2021abz,Himwich:2021dau,Strominger:2021mtt}.

For nonabelian gauge theory, matrix-valued ordered-pair kernels should define a hard-current one-cocycle with values in color-ordered analytic two-particle modules.  The web algebra should replace the commuting scalar kernel algebra in the mixed compatibility condition.  In gravity, the same support filtration is accompanied by orbital and spin differential operators.  At higher loop order, overlapping pairwise kernels may produce higher support layers \(\mathcal M_k\).  These extensions require new soft input and are not assumed in the results above.

\appendix
\section{Conventions and logarithmic Mellin renormalization}
\label{app:conventions-logarithmic-mellin}

The all-outgoing parametrization separates orientation from the positive Mellin energy.  Soft-energy logarithms are treated by finite-part Mellin transforms, which produce the logarithmic current module.

\subsection{Spinor-helicity variables and celestial coordinates}
\label{appA:spinor-celestial}

We use mostly-plus signature, so that a future-directed null direction associated with a point on the celestial sphere is
\begin{equation}
\label{eq:appA-q-vector}
  q^\mu(z,\bar z)
  =\bigl(1+z\bar z,\,z+\bar z,\,-i(z-\bar z),\,1-z\bar z\bigr)\, .
\end{equation}
It obeys
\begin{equation}
\label{eq:appA-q-inner}
  q(z,\bar z)^2=0\, ,
  \qquad
  q(z,\bar z)\cdot q(w,\bar w)=-2|z-w|^2\, .
\end{equation}
All external particles are written in an all-outgoing convention,
\begin{equation}
\label{eq:appA-p-parametrization}
  p_i^\mu=\eta_i\omega_i q^\mu(z_i,\bar z_i)\, ,
  \qquad
  \omega_i>0\, ,
  \qquad
  \eta_i=\pm1\, .
\end{equation}
The physical electric charge is denoted by \(Q_i\), while the charge entering Ward identities is
\begin{equation}
\label{eq:appA-signed-charge}
  e_i=\eta_iQ_i\, .
\end{equation}
With this convention energy variables are always integrated over the positive real line, and incoming particles are represented by the sign \(\eta_i=-1\) rather than by negative energy.

The spinor-helicity variables are
\begin{equation}
\label{eq:appA-spinors}
  \lambda_{i\alpha}=\sqrt{2\omega_i}\binom{1}{z_i}\, ,
  \qquad
  \tilde\lambda_{i\dot\alpha}=\eta_i\sqrt{2\omega_i}\binom{1}{\bar z_i}\, .
\end{equation}
Then \(p_{i\alpha\dot\alpha}=\lambda_{i\alpha}\tilde\lambda_{i\dot\alpha}\).  We take \(\epsilon^{12}=1\) and define
\begin{equation}
\label{eq:appA-brackets}
  \langle ij\rangle=\epsilon^{\alpha\beta}\lambda_{i\alpha}\lambda_{j\beta}
  =2\sqrt{\omega_i\omega_j}\,(z_j-z_i)\, ,
  \qquad
  [ij]=\epsilon^{\dot\alpha\dot\beta}\tilde\lambda_{i\dot\alpha}\tilde\lambda_{j\dot\beta}
  =2\eta_i\eta_j\sqrt{\omega_i\omega_j}\,(\bar z_j-\bar z_i)\, .
\end{equation}
\begin{equation}
\label{eq:appA-dot-bracket}
  2p_i\cdot p_j=-\langle ij\rangle[ij] \,,
  \qquad
  p_i\cdot p_j=-2\eta_i\eta_j\omega_i\omega_j|z_i-z_j|^2\, .
\end{equation}
The invariant used in the hard infrared factor is therefore
\begin{equation}
\label{eq:appA-sij}
  s_{ij}=-2p_i\cdot p_j-i0
  =4\eta_i\eta_j\omega_i\omega_j|z_i-z_j|^2-i0\, .
\end{equation}
The sign of the infinitesimal imaginary part is fixed by the all-outgoing amplitude convention.

A positive-helicity polarization vector with reference spinor \(r_\alpha\) is
\begin{equation}
\label{eq:appA-positive-polarization-spinor}
  \varepsilon_{+\,\alpha\dot\alpha}(s;r)
  =\frac{\sqrt2\,r_\alpha\tilde\lambda_{s\dot\alpha}}{\langle r s\rangle}\, .
\end{equation}
Changing \(r\) shifts \(\varepsilon_+\) by a multiple of the soft momentum, and hence leaves the ordered-pair kernel in appendix~\ref{app:one-loop-soft-photon-kernels} invariant.  In a local patch one may choose a reference spinor so that
\begin{equation}
\label{eq:appA-positive-polarization-vector}
  \varepsilon_+^\mu(z,\bar z)
  =\frac{1}{\sqrt2}\bigl(\bar z,\,1,\,-i,\,-\bar z\bigr)\, ,
  \qquad
  \varepsilon_+\cdot q(z,\bar z)=0\, .
\end{equation}
This representative is useful for deriving the celestial-coordinate form of the soft kernel.  The current algebra itself is gauge-invariant and does not depend on it.

Under \(g=\left(\begin{smallmatrix}a&b\\ c&d\end{smallmatrix}\right)\in SL(2,\mathbb C)\), the celestial coordinate transforms as
\begin{equation}
\label{eq:appA-mobius}
  z'=\frac{az+b}{cz+d}\, ,
  \qquad
  \bar z'=\frac{\bar a\bar z+\bar b}{\bar c\bar z+\bar d}\, ,
\end{equation}
and the null vector transforms as
\begin{equation}
\label{eq:appA-q-transform}
  \Lambda(g)^\mu{}_{\nu}q^\nu(z,\bar z)=|cz+d|^2q^\mu(z',\bar z')\, .
\end{equation}
A fixed momentum may be re-expanded either before or after the Lorentz transformation, with the inverse rescaling absorbed into the Mellin integration variable.  The two descriptions give the same primary transformation law.

The Mellin transform of a massless plane-wave operator is
\begin{equation}
\label{eq:appA-celestial-operator}
  \mathcal O_{\Delta,J}^{Q,\eta}(z,\bar z)
  =\int_0^\infty d\omega\,\omega^{\Delta-1}
  a_J^{Q,\eta}(\omega,z,\bar z)\, .
\end{equation}
For principal-series states \(\Delta=1+i\lambda\), \(\lambda\in\mathbb R\).  The two-dimensional weights are
\begin{equation}
\label{eq:appA-celestial-weights}
  h=\frac{\Delta+J}{2}\, ,
  \qquad
  \bar h=\frac{\Delta-J}{2}\, .
\end{equation}
The finite transformation law is
\begin{equation}
\label{eq:appA-primary-transformation}
  \mathcal O_{\Delta,J}^{Q,\eta}(z,\bar z)
  \longmapsto
  (cz+d)^{\Delta+J}(\bar c\bar z+\bar d)^{\Delta-J}
  \mathcal O_{\Delta,J}^{Q,\eta}(z',\bar z')\, .
\end{equation}
We use the convention in the conformal basis of celestial amplitudes \cite{Pasterski:2016qvg,Pasterski:2017kqt,Pasterski:2017kqt2,Pasterski:2021rjz}.

\subsection{Global conformal generators and the antiholomorphic spin doublet}
\label{appA:global-generators}

For a primary of weights \((h,\bar h)\), the infinitesimal generators acting on the insertion point are taken to be
\begin{equation}
\label{eq:appA-global-generators}
  \mathcal L_n^{(h)}=-z^{n+1}\partial_z-(n+1)h z^n\, ,
  \qquad
  \bar{\mathcal L}_n^{(\bar h)}=-\bar z^{n+1}\partial_{\bar z}-(n+1)\bar h\bar z^n\, ,
  \qquad n=-1,0,1\, .
\end{equation}
They satisfy
\begin{equation}
\label{eq:appA-sl2-commutator}
  [\mathcal L_m,\mathcal L_n]=(m-n)\mathcal L_{m+n}\, ,
  \qquad
  [\bar{\mathcal L}_m,\bar{\mathcal L}_n]=(m-n)\bar{\mathcal L}_{m+n}\, ,
  \qquad
  [\mathcal L_m,\bar{\mathcal L}_n]=0\, .
\end{equation}
The sign convention matches the finite transformation \eqref{eq:appA-primary-transformation}: the transformed field at the transformed point is obtained by exponentiating the vector fields \(-z^{n+1}\partial_z\) and \(-\bar z^{n+1}\partial_{\bar z}\), together with the corresponding weight terms.

The Banerjee--Mandal--Sahoo dipole currents reviewed in section~\ref{subsec:antiholo-dipole-current} carry an additional fundamental index of the antiholomorphic real form \cite{Banerjee:2026soft}.  Following their spinor conventions, we write
\begin{equation}
\label{eq:appA-Zbar-spinor}
  \bar Z^\alpha(\bar z)=\binom{\bar z}{1}\, ,
  \qquad
  \epsilon_{\alpha\beta}=\begin{pmatrix}0&1\\ -1&0\end{pmatrix}\, ,
  \qquad
  \bar Z_\alpha(\bar z)=\epsilon_{\alpha\beta}\bar Z^\beta(\bar z)=\binom{1}{-\bar z}\, .
\end{equation}
For \(g=\left(\begin{smallmatrix}a&b\\ c&d\end{smallmatrix}\right)\in SL(2,\mathbb R)_R\), acting by
\begin{equation}
\label{eq:appA-SL2R-action}
  \bar z'=\frac{a\bar z+b}{c\bar z+d}\, ,
\end{equation}
we have
\begin{equation}
\label{eq:appA-Z-transform}
  \bar Z^\alpha(\bar z')=(c\bar z+d)^{-1}g^\alpha{}_{\beta}\bar Z^\beta(\bar z)\, ,
  \qquad
  \bar Z_\alpha(\bar z')=(c\bar z+d)^{-1}(g^{-1})^\beta{}_{\alpha}\bar Z_\beta(\bar z)\, .
\end{equation}
The basic dipole Ward kernel of the Banerjee--Mandal--Sahoo Ward identity, eqs.~(4.3) and (6.3) of \cite{Banerjee:2026soft}, is
\begin{equation}
\label{eq:appA-dipole-kernel}
  \frac{\bar Z_{i\alpha}}{\bar z-\bar z_i}\, .
\end{equation}
Using \eqref{eq:appA-Z-transform} and
\begin{equation}
\label{eq:appA-cauchy-ratio-transform}
  \bar z'-\bar z_i'=\frac{\bar z-\bar z_i}{(c\bar z+d)(c\bar z_i+d)}\, ,
\end{equation}
one finds that covariance of the Ward identity is equivalent to
\begin{equation}
\label{eq:appA-D-alpha-transform}
  D_\alpha^{(\mu)}(\bar z)
  \longmapsto
  \frac{1}{c\bar z+d}\,
  g^\beta{}_{\alpha}D_\beta^{(\mu)}(\bar z')\, .
\end{equation}
\(D_\alpha\) is a weight-one-half antiholomorphic field valued in the fundamental two-dimensional spin representation, rather than an ordinary antiholomorphic current of weight one.  Accordingly, the invariant contraction of two dipole currents in section~\ref{sec:dipole-dipole-ope} has the pole order of a product of two fields of weight one-half, while the symmetric tensor channel produces the spin-one current checked in section~\ref{subsec:section7-representation-outlook}.

For a constant spinor \(\bar W^\alpha=(\bar W,1)^t\), the contracted kernel is
\begin{equation}
\label{eq:appA-contracted-kernel}
  \bar W^\alpha\frac{\bar Z_{i\alpha}}{\bar z-\bar z_i}
  =\frac{\bar W-\bar z_i}{\bar z-\bar z_i}\, .
\end{equation}
This expression separates the monopole charge distribution from the antiholomorphic first moment, as in the Banerjee--Mandal--Sahoo interpretation of the current.  The coefficient of \(\bar W\) is the monopole component \(\sum_i e_i/(\bar z-\bar z_i)\).  The constant term is the position-weighted component \(-\sum_i e_i\bar z_i/(\bar z-\bar z_i)\).  At large \(\bar z\), charge conservation removes the total monopole charge and leaves the dipole moment \cite{Banerjee:2026soft}.

\subsection{Higher logarithms in Mellin space}
\label{appA:logarithmic-mellin-transforms}

For comparison with the higher logarithmic tower, the same calculation gives
\begin{equation}
  \int_0^\infty d\omega\,\omega^{s-1}
  \left(\log\frac{\omega}{\mu}\right)^m f(\omega)
  =\bigl(\partial_s-\log\mu\bigr)^mM_f(s)
  =\frac{(-1)^m m!f(0)}{s^{m+1}}+O(s^{-m}) \, .
  \label{eq:appA-higher-log-transform}
\end{equation}
For a soft factor of homogeneity \(\omega^r\), replace \(s\) by \(\Delta+r\); the pole is centered at \(\Delta=-r\).  The Mellin derivative acts on celestial operators according to
\begin{equation}
  \partial_\Delta\mathcal O_{\Delta,J}^{Q,\eta}(z,\bar z)
  =\int_0^\infty d\omega\,\omega^{\Delta-1}\log\omega\,
  a_J^{Q,\eta}(\omega,z,\bar z) \, .
  \label{eq:appA-delta-derivative}
\end{equation}
Higher powers of the logarithm produce higher Jordan partners.  Only the first logarithmic layer enters the analysis.

\subsection{Finite parts, residues and scale changes}
\label{appA:finite-parts-scale-changes}

Let \(F(s)\) be meromorphic at \(s=0\), with Laurent expansion
\begin{equation}
\label{eq:appA-Laurent}
  F(s)=\frac{A_{-2}}{s^2}+\frac{A_{-1}}{s}+A_0+O(s)\, .
\end{equation}
We use
\begin{equation}
\label{eq:appA-residue-definitions}
  \operatorname{DRes}_{s=0}F=A_{-2}\, ,
  \qquad
  \operatorname{Res}_{s=0}F=A_{-1}\, ,
  \qquad
  \operatorname{FP}_{s=0}F=A_0\, .
\end{equation}
If the conformally soft point is \(\Delta=\Delta_\ast\), then \(s\) is replaced by \(\Delta-\Delta_\ast\).  These operations are applied after infrared subtraction of the hard amplitude.  They do not act on the unrenormalized S-matrix.

The renormalized plus-distributions make scale dependence transparent.  For a smooth cutoff \(\chi\) equal to one near \(\omega=0\), define
\begin{equation}
\label{eq:appA-plus-one-over-omega}
  \left\langle\left[\frac{1}{\omega}\right]_+^{\chi},f\right\rangle
  =\int_0^\infty \frac{d\omega}{\omega}\,\bigl(f(\omega)-\chi(\omega)f(0)\bigr)\, .
\end{equation}
Similarly,
\begin{equation}
\label{eq:appA-plus-log-over-omega}
  \left\langle\left[\frac{\log(\omega/\mu)}{\omega}\right]_+^{\chi},f\right\rangle
  =\int_0^\infty \frac{d\omega}{\omega}\,\log\frac{\omega}{\mu}
  \bigl(f(\omega)-\chi(\omega)f(0)\bigr)\, .
\end{equation}
Changing \(\chi\) changes these distributions by a multiple of \(\delta(\omega)\).  In celestial language this is a local contact counterterm at the conformally soft insertion.  The dependence on \(\mu\) is independent of \(\chi\):
\begin{equation}
\label{eq:appA-plus-scale-change}
  \left[\frac{\log(\omega/\mu')}{\omega}\right]_+^{\chi}
  =\left[\frac{\log(\omega/\mu)}{\omega}\right]_+^{\chi}
  -\log\frac{\mu'}{\mu}\left[\frac{1}{\omega}\right]_+^{\chi}\, .
\end{equation}
If \(S_\alpha\) denotes the ordinary conformally soft partner and \(D_\alpha^{(\mu)}\) the logarithmic representative normalized by \eqref{eq:appA-plus-log-over-omega}, then
\begin{equation}
\label{eq:appA-current-scale-change}
  D_\alpha^{(\mu')}=D_\alpha^{(\mu)}-\log\frac{\mu'}{\mu}\,S_\alpha\, .
\end{equation}
The triangular transformation accounts for the scheme dependence of the logarithmic current module at the order considered here.  Since \(S_\alpha\) acts in the one-particle layer, the dipole-hard class in \(F_2/F_1\) is independent of \(\mu\).  Appendix~\ref{app:hochschild-jacobi} formulates the corresponding scheme independence of the extension class.

The finite part of a logarithmic product is defined by subtracting the pole part in the same Laurent expansion.  For two logarithmic currents this gives the normal ordering used in section~\ref{sec:dipole-dipole-ope}.  Different finite subtractions differ by local one-particle currents and central contact terms; lemma~\ref{lem:one-particle-scheme-invariance} shows that the second-support class is unchanged.

\subsection{Cauchy kernels and contact terms on the celestial sphere}
\label{appA:cauchy-contact-terms}

We use
\begin{equation}
\label{eq:appA-area-measure}
  d^2z=\frac{i}{2}dz\wedge d\bar z\, ,
  \qquad
  \int_{\mathbb C}d^2z\,\delta^{(2)}(z-w)f(z,\bar z)=f(w,\bar w)\, .
\end{equation}
With this normalization,
\begin{equation}
\label{eq:appA-cauchy-distribution}
  \partial_{\bar z}\frac{1}{z-w}=\pi\delta^{(2)}(z-w)\, ,
  \qquad
  \partial_z\frac{1}{\bar z-\bar w}=\pi\delta^{(2)}(z-w)\, .
\end{equation}
The logarithmic form is
\begin{equation}
\label{eq:appA-log-laplacian}
  \partial_z\partial_{\bar z}\log|z-w|^2=\pi\delta^{(2)}(z-w)\, .
\end{equation}
Higher Cauchy singularities are defined by differentiation,
\begin{equation}
\label{eq:appA-higher-cauchy}
  \frac{1}{(z-w)^{m+1}}
  =\frac{(-1)^m}{m!}\partial_w^m\frac{1}{z-w}\, ,
  \qquad
  m\ge0\, .
\end{equation}
These identities are understood as distributional finite parts.  The same convention applies to antiholomorphic Cauchy kernels.

On \(\mathbb{CP}^1\) the formulas are used patchwise.  Changing patches replaces a Cauchy kernel by a transformed Cauchy kernel plus a smooth function.  Smooth terms and contact terms supported at a single insertion belong to \(F_1\).  The ordered-pair kernel also has a conormal singularity on a pair diagonal; proposition~\ref{prop:appC-conormal-separation} and the continued Plancherel transform identify its second-support class.

Equation~\eqref{eq:appA-cauchy-distribution} implies that a local counterterm obtained by differentiating a Cauchy pole at one hard insertion may change the ordinary soft-current action or a one-particle hard-current representative.  It cannot remove a pole at \(\bar z_a=\bar z_b\) for \(a\ne b\).  This is the distributional core of the non-absorption theorem in section~\ref{sec:dipole-hard-ope} and appendix~\ref{app:two-particle-projectors-support}.

\section{One-loop logarithmic soft-photon kernels}
\label{app:one-loop-soft-photon-kernels}

This appendix rewrites the known one-loop logarithmic soft theorem in the infrared-subtracted, normalized celestial form used in the main text.  Its leading high-energy term is local in the holomorphic variable of the emitting leg and has ordered-pair antiholomorphic support.

For leg \(a\), let \(\eta_a=\pm1\), \(Q_a\) be the physical charge, and \(e_a=\eta_aQ_a\).  The hard momenta are first kept slightly massive in order to regulate collinear regions.  The high-energy limit is taken only after the logarithmic coefficient has been extracted.  The hierarchy is
\begin{equation}
\label{eq:appB-hierarchy}
  \omega_s\ll m\ll E\, ,
\end{equation}
where \(\omega_s\) is the soft photon energy, \(m\) is the common mass scale of the charged external particles, and \(E\) is a hard energy scale.  In this regime the leading logarithmic coefficient is independent of the subleading power corrections in \(m/E\).  The massive derivation of the logarithmic theorem is due to Sahoo and Sen, and the all-order classical electromagnetic soft theorem gives a complementary classical organization of the late-time expansion \cite{Sahoo:2018lxl,Karan:2025emsoft,Choi:2024superphaserotation,Compere:2025multipole}.  The all-loop high-energy celestial interpretation and the dipole-current rewriting follow \cite{Campiglia:2019wxe,Krishna:2023fxg,Banerjee:2026soft}.

\subsection{Infrared-subtracted hard functions}
\label{appB:IR-subtracted-hard-functions}

The logarithmic soft theorem is an identity for infrared-subtracted hard functions.  We write the bare amplitude with \(n\) charged hard particles as
\begin{equation}
\label{eq:appB-IR-factorization}
  \mathcal A_n^{\rm bare}
  =Z_n^{\rm IR}(\epsilon_{\rm IR},\mu;\{p_i,Q_i\})\,
  \mathcal H_n^{\rm ren}(\mu;\{p_i,Q_i\})\, .
\end{equation}
Here \(Z_n^{\rm IR}\) contains the universal soft and collinear singularities, while \(\mathcal H_n^{\rm ren}\) is finite after the subtraction.  The precise regulator is immaterial for the algebraic statements of the paper, provided the same subtraction is used before and after insertion of the soft photon.  This convention is the abelian analogue of the hard-function factorization familiar from perturbative gauge theory; it is also the convention in which loop-corrected logarithmic soft theorems have a finite meaning \cite{Sahoo:2018lxl,Krishna:2023fxg}.

Let \(k=\omega_s q(w,\bar w)\) be the soft photon momentum.  The soft theorem is obtained from the ratio of the \((n+1)\)-point and \(n\)-point hard functions,
\begin{equation}
\label{eq:appB-hard-ratio}
  \mathcal H_{n+1}^{\rm ren}(k^+;1,\ldots,n)
  =\left[ S_{\rm W}^{+}(k)+\log\frac{\mu}{\omega_s}\,S_{\log}^{+}(k)+O(\omega_s^0)\right]
  \mathcal H_n^{\rm ren}(1,\ldots,n)\, .
\end{equation}
The Weinberg factor \(S_{\rm W}^{+}\) is a one-particle sum, whereas \(S_{\log}^{+}\) is the coefficient relevant here.  The form \(\log(\mu/\omega_s)\) keeps the subtraction scale explicit.  Changing \(\mu\) shifts the finite representative by an ordinary conformally soft current but leaves its class in \(F_2/F_1\) unchanged.

The celestial hard correlator is the Mellin transform of \(\mathcal H_n^{\rm ren}\),
\begin{equation}
\label{eq:appB-celestial-hard-transform}
  C_n(1,\ldots,n)
  =\int_0^\infty\prod_{i=1}^n d\omega_i\,\omega_i^{\Delta_i-1}
  \mathcal H_n^{\rm ren}(\eta_i\omega_i q_i,Q_i,J_i)\, .
\end{equation}
When a soft factor contains \(1/\omega_a\), the Mellin integral converts it into the shift \(T_a^{-1}\) defined in eq.~\eqref{eq:T-minus-one-section3}.  No other Mellin shift occurs in the leading high-energy logarithmic kernel.

\subsection{The one-loop logarithmic soft theorem}
\label{appB:one-loop-log-soft-theorem}

The normalized logarithmic insertion used in the main text is obtained by removing the common perturbative coefficient and fixing the remaining numerical factor by eq.~\eqref{eq:S0-Ward-pairwise}.  Set $q_s=q(w,\bar w)$, so that $k=\omega_s q_s$.  In this convention the leading high-energy ordered-pair coefficient is
\begin{equation}
\label{eq:appB-covariant-pairwise-factor}
  S_{\log,\rm pair}^{+}(k)
  =\sqrt2\sum_{a\ne b}Q_a^2 e_b\,
  \frac{(\varepsilon_+(k)\cdot p_a)(q_s\cdot p_b)-(\varepsilon_+(k)\cdot p_b)(q_s\cdot p_a)}
  {(q_s\cdot p_a)(p_a\cdot p_b)}\, .
\end{equation}
Here $q_s$ is the null direction of the soft photon and $\varepsilon_+(k)$ is the positive-helicity polarization vector.  Equation~\eqref{eq:appB-covariant-pairwise-factor} is a momentum-space factor; the Mellin shift \(T_a^{-1}\) appears only after the factor \(1/\omega_a\) is transformed.  No additional emitting-leg sign multiplies \(Q_a^2e_b\): the orientation sign is already contained in \(p_a=\eta_a\omega_a q_a\) and reappears as \(\eta_a/\omega_a\) in lemma~\ref{lem:appB-spinor-to-celestial}.  The full coefficient decomposes as
\begin{equation}
\label{eq:appB-soft-factor-decomposition}
  S_{\log}^{+}(k)
  =S_{\log,\rm pair}^{+}(k)+S_{\log,\rm self}^{+}(k)+S_{\log,\rm sub}^{+}(k)\, .
\end{equation}
The self term has no hard--hard Cauchy pole and belongs to $F_1$.  All statements about $F_2/F_1$ in this paper refer to the leading high-energy logarithmic coefficient $S_{\log,\rm pair}^{+}$.  The $m/E$-suppressed remainder lies outside this truncation, and no assertion about its second-support class is needed.

\subsection{Spinor-helicity and invariant forms}
\label{appB:spinor-helicity-form}

In the massless high-energy limit, use spinors
\begin{equation}
\label{eq:appB-spinors}
  p_{a\alpha\dot\alpha}=\lambda_{a\alpha}\tilde\lambda_{a\dot\alpha}\, ,
  \qquad
  k_{\alpha\dot\alpha}=\lambda_{s\alpha}\tilde\lambda_{s\dot\alpha}\, ,
\end{equation}
with the convention of eq.~\eqref{eq:spinors}.  A positive-helicity polarization with reference spinor \(r\) is
\begin{equation}
\label{eq:appB-positive-polarization-spinor}
  \varepsilon_{+\,\alpha\dot\alpha}(s;r)
  =\frac{\sqrt2\,r_\alpha\tilde\lambda_{s\dot\alpha}}{\langle r s\rangle}\, .
\end{equation}
Then the ordered-pair factor can be written as
\begin{equation}
\label{eq:appB-spinor-pairwise-kernel}
  S_{\log,\rm pair}^{+}(s)
  =\sum_{a\ne b}Q_a^2 e_b\,
  \widehat{\mathcal I}_{ab}^{+}(s;r)\, .
\end{equation}
where the normalized invariant ratio is
\begin{equation}
\label{eq:appB-Iab-spinor}
  \widehat{\mathcal I}_{ab}^{+}(s;r)
  =\sqrt2\,
  \frac{(\varepsilon_+(s;r)\cdot p_a)(k\cdot p_b)-(\varepsilon_+(s;r)\cdot p_b)(k\cdot p_a)}
  {(k\cdot p_a)(p_a\cdot p_b)}\, .
\end{equation}
The reference-spinor dependence is a gauge artifact.  If \(r\) is changed, \(\varepsilon_+(s;r)\) changes by a multiple of \(k\).  The numerator in \eqref{eq:appB-Iab-spinor} changes by
\begin{equation}
\label{eq:appB-gauge-change}
  (k\cdot p_a)(k\cdot p_b)-(k\cdot p_b)(k\cdot p_a)=0\, .
\end{equation}
\(\widehat{\mathcal I}_{ab}^{+}\) is gauge invariant and defines an unambiguous celestial distribution.

The denominator in eq.~\eqref{eq:appB-Iab-spinor} contains \(p_a\cdot p_b\) rather than \(k\cdot p_b\).  The soft pole is attached to the emitting leg \(a\), while the second leg enters through the relative direction between \(a\) and \(b\).  This invariant structure produces the antiholomorphic pole \((\bar z_a-\bar z_b)^{-1}\) below.

\begin{lemma}
\label{lem:appB-spinor-to-celestial}
With the null vectors and polarization of eqs.~\eqref{eq:q-vector}, \eqref{eq:p-parametrization} and \eqref{eq:positive-helicity-polarization}, the normalized invariant kernel \eqref{eq:appB-Iab-spinor} reduces to
\begin{equation}
\label{eq:appB-celestial-reduction-lemma}
  \widehat{\mathcal I}_{ab}^{+}(s)
  =\frac{\eta_a}{\omega_a}\,
  \frac{\bar w-\bar z_b}{(w-z_a)(\bar z_a-\bar z_b)}
\end{equation}
in the normalization of eq.~\eqref{eq:appB-covariant-pairwise-factor}.
\end{lemma}

\begin{proof}
Using eq.~\eqref{eq:positive-helicity-polarization}, one has
\begin{equation}
\label{eq:appB-polarization-dot-q}
  \varepsilon_+(w,\bar w)\cdot q(z_i,\bar z_i)
  =\sqrt2(\bar z_i-\bar w)
\end{equation}
with our mostly-plus metric.  The null inner products are
\begin{equation}
\label{eq:appB-null-products}
  q(w,\bar w)\cdot q(z_i,\bar z_i)=-2|w-z_i|^2\, ,
  \qquad
  q(z_a,\bar z_a)\cdot q(z_b,\bar z_b)=-2|z_a-z_b|^2\, .
\end{equation}
Substituting \(p_i=\eta_i\omega_i q_i\) gives
\begin{align}
\label{eq:appB-numerator-reduction}
  &(\varepsilon_+\cdot p_a)(q_s\cdot p_b)-(\varepsilon_+\cdot p_b)(q_s\cdot p_a)
  \nonumber\\
  &\quad
  =-2\sqrt2\,\eta_a\eta_b\omega_a\omega_b
  (\bar w-\bar z_a)(\bar w-\bar z_b)(z_b-z_a)
\end{align}
where \(q_s=q(w,\bar w)\).  The denominator is
\begin{equation}
\label{eq:appB-denominator-reduction}
  (q_s\cdot p_a)(p_a\cdot p_b)
  =4\eta_b\omega_a^2\omega_b
  (w-z_a)(\bar w-\bar z_a)(z_a-z_b)(\bar z_a-\bar z_b)\, .
\end{equation}
The quotient of \eqref{eq:appB-numerator-reduction} by \eqref{eq:appB-denominator-reduction} is \(2^{-1/2}\eta_a\omega_a^{-1}(\bar w-\bar z_b)/[(w-z_a)(\bar z_a-\bar z_b)]\).  Multiplication by the factor \(\sqrt2\) in \eqref{eq:appB-Iab-spinor} gives \eqref{eq:appB-celestial-reduction-lemma}.
\end{proof}

The factor \(1/\omega_a\) in \eqref{eq:appB-celestial-reduction-lemma} becomes the Mellin shift \(T_a^{-1}\).  Since \(e_a^2=Q_a^2\), the resulting coefficient is \(\eta_a e_a^2e_b\), in agreement with the ordered-pair dipole Ward kernel of \cite{Banerjee:2026soft}.  The coordinate-space ordered-pair operator is
\begin{equation}
\label{eq:appB-coordinate-pairwise-operator}
  S_{\log,\rm pair}^{+}(w,\bar w)
  =\sum_{a\ne b}
  \frac{\eta_a Q_a^2 e_b}{w-z_a}
  \frac{\bar w-\bar z_b}{\bar z_a-\bar z_b}
  T_a^{-1}\, .
\end{equation}

\subsection{Celestial-coordinate kernel}
\label{appB:celestial-coordinate-kernel}

The celestial kernel has two singular factors of different origin,
\begin{equation}
\label{eq:appB-two-poles}
  \frac{1}{w-z_a}
  \qquad\hbox{and}\qquad
  \frac{1}{\bar z_a-\bar z_b}\, .
\end{equation}
The first feature is the holomorphic soft pole, of the same type that appears in the tree-level conformally soft photon current.  The second feature is the relative antiholomorphic direction of the charged pair, giving a pole on the pairwise celestial diagonal of two hard insertions rather than in the soft coordinate.

Define the diagonal dipole multiplier
\begin{equation}
\label{eq:appB-dipole-multiplier-ab}
  d_b[\bar w](\bar z_a)
  = e_b\frac{\bar w-\bar z_b}{\bar z_a-\bar z_b}\, .
\end{equation}
The separation of the emitting-leg soft pole from this multiplier gives
\begin{equation}
\label{eq:appB-kernel-as-dipole-sum}
  S_{\log,\rm pair}^{+}(w,\bar w)
  =\sum_a\frac{\eta_a Q_a^2}{w-z_a}
  \left(\sum_{b\ne a} d_b[\bar w](\bar z_a)\right)T_a^{-1}\, .
\end{equation}
The inner sum is the finite part of the dipole-current Ward identity at \(\bar z=\bar z_a\).  The contracted dipole current obeys
\begin{equation}
\label{eq:appB-dipole-ward-recall}
  \left\langle D^{(\mu)}[\bar w](\bar z)\prod_i\mathcal O_i\right\rangle
  =\sum_i e_i\frac{\bar w-\bar z_i}{\bar z-\bar z_i}C_n\, .
\end{equation}
Taking \(\bar z\to\bar z_a\) and subtracting the self-pole produces
\begin{equation}
\label{eq:appB-finite-dipole-limit}
  {\rm FP}_{\bar z\to\bar z_a}
  \left[ D^{(\mu)}[\bar w](\bar z)
  T_a^{-1}\mathcal O_a(z_a,\bar z_a)\right]
  =\sum_{b\ne a}e_b\frac{\bar w-\bar z_b}{\bar z_a-\bar z_b}
  T_a^{-1}\mathcal O_a(z_a,\bar z_a)
\end{equation}
inside any hard correlator.  This term is the normal-ordered descendant \(\mathcal N_a^{(\mu)}[\bar w]\) of eq.~\eqref{eq:normal-ordered-dipole-descendant}.  The coordinate-space soft theorem is therefore equivalent to the local identity \eqref{eq:S0-local-dipole-Ward}.

The kernel \eqref{eq:appB-coordinate-pairwise-operator} is covariant under the antiholomorphic global conformal group.  Under \(\bar z\mapsto (a\bar z+b)/(c\bar z+d)\), the ratio
\begin{equation}
\label{eq:appB-ratio-covariance}
  \frac{\bar w-\bar z_b}{\bar z_a-\bar z_b}
\end{equation}
transforms by the same factor as the contracted dipole current in eq.~\eqref{eq:D-W-transform}.  Under the holomorphic group, \((w-z_a)^{-1}\) is the standard weight-one soft pole.  The shift \(T_a^{-1}\) changes the conformal dimension of the emitting field from \(\Delta_a\) to \(\Delta_a-1\), compensating the energy factor in the high-energy soft theorem.

After Mellin transformation of the hard energies, the covariant ordered-pair factor \eqref{eq:appB-covariant-pairwise-factor}, written in the Banerjee--Mandal--Sahoo normalization \cite{Banerjee:2026soft} and in the normalization used in the main text, gives the correlator identity
\begin{equation}
\label{eq:appB-coordinate-ward}
\begin{aligned}
  &\left\langle\mathsf S^0_\mu(w,\bar w)
  \prod_{i=1}^n\mathcal O_i\right\rangle_{\mathcal H,\mu}
  \\
  &\quad
  =\sum_{a\ne b}
  \frac{\eta_a Q_a^2 e_b}{w-z_a}
  \frac{\bar w-\bar z_b}{\bar z_a-\bar z_b}
  T_a^{-1}C_n(1,\ldots,n)\, .
\end{aligned}
\end{equation}
The coordinate reduction of the invariant soft factor is lemma~\ref{lem:appB-spinor-to-celestial}.  The factor \(1/\omega_a\) in that lemma becomes \(T_a^{-1}\) under the Mellin transform of leg \(a\), by the definition \eqref{eq:T-minus-one-section3}.  Multiplication by the charge factor \(Q_a^2e_b\) and summation over all ordered pairs gives \eqref{eq:appB-coordinate-ward}.  The self and subleading terms of \eqref{eq:appB-soft-factor-decomposition} are not included in the definition of \(\mathsf S^0_\mu\) used for the second-layer algebra; they are one-particle finite counterterms or power-suppressed corrections.

\subsection{Charge conservation and self terms}
\label{appB:charge-conservation-self-terms}

Self terms arise both from terms in the massive one-loop coefficient in which one charged particle carries all hard data and from bringing the dipole current to the emitting point $\bar z_a$.  In the local celestial expression the latter is the pole
\begin{equation}
\label{eq:appB-self-pole}
  e_a\frac{\bar w-\bar z_a}{\bar\xi-\bar z_a}
  T_a^{-1}\mathcal O_a(z_a,\bar z_a)
\end{equation}
subtracted in eq.~\eqref{eq:normal-ordered-dipole-descendant}.  Both contributions are local in one celestial insertion.  Consequently their kernels lie in \(F_1\).  They may affect the finite representative of the soft insertion, but not the class of the ordered-pair obstruction in \(F_2/F_1\).

Charge conservation acts differently.  Physical renormalized hard correlators obey
\begin{equation}
\label{eq:appB-charge-conservation}
  \sum_{i=1}^n e_i=0\, .
\end{equation}
This identity may be used to rewrite some constant pieces of the antiholomorphic numerator.  For example,
\begin{equation}
\label{eq:appB-charge-conservation-rewrite}
  \sum_{b\ne a}e_b\frac{\bar w-\bar z_b}{\bar z_a-\bar z_b}
  =\sum_{b\ne a}e_b\frac{\bar w-\bar z_a}{\bar z_a-\bar z_b}
  +\sum_{b\ne a}e_b\, .
\end{equation}
By \eqref{eq:appB-charge-conservation}, the last sum equals \(-e_a\) and contributes only to the one-particle sector.  The first term retains the pairwise pole \((\bar z_a-\bar z_b)^{-1}\).  Thus charge conservation can move representatives between the explicit pairwise kernel and the local self sector, but it cannot remove the residue on the pairwise diagonal unless all relevant charge products vanish.

The support-filtration argument is independent of the chosen representative.  Let \(R^{(1)}(w,\bar w)\in F_1\) be any one-particle counterterm and replace
\begin{equation}
\label{eq:appB-representative-shift}
  \mathsf S^0_\mu(w,\bar w)
  \longmapsto
  \mathsf S^0_\mu(w,\bar w)+R^{(1)}(w,\bar w)\, .
\end{equation}
Then the image in the associated graded is unchanged,
\begin{equation}
\label{eq:appB-associated-graded-invariance}
  \sigma_2(\mathsf S^0_\mu+R^{(1)})
  =\sigma_2(\mathsf S^0_\mu)\in F_2/F_1\, .
\end{equation}
Here \(\sigma_2\) denotes the projection to the second support layer.  Appendix~\ref{app:two-particle-projectors-support} gives the microlocal proof that this layer is represented by two-particle celestial primaries.  

\subsection{Review of the Banerjee--Mandal--Sahoo dipole Ward identity}
\label{appB:review-dipole-ward-identity}

This subsection rewrites the Banerjee--Mandal--Sahoo dipole Ward identity and its local normal-ordered form, eqs.~(4.3)--(4.7) of \cite{Banerjee:2026soft}, in the conventions of this appendix.  The contracted dipole current has the Ward identity
\begin{equation}
\label{eq:appB-dipole-current-definition}
  \left\langle D^{(\mu)}[\bar W](\bar z)
  \prod_i\mathcal O_i\right\rangle_{\mathcal H,
  \mu}
  =\sum_i e_i\frac{\bar W-\bar z_i}{\bar z-\bar z_i}C_n\, .
\end{equation}
Set \(\bar W=\bar w\) and take the finite part at the emitting point \(\bar z=\bar z_a\) after subtracting the self-pole of leg \(a\).  By definition,
\begin{equation}
\label{eq:appB-normal-ordered-finite-part}
\begin{aligned}
  &\left\langle
  :D^{(\mu)}[\bar w]T_a^{-1}\mathcal O_a:(z_a,\bar z_a)
  \prod_{i\ne a}\mathcal O_i
  \right\rangle_{\mathcal H,
  \mu}
  \\
  &\quad
  =\sum_{b\ne a}e_b
  \frac{\bar w-\bar z_b}{\bar z_a-\bar z_b}
  T_a^{-1}C_n\, .
\end{aligned}
\end{equation}
Multiplying by the holomorphic soft pole and by the emitting-leg charge factor gives
\begin{equation}
\label{eq:appB-local-S0-from-dipole}
\begin{aligned}
  &\sum_a\frac{\eta_aQ_a^2}{w-z_a}
  \left\langle
  :D^{(\mu)}[\bar w]T_a^{-1}\mathcal O_a:(z_a,\bar z_a)
  \prod_{i\ne a}\mathcal O_i
  \right\rangle_{\mathcal H,
  \mu}
  \\
  &\quad
  =\sum_{a\ne b}
  \frac{\eta_aQ_a^2 e_b}{w-z_a}
  \frac{\bar w-\bar z_b}{\bar z_a-\bar z_b}
  T_a^{-1}C_n\, .
\end{aligned}
\end{equation}
The right hand side is the ordered-pair identity \eqref{eq:appB-coordinate-ward}.  Hence
\begin{equation}
\label{eq:appB-S0-local-final}
  \mathsf S^0_\mu(w,\bar w)
  \sim
  \sum_a\frac{\eta_aQ_a^2}{w-z_a}
  :D^{(\mu)}[\bar w]T_a^{-1}\mathcal O_a:(z_a,\bar z_a)
\end{equation}
inside all renormalized hard correlators, which is the Banerjee--Mandal--Sahoo local identity \eqref{eq:S0-local-dipole-Ward}.

The same rewriting also displays the two-layer decomposition of the Ward identity.  The one-particle part is the local normal-ordering subtraction and the possible self term \(S_{\log,\rm self}^{+}\).  The second-layer part is
\begin{equation}
\label{eq:appB-second-layer-symbol}
  \sigma_2(\mathsf S^0_\mu)(w,\bar w)
  =\sum_{a\ne b}
  \eta_aQ_a^2e_b\,
  \frac{1}{w-z_a}
  \frac{\bar w-\bar z_b}{\bar z_a-\bar z_b}
  T_a^{-1}
  \quad \in F_2/F_1\, .
\end{equation}
Commutation with a Mellin-difference label produces the dipole-hard OPE, while composition of two second-support symbols produces the mixed kernel identity.

\section{Analytic two-particle projectors and support filtration}
\label{app:two-particle-projectors-support}

One-particle kernels and pairwise conormal singularities are separated by the support filtration.  Meromorphic continuation through the shift \(T^{-1}\) gives the two-particle Plancherel resolution used in the minimality theorem.

\subsection{Analytic test spaces and the support filtration}
\label{appC:configuration-spaces}

Let \(X=\mathbb{CP}^1\).  In an affine patch, celestial coordinates are paired with compactly supported smooth test functions.  For the Mellin variable, fix \(\varepsilon>0\) and define \(\mathscr W_\varepsilon\) to be the Fr\'echet space of functions holomorphic in
\begin{equation}
  -\varepsilon<\operatorname{Re}\Delta<1+\varepsilon
\end{equation}
whose restrictions to every closed vertical substrip are Schwartz in \(\operatorname{Im}\Delta\).  The seminorms are
\begin{equation}
  p_{N,k,I}(F)=
  \sup_{\operatorname{Re}\Delta\in I}
  (1+|\operatorname{Im}\Delta|)^N
  |\partial_\Delta^kF(\Delta)| \, ,
  \label{eq:appC-analytic-seminorms}
\end{equation}
with \(I\Subset(-\varepsilon,1+\varepsilon)\).  The shifts \(T^{\pm1}\) are continuous between the corresponding boundary-value spaces.

For a species \(\chi=(J,Q,\eta)\), set
\begin{equation}
  \mathcal V_\chi^{\rm an}
  =\mathscr W_\varepsilon\widehat\otimes C_c^\infty(X)\otimes V_\chi \, ,
  \qquad
  \mathcal V_{ab}^{\rm an}
  =\mathcal V_{\chi_a}^{\rm an}\widehat\otimes
   \mathcal V_{\chi_b}^{\rm an} \, .
  \label{eq:appC-one-particle-test-space}
\end{equation}
Their strong duals contain the distributional celestial operators and the analytically continued hard correlators.  Restriction to \(\operatorname{Re}\Delta=1\) recovers the ordinary principal-series wave packets.

Let \(\mathscr K_n\) be the algebra of distribution kernels generated by local current insertions, Mellin-difference operators, and the normalized logarithmic soft kernels.  A kernel belongs to \(F_1\mathscr K_n\) when, after smearing all current points, its singular support contains no pair diagonal.  The layer \(F_2\mathscr K_n\) allows one connected ordered pair of hard coordinates.  In local coordinates,
\begin{equation}
  \operatorname{sing\,supp}K\subset
  \bigcup_i\Delta_{xi}
  \quad(K\in F_1),
  \qquad
  \operatorname{sing\,supp}K\subset
  \bigcup_i\Delta_{xi}\cup\bigcup_{a\ne b}\Delta_{ab}
  \quad(K\in F_2) \, ,
  \label{eq:appC-support-filtration}
\end{equation}
with the second condition understood modulo \(F_1\) and with at most one connected hard pair in each summand.  Mellin shifts do not change celestial singular support.

For the rational kernels in the paper, the Cauchy residue characterizes the quotient \(F_2/F_1\)
\begin{equation}
  \operatorname{Res}_{ab}K
  =\frac{1}{2\pi i}\oint_{|\bar z_a-\bar z_b|=\delta}
  d(\bar z_a-\bar z_b)\,K \, .
  \label{eq:appC-pair-residue-definition}
\end{equation}
Every one-particle kernel has vanishing residue on every hard-pair diagonal,
\begin{equation}
  K\in F_1\quad\Longrightarrow\quad
  \operatorname{Res}_{ab}K=0\qquad(a\ne b) \, .
  \label{eq:appC-F1-zero-pair-residue}
\end{equation}
The logarithmic kernel
\begin{equation}
  K_{w,\bar w}^{ab}
  =\frac{\eta_aQ_a^2e_b}{w-z_a}
   \frac{\bar w-\bar z_b}{\bar z_a-\bar z_b}T_a^{-1}
  \label{eq:appC-basic-pair-kernel}
\end{equation}
has residue
\begin{equation}
  \operatorname{Res}_{ab}K_{w,\bar w}^{ab}
  =\frac{\eta_aQ_a^2e_b}{w-z_a}
  (\bar w-\bar z_b)T_a^{-1} \, .
  \label{eq:appC-basic-residue}
\end{equation}
It is nonzero for generic charged data.

\subsection{Conormal separation}
\label{appC:conormal-wavefront}

The distribution \((\bar z_a-\bar z_b)^{-1}\) is conormal to the real codimension-two diagonal \(\Delta_{ab}\) in the standard microlocal sense \cite{Hormander:1998},
\begin{equation}
  \operatorname{WF}\!\left((\bar z_a-\bar z_b)^{-1}\right)
  \subset N^*\Delta_{ab}\setminus0 \, .
  \label{eq:appC-cauchy-wavefront}
\end{equation}
A one-particle current kernel has wavefront set in the union of conormal bundles to the current-leg diagonals.  Smearing the current point cannot create a component in \(N^*\Delta_{ab}\).  Mellin shifts and analytic continuation in the external dimensions act in spectral variables and preserve this celestial wavefront statement.

\begin{proposition}[Conormal separation]
\label{prop:appC-conormal-separation}
Let \(K\in F_2\mathscr K_n\) be a finite sum of local one-particle kernels and pairwise Cauchy kernels.  If the projection of \(\operatorname{WF}(K)\) to \(N^*\Delta_{ab}\setminus0\) is nonempty for some \(a\ne b\), then \([K]\neq0\) in \(F_2/F_1\).  For the meromorphic kernels of the logarithmic soft theorem, the same conclusion follows from \(\operatorname{Res}_{ab}K\neq0\).
\end{proposition}

\begin{proof}
If \([K]=0\), then \(K\) differs from zero by an element of \(F_1\).  Such an element has no wavefront component conormal to a hard-pair diagonal, contradicting the hypothesis.  For a Cauchy kernel, the leading conormal symbol is its residue; hence a nonzero residue gives the stated wavefront component.
\end{proof}

Proposition~\ref{prop:appC-conormal-separation} gives the microlocal form of theorem~\ref{thm:nonabsorption-section4}; local counterterms and Mellin-subtraction changes remain in \(F_1\).

\subsection{Meromorphic Plancherel resolution}
\label{appC:plancherel-shadow}

\begin{assumption}[Meromorphic two-particle continuation]
\label{ass:appC-meromorphic-continuation}
There is an exhaustion by compact spectral strips \(K_1\subset K_2\subset\cdots\) such that: (i) the normalized two-particle projectors, their inverse transform, and the Knapp--Stein intertwiners extend meromorphically in the external and exchanged dimensions; (ii) each \(K_j\) meets only finitely many poles, all of finite rank; (iii) away from these poles the continued kernels act continuously on the nuclear Fr\'echet test spaces defined in \eqref{eq:appC-one-particle-test-space}; and (iv) restriction to smaller strips commutes with the projectors, residue maps, and shadow intertwiners.
\end{assumption}

The principal-line decomposition and the Knapp--Stein intertwiners are standard.  Existing results establish distributional continuation of celestial transforms, meromorphic partial waves, and shadow intertwiners on the relevant principal or pole-free domains \cite{Dobrev:1977qv,Atanasov:2021cje,Borji:2024distributional,Pacifico:2025cpw,Himwich:2025shadow}.  Assumption~\ref{ass:appC-meromorphic-continuation} asks for their compatibility under the shifted dipole continuation and the strip restriction maps.  It is the expected compatibility of these continuation maps in the dipole channel, and it is the only analytic hypothesis used in the global support--primary identification and the full minimality theorem.  All local residue statements and all stripwise coefficient formulas used in the main text are independent of this global inverse-limit assumption.

On the unitary principal line, diagonal \(SL(2,\mathbb C)\) harmonic analysis gives projectors \(\mathbb P_{ab}(\nu,\ell)\) satisfying
\begin{equation}
  \sum_\ell\int_{\mathbb R}d\nu\,
  \rho_\ell(\nu)\mathbb P_{ab}(\nu,\ell)
  =\mathbf1_{ab} \, .
  \label{eq:appC-projector-completeness}
\end{equation}
The intertwining kernels are meromorphic functions of the external dimensions and of the exchanged dimension.  Their distributional continuation is defined by pairing with \(\mathscr W_\varepsilon\) and deforming the spectral contour \cite{Borji:2024distributional,Pacifico:2025cpw,Himwich:2025shadow}.

Let \(\gamma_s\) be the path \(\Delta_a(s)=1-s+i\lambda_a\), \(0\leq s\leq1\).  Deforming the projectors along \(\gamma_s\) gives
\begin{equation}
\begin{aligned}
  \mathbf1_{ab}^{\rm an}={}&
  \sum_\ell\int_{\mathbb R}d\nu\,
  \rho_\ell(\nu)
  \mathbb P_{ab}^{\rm an}
  (\nu,\ell;\Delta_a-1,\Delta_b)
  \\
  &+\sum_{r\in\mathfrak R_{ab}}
  \mathbb P_{ab}^{\rm res}(r;\Delta_a-1,\Delta_b) \, .
  \label{eq:appC-analytic-completeness}
\end{aligned}
\end{equation}
The residue set \(\mathfrak R_{ab}\) consists of the poles crossed by the contour.  The formula is independent of the deformation path as long as no pole endpoint is crossed.  A different path changes the continuous and discrete pieces separately but not their sum.

The continued Knapp--Stein map identifies the two labels of one Plancherel channel,
\begin{equation}
  \mathcal S_{\nu,\ell}:\mathcal V_{1+i\nu,\ell}
  \longrightarrow \mathcal V_{1-i\nu,-\ell},
  \qquad (\nu,\ell)\sim(-\nu,-\ell) \, .
  \label{eq:appC-shadow-identification}
\end{equation}
This identification removes double counting in the partial-wave channel label.  It does not identify Mellin-basis and shadow-basis operators as boundary insertions; a shadow-completed OPE may retain both representatives while the partial-wave transform groups them into one channel \cite{Himwich:2025shadow,Liu:2026shadowcompletion}.
On every compact spectral strip away from poles, \(\mathcal S_{\nu,\ell}\) is a continuous map of nuclear Fr\'echet test spaces, hence its graph is closed.  Let \(\mathscr T_{ab}(K)\) be the analytic two-particle test space on a compact pole-free strip \(K\), let \(\mathcal N_{ab}^{\rm sh}(K)\) be the closed span of the graph relations \((v,-\mathcal S_{\nu,\ell}v)\) and the finite-dimensional intertwiner kernels on the boundary of \(K\), and set
\begin{equation}
  \mathscr Q_{ab}(K)
  =\mathscr T_{ab}(K)/\mathcal N_{ab}^{\rm sh}(K) .
  \label{eq:appC-stripwise-shadow-quotient}
\end{equation}
This quotient is complete and Hausdorff.  The residue spaces crossed while the contour is moved to \(K\) are adjoined as finite-dimensional summands.

\begin{proposition}[Stripwise continued Plancherel transform]
\label{prop:appC-stripwise-plancherel}
For each compact pole-free strip \(K\) satisfying assumption~\ref{ass:appC-meromorphic-continuation}, the continued Plancherel transform together with the crossed-pole residue maps is a topological isomorphism from the pairwise conormal quotient on \(K\) to \(\mathscr Q_{ab}(K)\) plus its finite-dimensional residue channels.
\end{proposition}

\begin{proof}
On the unitary line the normalized Plancherel transform and its inverse are continuous and their two compositions are the identity.  Assumption~\ref{ass:appC-meromorphic-continuation} makes both compositions meromorphic operator-valued functions on the strip.  Their difference from the identity vanishes on the unitary line and therefore vanishes throughout every connected pole-free component by the identity theorem.  Moving the contour across a pole adds its finite-rank residue projector, so the continued transform and its inverse remain mutually inverse after the residue channels are adjoined.  The graph of the continued Knapp--Stein map is closed, the quotient removes its kernel, and the residue summands account for the cokernel.  The contour formulas are continuous in the strip topology.
\end{proof}

For an exhaustion \(K_1\subset K_2\subset\cdots\), the restriction maps on the stripwise quotients and residue data form a compatible inverse system.  Define
\begin{equation}
  \mathcal M_{2,ab}^{\rm an}
  =\varprojlim_j\left(
    \mathscr Q_{ab}(K_j)
    \oplus
    \bigoplus_{r\in\mathfrak R_{ab}(K_j)}
    \mathcal M_{ab}^{\rm res}(r)
  \right).
  \label{eq:appC-projective-shadow-quotient}
\end{equation}
It is a closed subspace of a product of complete Hausdorff locally convex spaces and is therefore complete and Hausdorff.  This inverse-limit definition avoids an interchange of quotient and projective limit; it is canonically equivalent to the compatible shadow quotient used in the main text.  The shadow relation and crossed-pole channels are thus built into the analytic two-particle module.

\begin{proposition}[Analytic support--primary identification]
\label{prop:appC-F2-M2-identification}
For kernels generated by the logarithmic soft theorem and Mellin-difference hard currents,
\begin{equation}
  F_2/F_1\simeq\mathcal M_2 \, ,
  \label{eq:appC-F2-M2}
\end{equation}
where \(\mathcal M_2\) is the analytically continued module \eqref{eq:M2-definition}.  The map is the continuous Plancherel transform plus the crossed-pole residue projections in \eqref{eq:appC-analytic-completeness}.
\end{proposition}

\begin{proof}
On each compact pole-free strip, proposition~\ref{prop:appC-stripwise-plancherel} identifies the conormal quotient with the stripwise primary data and the residues of every pole crossed during the contour deformation.  Assumption~\ref{ass:appC-meromorphic-continuation}(iv) makes these identifications compatible with restriction, hence their inverse limit is a continuous bijection to \(\mathcal M_{2,ab}^{\rm an}\).  The inverse is continuous stripwise and therefore continuous for the initial inverse-limit topology.  If all continuous and residue coefficients vanish, every stripwise class is zero; by proposition~\ref{prop:appC-conormal-separation} the remaining representative is one-particle local.  Conversely, any nonzero continuous or residue coefficient defines a nonzero support class.  Summing over charged ordered pairs gives \eqref{eq:appC-F2-M2}.
\end{proof}

\subsection{Continuity of the generated family}
\label{appC:generation-M2}
The angular and Mellin density statements are proved in lemmas~\ref{lem:angular-cyclicity-section4} and \ref{lem:mellin-cyclicity-section4}.  On every compact strip, multiplication by a smooth nonvanishing soft profile, the Mellin shift, and the continued Plancherel transform are continuous.  The generated family therefore maps densely to each stripwise quotient of proposition~\ref{prop:appC-stripwise-plancherel}.  Lemma~\ref{lem:channel-annihilator-section4} then passes this density through the closed shadow quotient and the explicit residue summands, while compatibility of the restriction maps passes it to the inverse limit \eqref{eq:appC-projective-shadow-quotient}.  These continuity statements are the analytic-topological input used in theorem~\ref{thm:soft-accessible-generation-section4}.

\section{Filtered Hochschild and Chevalley--Eilenberg complexes}
\label{app:hochschild-jacobi}

Hochschild cochains govern the ordered celestial product, and antisymmetrization gives the Chevalley--Eilenberg complex of the current algebra \cite{Hochschild:1945,Gerstenhaber:1964,ChevalleyEilenberg:1948}.  The formal parameter records support-extension degree rather than physical loop order.

\subsection{The first support extension}
\label{appD:filtered-ope-complex}

Let \(\mathscr A_1\) be the one-particle ordered current algebra generated by \(\mathfrak h_{\rm hard}\) and \(\mathfrak s_{\log}\), with all local contact terms included.  Let \(E=\mathcal M_2\) be the analytic second-support module.  Work in
\begin{equation}
  \widehat{\mathscr A}
  =\mathscr A_1\oplus tE,
  \qquad t^2=0 \, ,
  \label{eq:appD-A1-definition}
\end{equation}
which is the ordered-product analogue of retaining \(F_1\oplus F_2/F_1\) and discarding \(F_3\).

Write
\begin{equation}
  m_t=m_0+t\,c \, ,
  \label{eq:appD-deformed-product}
\end{equation}
where \(m_0\) is the one-particle OPE product and \(c\in C^2(\mathscr A_1,E)\) is an ordered representative of the dipole-hard second-support term.  We choose it so that
\begin{equation}
  c(X,H[\Phi])-c(H[\Phi],X)
  =\mathfrak m_X(\Phi) \, ,
  \label{eq:appD-ordered-representative}
\end{equation}
while its hard-hard component and its antisymmetric soft-soft component vanish.  A different choice of the symmetric part of \(c\) changes the ordered product but not the current-algebra extension.

The first-order associator is the Hochschild differential
\begin{equation}
\begin{aligned}
  (d_{\rm Hoch}c)(a,b,d)={}&
  a\cdot c(b,d)-c(m_0(a,b),d)
  \\
  &+c(a,m_0(b,d))-c(a,b)\cdot d \, .
  \label{eq:appD-Hochschild-differential}
\end{aligned}
\end{equation}
Associativity modulo \(t^2\) is equivalent to
\begin{equation}
  d_{\rm Hoch}c=0 \, .
  \label{eq:appD-Hochschild-cocycle}
\end{equation}
For the celestial kernels, this equation is inherited from associative composition in the finite-part kernel algebra.  Every second-support term is represented in the internal target \(E\).

\subsection{Antisymmetrization}
\label{appD:hochschild-to-CE}

Let \(\mathfrak g_1\) be the commutator Lie algebra of \(\mathscr A_1\).  The antisymmetrization of \(c\),
\begin{equation}
  \nu(a,b)=c(a,b)-c(b,a) \, ,
  \label{eq:appD-nu-antisym}
\end{equation}
is the cochain \eqref{eq:nu-definition-section6}.  Its Chevalley--Eilenberg differential is
\begin{equation}
\begin{aligned}
  (d_{\rm CE}\nu)(x,y,z)={}&
  x\cdot\nu(y,z)-y\cdot\nu(x,z)+z\cdot\nu(x,y)
  \\
  &-\nu([x,y]_1,z)+\nu([x,z]_1,y)-\nu([y,z]_1,x) \, .
  \label{eq:appD-CE-differential}
\end{aligned}
\end{equation}

\begin{proposition}[Hochschild--Lie compatibility]
\label{prop:appD-antisym-Hochschild}
If \(d_{\rm Hoch}c=0\), then \(d_{\rm CE}\nu=0\).  Conversely, for the mixed component used here, \(d_{\rm CE}\nu=0\) is the antisymmetric part of the ordered associativity equation.
\end{proposition}

\begin{proof}
Expand the six ordered associators obtained by antisymmetrizing \((a,b,d)\).  The terms containing the ordered product combine into the three Lie brackets in \eqref{eq:appD-CE-differential}; the remaining terms give the module action.  Thus the antisymmetrization of \(d_{\rm Hoch}c\) is \(d_{\rm CE}\nu\).  In this extension the antisymmetric component is \(\nu(X,H[\Phi])=\mathfrak m_X(\Phi)\); the reverse implication on mixed triples follows directly.
\end{proof}

\subsection{The hard-current and mixed cocycle equations}
\label{appD:mixed-symbol}

On a triple \((X,H[\Phi],H[\Psi])\), the equation \(d_{\rm CE}\nu=0\) becomes
\begin{equation}
  \mathfrak m_X([\Phi,\Psi]_{\star})
  =\Phi\cdot\mathfrak m_X(\Psi)
   -\Psi\cdot\mathfrak m_X(\Phi) \, ,
  \label{eq:appD-hard-one-cocycle}
\end{equation}
which is theorem~\ref{thm:hard-current-one-cocycle}.  On a triple \((X,Y,H[\Phi])\), it becomes
\begin{equation}
\begin{aligned}
  0={}&\rho_2(X)\mathfrak m_Y(\Phi)
      -\rho_2(Y)\mathfrak m_X(\Phi)
      -\mathfrak m_{[X,Y]_1}(\Phi)
  \\
  &-\mathfrak m_Y(\mathcal L_X\Phi)
   +\mathfrak m_X(\mathcal L_Y\Phi) \, .
  \label{eq:appD-second-layer-Hochschild}
\end{aligned}
\end{equation}

The mixed identity in theorem~\ref{thm:second-layer-jacobi} reduces on every ordered pair \((a,b)\) to
\begin{equation}
\begin{aligned}
  &[K_X^{ab},[K_Y^{ab},\Phi_{ab}]]_{\rm fp}
  -[K_Y^{ab},[K_X^{ab},\Phi_{ab}]]_{\rm fp}
  \\
  &\hspace{3.8cm}
  -[[K_X^{ab},K_Y^{ab}]_{\rm fp},\Phi_{ab}]_{\rm fp}=0 \, .
  \label{eq:appD-mixed-symbol}
\end{aligned}
\end{equation}
For \(\xi=\bar z_a-\bar z_b\), replace a singular monomial \(\xi^{-r}\) by \(\xi^{-r+\varepsilon}\).  Products converge for \(\operatorname{Re}\varepsilon\) sufficiently large and continue meromorphically to \(\varepsilon=0\).  The finite-part associator is supported at \(\xi=0\), hence it is a one-particle local term and lies in \(F_1\).  Associativity is used only after projection to \(F_2/F_1\), where these local terms vanish.  Equation~\eqref{eq:appD-mixed-symbol} is therefore the ordinary commutator Jacobi identity in the quotient; different local extensions give the same class.

\subsection{Extension, minimality and scheme independence}
\label{appD:renormalized-extension}

The Lie bracket on \(\mathfrak g_1\oplus tE\) is
\begin{equation}
  [(x,tu),(y,tv)]
  =([x,y]_1,\,t(x\cdot v-y\cdot u+\nu(x,y))) \, .
  \label{eq:appD-ren-bracket}
\end{equation}

For the bracket \eqref{eq:appD-ren-bracket}, the \(\mathfrak g_1\)-component of the Jacobiator vanishes because \(\mathfrak g_1\) is a Lie algebra, while its \(tE\)-component is \((d_{\rm CE}\nu)(x,y,z)\).  The bracket is therefore Lie modulo \(t^2\) exactly when \(d_{\rm CE}\nu=0\).  For the celestial cochain \eqref{eq:nu-definition-section6}, this condition follows from \eqref{eq:appD-hard-one-cocycle} and the pairwise identity \eqref{eq:appD-mixed-symbol}.

The analytic step in theorem~\ref{thm:minimal-closure-section6} is the generation argument.  The pairwise residue criterion gives a nonzero class in \(F_2/F_1\) for every Mellin label with nonzero backward difference.  Theorem~\ref{thm:soft-accessible-generation-section4} identifies the closed orbit under the required profile, Mellin and diagonal Lorentz actions with \(\mathcal M_2\).  Since a one-particle redefinition lies in \(F_1\), it leaves this generated subquotient unchanged.

A scheme change modifies the logarithmic kernel by \(K_X\mapsto K_X+r_X\) with \(r_X\in F_1\), and a finite hard-current redefinition changes \(H[\Phi]\) by \(tH[R(\Phi)]\).  The induced change of the mixed cochain is
\begin{equation}
  \mathfrak m_X(\Phi)
  \longmapsto
  \mathfrak m_X(\Phi)+\sigma_2([r_X,\Phi])
  =\mathfrak m_X(\Phi) \, ,
  \label{eq:appD-scale-change-cocycle}
\end{equation}
while the change produced by \(R\) is a Chevalley--Eilenberg coboundary in the chosen representative.  Therefore
\begin{equation}
  [\nu]_{F_2/F_1}
  \quad\text{and}\quad
  \mathfrak g_{\rm ren}^{\log}
  \label{eq:appD-scale-independent}
\end{equation}
are independent of local one-particle subtraction schemes.  Local central terms in the dipole product remain scheme dependent and are confined to \(F_1\).

\section*{Acknowledgments}
I thank Shamik Banerjee, Raju Mandal and Biswajit Sahoo for correspondence, and Shamik Banerjee for pointing out that an earlier version did not attribute their dipole-current construction explicitly enough.

\section*{Data Availability}
No datasets were generated or analyzed in this work.

\bibliographystyle{JHEP}
\bibliography{biblio}

\end{document}